%% file: main.tex
\documentclass{article}
\usepackage{ijcai_arxiv}

\usepackage{times}
\usepackage{soul}
\usepackage{url}
\usepackage[hidelinks]{hyperref}
\usepackage[utf8]{inputenc}
\usepackage[small]{caption}
\usepackage{booktabs}
\usepackage[switch]{lineno}

\usepackage{algorithm}
\usepackage{algorithmicx}%
\usepackage{algpseudocode}%

\usepackage{comment}

\usepackage{graphicx}%
\usepackage{multirow}%
\usepackage{amsmath,amsfonts,amssymb}
\usepackage{amsthm}
\usepackage{mathrsfs}
\usepackage[title]{appendix}%
\usepackage{xcolor}%
\usepackage{textcomp}%
\usepackage{manyfoot}%
\usepackage{booktabs}%

\usepackage{graphicx}

\newtheorem{theorem}{Theorem}
\newtheorem{lemma}[theorem]{Lemma}

\theoremstyle{plain}

\theoremstyle{definition}
\newtheorem{definition}[theorem]{Definition}

\theoremstyle{remark}

\definecolor{darkcyan}{rgb}{0.0, 0.55, 0.55}
\newcommand{\obj}{\texttt{Obj}}
\newcommand{\opt}{\mathsf{OPT}}
\newcommand{\avg}{\mathsf{AVG}}

\renewcommand{\kappa}{\mathcal{K}}
\renewcommand{\epsilon}{\varepsilon}
\newcommand{\col}{\texttt{col}}
\newcommand{\cost}{\texttt{cost}}
\renewcommand{\S}{\mathcal{S}}
\newcommand{\dist}{\texttt{dist}}

\title{Improved Rank Aggregation under Fairness Constraint\footnote{This project received funding from an MoE AcRF Tier 2 grant (MOE-T2EP20221-0009), an MoE AcRF Tier 1 grant (T1 251RES2303), and a Google South \& South-East Asia Research Award.}}
\author{
Diptarka Chakraborty$^1$\and
Himika Das$^2$\and
Sanjana Dey$^3$\And
Alvin Hong Yao Yan$^1$
\\
\affiliations
$^1$National University of Singapore\\
$^2$TU Wien\\
$^3$UMONS\\
}

\begin{document}

\maketitle

\begin{abstract}
Aggregating multiple input rankings into a consensus ranking is essential in various fields such as social choice theory, hiring, college admissions, web search, and databases. A major challenge is that the optimal consensus ranking might be biased against individual candidates or groups, especially those from marginalized communities. This concern has led to recent studies focusing on fairness in rank aggregation. The goal is to ensure that candidates from different groups are fairly represented in the top-$k$ positions of the aggregated ranking.

We study this fair rank aggregation problem by considering the Kendall tau as the underlying metric. While we know of a polynomial-time approximation scheme (PTAS) for the classical rank aggregation problem, the corresponding fair variant only possesses a quite straightforward 3-approximation algorithm due to Wei et al., SIGMOD'22, and Chakraborty et al., NeurIPS'22, which finds closest fair ranking for each input ranking and then simply outputs the best one. 

In this paper, we first provide a novel algorithm that achieves $(2+\epsilon)$-approximation (for any $\epsilon > 0$), significantly improving over the 3-approximation bound. Next, we provide a $2.881$-approximation fair rank aggregation algorithm that works irrespective of the fairness notion, given one can find a closest fair ranking, beating the 3-approximation bound. We complement our theoretical guarantee by performing extensive experiments on various real-world datasets to establish the effectiveness of our algorithm further by comparing it with the performance of state-of-the-art algorithms.
\end{abstract}

\section{Introduction}
Ranking a list of alternatives to prioritize desirable outcomes among a set of candidates is ubiquitous across various applications, such as hiring, admissions, awarding scholarships, and approving loans. When multiple voters provide preference orders or rankings on candidates, which may conflict, the task of producing a single consensus ranking is the classical \emph{rank aggregation} problem. This problem is central to many fields, from social choice theory~\cite{brandt2016} to information retrieval~\cite{harman1992}. Its origins trace back to the 18th century~\cite{borda1781memoire,condorcet1785essai}, and it has been extensively studied from a computational standpoint over the past few decades~\cite{dwork2001,fagin2003,gleich2011,azari2013}. When formulated as an optimization problem, one of the most popular versions seeks to find a consensus ranking that minimizes the sum of distances to the input rankings~\cite{kemeny1959,young1988,young1978,dwork2001,ailon2008}.

In this paper, we address the rank aggregation problem with an additional fairness constraint on the final consensus ranking. Ranking algorithms are commonly used to select the top candidates for various opportunities and services, such as admissions or scholarships in the education system, job hiring, or the allocation of medical care during emergencies like pandemics. In today's context, it is essential for any ranking algorithm to produce a fair ranking to ensure equitable selection and to avoid the risk of reinforcing extreme ideologies or stereotypes about marginalized communities based on sensitive attributes such as gender, race, or caste~\cite{costello2016views,kay2015unequal,bolukbasi2016man}. For example, systems like job and education reservations in India~\cite{borooah2010social} or affirmative action-based university admissions in the USA~\cite{deshpande2005affirmative} have been implemented to address under-representation and discrimination.

We consider the notion of proportionate fairness, also known as $p$-fairness~\cite{baruah1996proportionate}, which ensures that each of the protected classes within the population is fairly represented in the top ``most relevant" (top-$k$) positions of the final consensus ranking. The study of proportionate fairness in the context of rank aggregation was first explored in~\cite{wei2022rank} and~\cite{chakraborty2022}. In this paper, we use the following definition of \emph{fair ranking} from~\cite{chakraborty2022}.\footnote{Note, the definition used in~\cite{wei2022rank}, though similar, is slightly restrictive.}

\begin{definition}[Fair Ranking]
    \label{def:fair-ranking}
    Consider a partition of $d$ candidates into $g$ groups $G_1,\cdots, G_g$. For each group $G_i$ ($i \in [g]$), let us consider two parameters $\alpha_i,\beta_i \in [0,1]$. For $\bar{\alpha} = \left( \alpha_1,\cdots,\alpha_g \right)$, $\bar{\beta} = \left( \beta_1,\cdots,\beta_g \right)$, and $k \in [d]$, a ranking $\pi$ (on $d$ candidates) is said to be \emph{$(\bar{\alpha},\bar{\beta})$-$k$-fair} if for each $G_i$:
    \begin{itemize}
        \item \textit{Minority Protection}: The top-$k$ positions $\pi(1),\ldots, \pi(k)$ contain at least $\lfloor \alpha_i \cdot k \rfloor$ candidates from $G_i$, and 
        \item \textit{Restricted Dominance}: The top-$k$ positions $\pi(1),\ldots, \pi(k)$ contain at most $\lceil \beta_i \cdot k \rceil$ candidates from $G_i$.
    \end{itemize}
\end{definition}
It is important to note that other notions of fair ranking, such as top-$k$ statistical parity, have been considered previously~\cite{kuhlman2020rank}. However, this approach is quite restrictive and does not satisfy the criteria for proportionate fairness. For a concrete example demonstrating why proportionate fairness is a much stronger concept than statistical fairness, see~\cite{wei2022rank}.

Given a set of $n$ rankings provided by voters on $d$ candidates, the \emph{fair rank aggregation} problem asks to find a fair consensus ranking that minimizes the sum of distances to the input rankings. Various distance measures have been considered in the literature to capture the dissimilarity between pairs of rankings, with the \emph{Kendall tau distance} -- which counts the number of pairwise disagreements between two rankings -- being one of the most popular. \cite{wei2022rank} and~\cite{chakraborty2022} proposed the following simple algorithm: Find a closest fair ranking for each input ranking and then output the one with the minimum sum of distances. \cite{wei2022rank} and~\cite{chakraborty2022} provided a 2-approximation and exact algorithm respectively for the \emph{closest fair ranking} problem. A straightforward application of the triangle inequality shows that this simple strategy only achieves a $3$-approximation for the fair rank aggregation. Moreover, since the final output ranking is close to one of the input rankings, it is essentially influenced by the preference order of a single voter. To date, there is no improvement over this $3$-factor approximation guarantee. It is also worth noting that without any fairness constraint, the classical rank aggregation problem (known to be \texttt{NP}-hard~\cite{bartholdi1989voting,Dwork2002RankAR}) has an $(1+\epsilon)$-approximation algorithm for any $\epsilon > 0$~\cite{mathieu2009rank}.

\subsection{Our contribution} The main contribution of our paper is the development of a new algorithm for the fair rank aggregation problem under proportional fairness. Our algorithm achieves a $(2+\epsilon)$-approximation, for any $\epsilon >0$. To demonstrate our result, we design a two-stage procedure. First, we introduce a new problem of partitioning a colored graph in a \emph{colorful} manner while minimizing the cost of ``backward" edges across the cut. We develop a novel algorithm to solve this problem exactly when the input graph is a weighted tournament that satisfies certain natural properties on edge weights. This problem can be thought of as a variant of the \emph{constraint cut} problem on a special graph class, which is, in general, NP-hard. Different cut problems find applications in other fairness questions (e.g.,~\cite{dinitz2022fair}), and thus, our result on the variant of the constraint cut problem could be of independent interest. Next, we construct a weighted tournament graph from the fair rank aggregation instance. We then apply the solution of an optimal colorful partitioning of that tournament and use the known PTAS for the rank aggregation problem on both partitions separately to produce a fair ranking over the entire set of candidates. Finally, we argue that the output ranking attains $(2+\epsilon)$-approximation for any $\epsilon >0$.

We implement our algorithm and compare it against baselines on multiple standard datasets by varying different parameters. Our results show that the output of our algorithm achieves a significantly better objective value (i.e., the sum of distances) compared to state-of-the-art algorithms for fair rank aggregation. Furthermore, although our theoretical analysis guarantees only a $(2+\epsilon)$-approximation for our proposed algorithm, it consistently performs much better in practice -- the output is almost always very close to an optimal solution.

Our next contribution is a generic fair rank aggregation algorithm that achieves a $2.881$-approximation. We emphasize that our algorithm works \emph{irrespective} of the fairness notion under consideration as long as there is an efficient procedure to compute a closest fair ranking for any input ranking (even an approximately close fair ranking procedure suffices, albeit with a worse approximation factor for the aggregation problem). Thus, our algorithm provides an approximation guarantee not only with respect to a specific type of fairness but also concerning any plausible definition of fairness. As an immediate corollary, we achieve $2.881$-approximation for the fair rank aggregation under a stronger fairness notion like \emph{block fairness} introduced by~\cite{chakraborty2022}. Further, our generic algorithm works even if the group information -- which candidate belongs to which group -- is not known (e.g., as in \emph{robust fairness}~\cite{kliachkin2024fairness}). Ours is the first generic approximation algorithm that breaks below the straightforward 3-factor bound obtained by the naive use of the triangle inequality. We present our generic (deterministic) algorithm in Section~\ref{sec:improvedfra}, whose running time can be improved significantly using random sampling and \emph{coreset} construction; however, such a randomized procedure requires a much more intricate analysis as detailed in Section~\ref{sec:faster-generic-algo}.

\subsection{Other related works}
The rank aggregation problem without any fairness constraint has also been studied under different other metrics, including Spearman footrule~\cite{Dwork2002RankAR}, Ulam~\cite{chakraborty2021approximating,chakraborty2023clustering}. \cite{chakraborty2022} considered the rank aggregation problem with the fairness constraint under both Spearman footrule and Ulam metric, and showed a $3$-approximation guarantee. For the Ulam metric, they in fact provided a $(3-\delta)$-approximation result, for some constant $\delta \le 2^{-30} $, albeit only for a constant number of groups.

Apart from the rank aggregation problem, other ranking problems have also been studied under fairness. E.g.,~\cite{celis2018ranking} explored the problem of finding the closest proportional fair ranking to a given ranking for metrics such as Bradley-Terry, DCG, and Spearman footrule. Ensuring \emph{robust fairness} in rankings has also been studied~\cite{kliachkin2024fairness}, where given an input ranking, the goal is to find a close ranking that is more fair, even when the protected attributes are not known. We emphasize that such an algorithm may not necessarily be able to find a good aggregate ranking.

The rank aggregation problem is essentially the 1-clustering (1-median) problem, where the input is a set of rankings. The past few years have witnessed a surge in research on fair clustering~\cite{huang2019coresets,chen2019proportionally,bera2019fair,backurs2019scalable}. However, we must note that the notion of fairness in the general clustering context differs from that in the rank aggregation.

\section{Preliminaries}
\label{sec:prelims}
\paragraph{Notations.} For any $n \in \mathbb{N}$, let $[n]$ denote the set $\{1, 2, \cdots , n\}$. We refer to the set of all rankings (or permutations) over $[d]$ by $\S_d$. We consider any permutation $\pi \in \S_d$ as a sequence of numbers $\pi(1), \pi(2), \ldots , \pi(d)$ where the rank of $\pi(i)$ is $i$. For any two elements $a, b \in [d]$ and a permutation $\pi \in \S_d$, we use the notation $a \prec_{\pi} b$ to denote that the rank of $a$ is less than that of $b$ in $\pi$. 

\paragraph{Distance metric and fair rank aggregation.} In this paper, we consider the Kendall tau distance to measure the dissimilarity between any two rankings or permutations.

\begin{definition}[Kendall tau distance]
Given two permutations $\pi_1, \pi_2 \in \S_d$, the \emph{Kendall tau distance} between them, denoted by $\kappa(\pi_1, \pi_2)$, is the number of pairwise disagreements between $\pi_1$ and $\pi_2$, i.e.,
    \[\kappa(\pi_1, \pi_2) := \left| \left\{(a, b) \in [d] \times [d] \mid \, a \prec_{\pi_1} b \text{\hspace{2mm}but\hspace{2mm}} b \prec_{\pi_2} a\right \}\right |\]
\end{definition}

Next, we define the fair rank aggregation problem.

\begin{definition} (Fair Rank Aggregation) Given a set $S $ of rankings over $d$ candidates that are partitioned into $g$ groups $G_1,\cdots, G_g$, $\bar{\alpha} = \left( \alpha_1,\cdots,\alpha_g \right) \in [0,1]^g$, $\bar{\beta} = \left( \beta_1,\cdots,\beta_g \right) \in [0,1]^g$, and $k \in [d]$, the \emph{fair rank aggregation} problem asks to find a $(\bar{\alpha},\bar{ \beta})$-$k$-fair ranking $\sigma \in \S_d$ that minimizes the objective function $\obj(S, \sigma) := \sum_{\pi \in S} \kappa(\pi, \sigma)$.
\end{definition}

Note that in the above definition, the minimization is over the set of all $(\bar{\alpha},\bar{ \beta})$-$k$-fair rankings in $\S_d$. When the set $S$ is clear from context, for brevity, we simply refer to the objective value as $\obj(\sigma)$. Let $\sigma^*$ be an optimal fair aggregated rank, and $\opt(S) := \obj(S, \sigma^*)$. We call a $(\bar{\alpha},\bar{ \beta})$-$k$ fair ranking $\tilde{\sigma}$ a \emph{$c$-approximate fair aggregate ranking} (for some $c \ge 1$) for the set $S$ iff $\obj(S, \tilde{\sigma}) \le c \cdot \opt(S)$.

\paragraph{Weighted tournament.}
A weighted tournament $T = (V,A)$ is a directed graph where for every pair of vertices $u,v \in V$, both the edges $(u,v)$ and $(v,u)$ are present with some non-negative weight. It is well-known that the rank aggregation problem can be cast as \emph{feedback arc set} problem on a weighted tournament (see~\cite{ailon2008}), where the corresponding weighted tournament $T = (V,A)$ with weight function $w : A\to \mathbb{R}$ satisfies the following:

\begin{itemize}
    \item \textit{Probability Constraints:}
    \begin{equation}
        \label{eq:prob} \forall i, j \in V, \; w(i, j) + w(j, i) = 1, 
    \end{equation}
    
    \item \textit{Triangle Inequality:} 
    \begin{equation}
        \label{eq:triangle} \forall i, j, k \in V, \; w(i, j) \leq w(i, k) + w(k, j).
    \end{equation}
\end{itemize}

For a weighted tournament $T = (V, A)$ with weight function $w : A \to \mathbb{R}$, for any \(v \in V\), the \emph{in-neighborhood} of \(v\) is defined as \(N(v) := \{u \in V \mid (u,v) \in A\}\), and the \emph{weighted in-degree} of \(v\) is defined as \(\delta(v) := \sum_{u \in N(v)} w(u, v)\).

\section{Colorful Bi-partition on Tournaments} \label{sec:color-partition}

In this section, we first introduce the \emph{colorful bi-partition} problem defined on a directed weighted graph with colored vertices. Then, we provide an algorithm to solve that problem when the input graph is a tournament that satisfies both the probability constraint and the triangle inequality. In the next section, we discuss how to use the solution of the colorful bi-partition problem on tournaments to get an approximation algorithm for the fair rank aggregation problem.

Consider a weighted directed graph $G=(V,A)$ with a weight function $w:A \to \mathbb{R}$ defined on arcs/edges and a color function $\col: V \to [g]$ (for some integer $g \ge 1$) defined on vertices, and $\bar{\alpha} = \left( \alpha_1,\cdots,\alpha_g \right) \in [0,1]^g$, $\bar{\beta} = \left( \beta_1,\cdots,\beta_g \right) \in [0,1]^g$. We call a subset $S \subseteq V$ \emph{$(\bar{\alpha},\bar{\beta})$-colorful} if for each color $i \in [g]$, $S$ contains at least $\lfloor \alpha_i \cdot |S| \rfloor$ and at most $ \lceil \beta_i \cdot |S| \rceil $ many vertices of color $i$, i.e., 
\[
\forall {i \in [g]},\; \lfloor \alpha_i \cdot |S| \rfloor \le \left| S \cap \col^{-1}(i) \right| \le \lceil \beta_i \cdot |S| \rceil .
\]

\begin{definition}[Colorful Bi-partition Problem]
    \label{def:colorful-partition}
    Given a weighted colored directed graph $G=(V,A)$ with $w:A \to \mathbb{R}$, $\col: V \to [g]$ (for some integer $g \ge 1$), $\bar{\alpha} \in [0,1]^g$, $\bar{\beta} \in [0,1]^g$, and an integer $k$, the \emph{colorful bi-partition} problem asks to find a partitioning of $V$ into (disjoint) sets $L$ and $V \setminus L$ such that 
    (i) $|L| = k$, and
    (ii) $L$ is $(\bar{\alpha},\bar{\beta})$-colorful;
while minimizing the \emph{cost of the partition} $(L,V\setminus L)$ defined as the total weights of the arcs going from $V \setminus L$ to $L$, i.e., $\cost(L, V\setminus L):=\sum_{(y, x) \in A : x \in L, y \in V\setminus L} w(y,x)$.
\end{definition}
Note, in the above problem, we want only $L$ to be colorful, so $V\setminus L$ need not be colorful. Since specifying the set $L$ suffices to identify the partition $(L, V \setminus L)$, for brevity, we use $\cost(L)$ to denote $\cost(L , V \setminus L)$.

Next, we provide a (deterministic) algorithm to solve the colorful bi-partition problem on tournaments, satisfying both the probability constraint (Equation~\ref{eq:prob}) and the triangle inequality constraint (Equation~\ref{eq:triangle}).

\begin{theorem}
    \label{thm:bi-partition}
    There is an algorithm that, given a weighted colored tournament $T=(V,A)$ with $w:A \to \mathbb{R}$, $\col: V \to [g]$ (for some integer $g \ge 1$), satisfying both the probability and the triangle inequality constraints, and an integer $k$, $\bar{\alpha} \in [0,1]^g$, $\bar{\beta} \in [0,1]^g$, finds an optimal colorful bi-partition in time $O(|A| + |V|\log|V|)$.
\end{theorem}

\begin{algorithm}[ht]
\caption{\textsc{Colorful Bi-Partition}}
\label{alg:min_bi_partition}
\begin{algorithmic}[1]
\Procedure{ColBipartition}{$T=(V,A), \bar{\alpha}, \bar{\beta}, k$}
    \State Initialize an empty set $L$
    \For{$i \gets 1$ to $g$}
        \State Sort vertices of color $i$, i.e., in $\col^{-1}(i)$, in non-decreasing order by their weighted in-degrees and process them in that sorted order
        \State $count_i \gets \lfloor \alpha_i \cdot  k \rfloor$
        \For{each vertex $v \in  \col^{-1}(i)$}
            \If{$count_i > 0$}
                \State Add vertex $v$ to set $L$
                \State $count_i \gets count_i - 1$
            \EndIf
        \EndFor
    \EndFor
    \State Sort remaining vertices ($V \setminus L$) collectively in non-decreasing order by their weighted in-degrees and process them in that sorted order 
    \For{each $v \in V\setminus L$} 
        \If{$|L| \leq k$}
           \State $i \gets \col(v)$
           \If{ $|L \cap \col^{-1}(i)| \leq \lceil\beta_i \cdot k \rceil$}
               \State Add vertex $v$ to set $L$
           \EndIf
        \EndIf
    \EndFor
    \\
    \Return The partition $(L, V\setminus L)$
\EndProcedure
\end{algorithmic}
\end{algorithm}

\paragraph{Description of the algorithm. }
Our colorful bi-partition algorithm (Algorithm~\ref{alg:min_bi_partition}) for tournament $T$ works as follows: 
\begin{description}
    \item 1. \textbf{Sorting vertices of each color:} For each color $i \in [g]$, arrange the vertices in non-decreasing order based on their weighted in-degrees (Line 4).
    \item 2. \textbf{Initial selection in $L$:} To select the vertices in $L$: Take the first $\lfloor \alpha_i \cdot k \rfloor $ vertices according to the sorted order of the vertices of color $i$, $\forall i \in [g]$ (Lines 5 -- 11).
    \item 3. \textbf{Filling up the set $L$:} Order all the remaining vertices collectively in the non-decreasing order of their weighted in-degrees (Line 13), and continue adding elements to $L$ as per this sorted order. If adding an element causes the number of elements from any color $i$ to exceed the $\lceil \beta_i \cdot k \rceil$ bound, skip the element and proceed to the next in the collective ordering (Lines 14 -- 21).
    \item 4. \textbf{The final bi-partition:} Once $L$ is filled, the remaining $V \setminus L$ forms the other set of the bi-partition (Line 22).
\end{description}

\paragraph{Running time of Algorithm~\ref{alg:min_bi_partition}. } The running time of the algorithm depends on two main factors: the total in-degree calculation of each vertex and the sorting of the vertices according to their respective in-degrees. The rest of the steps in the algorithm look at the vertices at least once and at most twice to get the bi-partition. To calculate the weighted in-degree of each vertex, we need to calculate the sum of the weights of the edges incoming to that vertex. That takes a total time of $O(|A|)$ as we need to scan all the edges in the graph. We need $O(|V|\log|V|)$ time to sort the vertices by their weighted in-degrees. The rest of the algorithm needs a linear scan to get a bi-partition, which takes $O(|V|)$ time. In our tournament $T$, we have $|V| = d$ and $|A| = d(d-1)$. Thus, the running time is $O(d^2)$.

\paragraph{Approximation guarantee. }
Let us now introduce a few notations that are useful for the analysis. Consider a graph $T = (V,A)$. For two sets \(X,Y \subseteq V\), let \(A(X,Y)\) denote the set of arcs from \(X\) to \(Y\) in $T$, i.e., 
\[
A(X,Y) : = \left\{ (x,y) \in A \mid x \in X, y \in Y \right\}.
\]
When $X$ is the singleton set $\{x\}$, for notational convenience, we use $A(x,Y)$ to denote $A(X,Y)$. For any subset of arcs $A' \subseteq A$, we use $w(A')$ to denote the sum of weights of the arcs in $A'$, i.e., $w(A') : = \sum_{(x,y)\in A'} w(x,y)$. Note, for any partition $(P,V \setminus P)$, $\cost(P) = w(A(V \setminus P,P))$.
\begin{center}
    \begin{figure}[h]
        \centering
        \includegraphics[width = 0.8\textwidth]{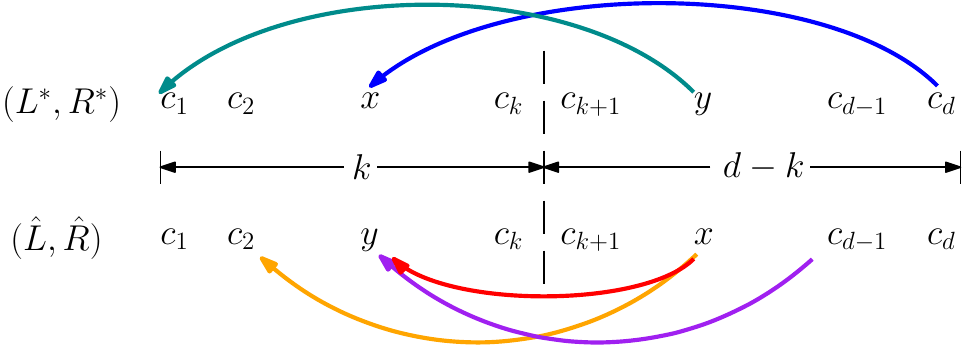}
        \caption{Vertices $\{c_1, c_2, \ldots, c_d\}$ are sorted by their weighted in-degrees. For the bi-partition $(L^*, R^*)$, one of the edges in $A(y, L^*)$ is shown in \textcolor{darkcyan}{cyan} and one of the edges in $A(R^*, x)$ is shown in \textcolor{blue}{blue}. For the bi-partition $(\hat{L}, \hat{R})$, the edge $(x, y)$ is shown in \textcolor{red}{red}. Also, one of the edges in $A(x, L^*)$ is shown in \textcolor{orange}{orange} and one of the edges in $A(R^*, y)$ is shown in \textcolor{violet}{violet}.}
        \label{fig:fairPartition}
    \end{figure}
\end{center}

We first argue that if we have a vertex in $V \setminus L$ with weighted in-degree smaller than or equal to that of some vertex in $L$, swapping them cannot lead to a new bi-partition with greater cost.

\begin{lemma}
\label{lemma:optimal-bipartition-swap}
    Let $(L^*, R^*)$ be any (not necessarily colorful) bi-partition. Suppose there exists $x \in L^*$ and $y \in R^*$ such that $\delta(x) \ge \delta(y)$. Then for $\hat{L}:=(L^* \setminus \{x\}) \cup \{y\}$, $\cost(\hat{L}) \le \cost(L^*)$.
\end{lemma}
\begin{proof}[Proof of Lemma~\ref{lemma:optimal-bipartition-swap}]

    Let us define two sets $L' := L^* \setminus \{x\}$ and $R' := R^* \setminus \{y\}$. (Note, $x, y \not \in L' \cup R'$.)
    
    Observe, by the probability constraints (Equation~\ref{eq:prob}) with respect to vertices $x$ and $y$ with the set $L'$, we get:
    \begin{equation}\label{eq:forxL}
        w(A(L', x)) + w(A(x, L')) = |L'|
    \end{equation}
    \begin{equation}\label{eq:foryL}
        w(A(L', y)) + w(A(y, L')) = |L'|
    \end{equation}

    Then by definition, we can write
    \begin{align}\label{eq:befSwap}
        w(A(R^*, L^*)) & = w(A(y, L')) + w(A(R', x)) + w(A(R', L')) + w(y, x).
    \end{align}

    Let us now swap $x$ and $y$ to form a new bi-partition. Let $\hat{L} = L' \cup \{y\}$, and $\hat{R} = R' \cup \{x\}$. Again, by definition,
    \begin{align}\label{eq:aftSwap}
        w(A(\hat{R}, \hat{L})) & = w(A(x,L')) + w(A(R', y)) + w(A(R', L')) + w(x, y).
    \end{align}
    
    Using Equation~\ref{eq:foryL} and Equation~\ref{eq:befSwap}, we get:
    \begin{align}\label{eq:1}
        w(A(R^*, L^*)) & = |L'| - w(A(L', y)) + w(A(R', x)) + w(A(R', L')) + w(y, x).
    \end{align} 
    Similarly, using Equation~\ref{eq:forxL} and Equation~\ref{eq:aftSwap}, we get:
    \begin{align}\label{eq:2}
         w(A(\hat{R}, \hat{L})) & = |L'| - w(A(L', x)) + w(A(R', y)) + w(A(R', L'))   + w(x, y).
    \end{align}
    
    Thus, by subtracting Equation~\ref{eq:1} from Equation~\ref{eq:2}, we get:
    \begin{align*}
        w(A(\hat{R}, \hat{L})) - w(A(R^*, L^*)) & = - w(A(L', x)) + w(A(R', y)) + w(x, y) + w(A(L', y)) - w(A(R', x)) - w(y, x)\\
        & = \left( w(A(R', y)) + w(A(L', y)) + w(x, y) \right) - \left( w(A(R', x)) + w(A(L', x)) + w(y, x) \right) \\
        & = \delta(y) - \delta(x) \qquad \qquad \text{(By definition of weighted in-degrees)}\\
        & \le 0.
    \end{align*}
   Thus, we have
    $w(A(\hat{R}, \hat{L})) \le w(A(R^*, L^*))$, which concludes the proof of the lemma.
\end{proof}
Now, by assuming the above lemma, we prove the main result of this section.

\begin{proof}[Proof of Theorem~\ref{thm:bi-partition}]
    Let $(L^*, R^*)$ be an (arbitrary) optimal colorful bi-partition. Recall that the bi-partition output by Algorithm~\ref{alg:min_bi_partition} is $(L,R)$. We first show that there exists a colorful bi-partition $\hat{L}$ (if not $L^*$) such that for all $i \in [g]$, $|\col^{-1}(i) \cap \hat{L}| = |\col^{-1}(i) \cap L|$ and $\cost(\hat{L}) \le \cost(L^*)$.

    Consider some $i \in [g]$ such that $|\col^{-1}(i) \cap L^*| > |\col^{-1}(i) \cap L|$ and some $j \in [g]$ such that $|\col^{-1}(j) \cap L^*| < |\col^{-1}(j) \cap L|$. Let $x = \arg\max_{u \in \col^{-1}(i) \cap (L^* \setminus L)} \delta(u) $, and $y = \arg\min_{v \in \col^{-1}(j) \cap (L \setminus L^*)} \delta(v)$. We now argue that $\delta(y) \le \delta(x)$. Suppose $y$ is added to $L$ while executing Line 18 of Algorithm~\ref{alg:min_bi_partition}. Then, as Algorithm~\ref{alg:min_bi_partition} iterates through the vertices in non-decreasing order of weighted in-degrees in Lines 14 -- 21, it must have added $y$ to $L$, and never encountered $x$ before terminating. Therefore, it must be that $\delta(y) \le \delta(x)$. Suppose instead that $y$ is added to $L$ while executing Line 8. As $L^*$ is colorful, we have $\lfloor \alpha_i \cdot k \rfloor \le |\col^{-1}(j) \cap L^*| < |\col^{-1}(j) \cap L|$. Therefore there must exist some $y' \in \col^{-1}(j)$ such that $\delta(y) \le \delta(y')$, and $y'$ is added to $L$ in our algorithm while executing Line 18. Then we conclude that $\delta(y) \le \delta(y') \le \delta(x)$ (as in Lines 14 -- 21, it must have added $y'$ to $L$ and never encountered $x$ before terminating). Now, consider the partition $L' = (L^* \setminus \{x\}) \cup \{y\}$. Observe, $(L',V\setminus L')$ is also a colorful bi-partition. By Lemma~\ref{lemma:optimal-bipartition-swap}, $\cost(L') \le \cost(L^*)$.

    By using the above argument repeatedly, we obtain a colorful bi-partition $(\hat{L}, V\setminus \hat{L})$ such that $\forall_{i \in [g]}$, $|\col^{-1}(i) \cap \hat{L}| = |\col^{-1}(i) \cap L|$ and $\cost(\hat{L}) \le \cost(L^*)$.

    We next prove that $\cost(L) \le \cost(\hat{L})$. Suppose $L \neq \hat{L}$. Then consider an $i \in [g]$ such that $\col^{-1}(i) \cap \hat{L} \neq \col^{-1}(i) \cap L$. By construction, our algorithm always adds vertices of color $i$ to $L$ by order of non-decreasing weighted in-degree. This implies that there exists a $v \in \col^{-1}(i) \cap (\hat{L} \setminus L)$ such that $\delta(v) \ge \min_{y \in \col^{-1}(i) \cap (L \setminus \hat{L})} \delta (y)$; otherwise, $v$ would have been added to $L$. Then for $\hat{L'} = (\hat{L} \setminus \{v\}) \cup \{y\}$, by Lemma~\ref{lemma:optimal-bipartition-swap}, $\cost(\hat{L'}) \le \cost(\hat{L})$. By repeating this swapping argument, we obtain the bi-partition $(L,R)$. As each swap can only reduce the cost, we have that $\cost(L) \le \cost(\hat{L}) \le \cost(L^*)$, showing that the output of Algorithm~\ref{alg:min_bi_partition} is an optimal colorful bi-partition.
\end{proof}

\section{Approximating the Fair Aggregate Ranking}
\label{sec:fair-rank}
In this section, we show our main result by designing an algorithm that finds $(2+\epsilon)$-approximate fair aggregated rank for any $\epsilon > 0$. In particular, we prove the following theorem.

\begin{theorem}
    \label{thm:main}
    For any $\epsilon > 0$, there exists a $(2+\epsilon)$-approximation algorithm for the fair rank aggregation problem that runs in time $O(d^3 \log d + nd^2)$, where $O(\cdot)$ hides the dependency on $1/\epsilon$.
\end{theorem}

To show the above result, we follow a two-step procedure.
\begin{itemize}
    \item Step I: Create a weighted colored tournament $T = (V,A)$ on $d$ vertices, which is an instance of the colorful bi-partition problem, and then use Algorithm~\ref{alg:min_bi_partition} to get an optimal bi-partition $(L, V \setminus L)$.
    \item Step II: Use $(L, V \setminus L)$ and apply the known PTAS for the rank aggregation problem without fairness constraint on the portions of $L$ and $V \setminus L$ separately to get an approximate fair aggregated rank.
\end{itemize}

 We start with proving the following theorem, which, together with a known PTAS for the rank aggregation problem (without the fairness constraint), establishes Theorem~\ref{thm:main}.

\begin{theorem}
\label{thm:fair-approx-ratio}
    Suppose there is a $t_1(d, n)$-time $c_1$-approximation algorithm $\mathcal{A}_1$ for some $c_1 \ge 1$ for the rank aggregation problem, and a $t_2(d)$-time $c_2$-approximation algorithm $\mathcal{A}_2$ for some $c_2\ge 1$ for the colorful bi-partition problem on tournaments satisfying both the probability and the triangle inequality constraints. Then there exists a $(c_1 + c_2)$-approximation algorithm for the fair rank aggregation problem with running time $O(nd^2 + t_1(d, n) + t_2(d))$. 
\end{theorem}

\paragraph{Description of the algorithm. }
Let us start by defining a few useful notations. For any set of elements $I \subseteq [d]$, let $\pi_I$ represent the restriction of $\pi$ to the elements in $I$. That is, delete all elements that are not contained in $I$ from $\pi$. E.g., for $\pi = (2, 6, 3, 5, 1, 4)$ and $I = \{1, 2, 3\}$, $\pi_I = (2, 3, 1)$.

Suppose we are given a set $S$ of rankings over $[d]$, of size $n$. We construct a weighted tournament graph $T$ from the rank aggregation instance by setting $V = [d]$. We set the $\col$ function such that for all $v \in V$, $\col(v) = i$ if $v \in G_i$. For every pair of elements $a, b$, let $n_{ab} = |\{ \pi \in S \mid a \prec_{\pi} b\}|$. Set the weight of the edge $(a, b)$ to be $w(a,b) = n_{ab} / n$. Observe that the edge weights of $T$ obey the probability constraint and the triangle inequality constraint. Then we run the algorithm $\mathcal{A}_2$ with the graph $T$ and parameters $\bar{\alpha}, \bar{\beta}, k$ and obtain a partitioning $L$ and $V \setminus L$. For brevity, let $R:= V \setminus L$.

We then restrict the input rankings to $L$ and $R$; let $S_L = \{\pi_L \mid \pi \in S\}$ and $S_R = \{\pi_R \mid \pi \in S\}$. We apply the rank aggregation algorithm $\mathcal{A}_1$ on $S_L$ and $S_R$ separately, to obtain $\pi^p_{L}$ and $\pi^p_{R}$ respectively. Construct $\pi^p$ by concatenating $\pi^p_{L}$ with $\pi^p_{R}$ and return $\pi^p$ as the output aggregate ranking. 

We give a formal description of the algorithm in~\ref{alg:fair-aggregation}.

\begin{algorithm}[!h]
\caption{\textsc{Fair Rank Aggregation}}
\label{alg:fair-aggregation}
\begin{algorithmic}[1]
\Procedure{FairRankAggregation}{$S$}
    \State Initialize tournament graph $T = (V, A)$
    \State Set $V = [d]$. 
    \For{$a \in [d], b \in [d]$}
        \State $w(a, b) = |\{ \pi \in S \mid a \prec_{\pi} b\}| / |S|$
    \EndFor
    \State Call $\mathcal{A}_2$ with $T$ and parameters $\bar{\alpha}, \bar{\beta}, k$ to obtain set $L$; let $R := V \setminus L$
    \State Compute set $S_L = \{\pi_L \mid \pi \in S\}$
    \State Compute set $S_R = \{\pi_R \mid \pi \in S\}$
    \State Call $\mathcal{A}_1$ on $S_L$ to obtain a ranking $\pi^p_L$
    \State Call $\mathcal{A}_1$ on $S_R$ to obtain a ranking $\pi^p_R$
    \State $\pi^p \gets $ concatenation of $\pi^p_L$ and $\pi^p_R$
    \\
    \Return $\pi^p$
\EndProcedure
\end{algorithmic}
\end{algorithm}

\paragraph{Running time of Algorithm~\ref{alg:fair-aggregation}. }
Observe that the running time to construct $T$ is $O(nd^2)$, as there are $O(d^2)$ pairs of elements, and each input ranking must be inspected. Constructing $S_L$ and $S_R$ can be done in time $O(nd)$. We run algorithm $\mathcal{A}_1$ twice, and $\mathcal{A}_2$ once. Therefore, it is clear that the running time is $O(nd^2 + t_1(d, n) + t_2(d))$.

\paragraph{Approximation guarantee. }

Recall the Kendall-Tau distance between two rankings is equal to the number of pairwise disagreements. Observe that given a partition of $[d]$ into $L$ and $R$, the Kendall-tau distance between any two rankings can be decomposed as follows \[ \kappa(\pi, \sigma) = \kappa_{L \times L}(\pi, \sigma) + \kappa_{R \times R}(\pi, \sigma)+ \kappa_{L \times R}(\pi, \sigma)\]

where for any $X \subseteq [d], Y \subseteq [d]$, \[ \kappa_{X \times Y}(\pi, \sigma) := |\{(a, b) \in X \times Y \mid a \prec_{\pi} b \text{\hspace{2mm}but\hspace{2mm}} b \prec_{\sigma} a\}| \]

Similarly, we decompose the objective value of any ranking $\pi$ to $S$ as follows
\begin{align*}
    & \obj(S, \pi) = \obj(S, \pi)_{L \times L} + \obj(S, \pi)_{R \times R} + \obj(S, \pi)_{L \times R}
\end{align*}

where for any $X \subseteq [d], Y \subseteq [d]$, \[ \obj(S, \pi)_{X \times Y} := \sum_{\pi_i \in S} \kappa_{X \times Y}(\pi, \pi_i). \]

\begin{proof}[Proof of Theorem~\ref{thm:fair-approx-ratio}]

    First observe, by construction, $\pi^p$ is a $(\bar{\alpha}, \bar{\beta})$-$k$-fair ranking. This is because only the elements of $L$ are placed in the top-$k$ positions of $\pi^p$, and the set $L$ is $(\bar{\alpha},\bar{\beta})$-colorful. Next, we analyze the approximation ratio. For analysis, let $\pi^*$ be an (arbitrary) optimal $(\bar{\alpha}, \bar{\beta})$-$k$-fair aggregate ranking. We denote the set of elements placed in its top-$k$ positions as $L^*$ and the set of remaining elements as $R^*$.

    As $L$ and $R$ form a partition over the set $[d]$, 
    \begin{align}
    \label{eqn:obj-decomposition}
        & \obj(\pi^p) = \obj(\pi^p)_{L \times L} + \obj(\pi^p)_{R \times R} + \obj(\pi^p)_{L \times R} .
    \end{align}
    
    We first consider the term $\obj(\pi^p)_{L \times R}$. Recall that the set $L$ is formed using the $c_2$-approximation algorithm $\mathcal{A}_2$ for the colorful bi-partition problem on $T$. Observe that by the construction, the weight of an edge $(a, b)$ in $T$ is equal to the contribution to the cost of ordering $b$ before $a$ in an aggregate ranking. More specifically, $\cost(L) = \obj(S, \pi^p)_{L \times R}$. As $(L^* , R^*)$ constitute a feasible colorful bi-partition of the tournament $T$, we deduce that $\obj(\pi^p)_{L \times R} \leq c_2 \cdot \obj(\pi^*)_{L^* \times R^*}$. Further, it is straightforward to see that $\obj(\pi^*)_{L^* \times R^*} \le \obj(\pi^*)$, as the remaining pairs can only increase the objective cost. Therefore, 
    
    \begin{equation}
    \label{eqn:crossing-cost-bound}
        \obj(S, \pi^p)_{L \times R} \le c_2 \cdot \obj(S, \pi^*) .
    \end{equation}

    We next analyze the contribution of the term $\obj(S, \pi^p)_{L \times L}$. Recall $\pi^p_{L}$ is obtained by using a $c_1$-approximation algorithm $\mathcal{A}_1$ (for the rank aggregation problem without fairness constraint) on the rankings $S_L$. Note, $\obj(S, \pi^p)_{L \times L} = \obj(S_L, \pi^p_L)$. Now, let $\sigma^*_L$ be an (arbitrary) optimal aggregate ranking over the candidates in $L$ for the rankings $S_L$. We must have that $\obj(S_L, \sigma^*_L) \le \obj(S, \pi^*)_{L \times L}$, as all rankings in $S_L$ maintain identical pairwise ordering for elements in $L$ as their corresponding full rankings in $S$. Therefore, 
    \begin{align*}
        \obj(S, \pi^p)_{L \times L} & \le c_1 \cdot \obj(S_L, \sigma^*_L) \nonumber \\
        & \le c_1 \cdot \obj(S, \pi^*)_{L \times L} . 
    \end{align*}
    
    The same argument also holds when we consider $\pi^p_R$. Since $L$ and $R$ are disjoint, we have that $\obj(S, \pi^*)_{L\times L} + \obj(S, \pi^*)_{R \times R} \le \obj(S, \pi^*)$. So, we conclude that
    \begin{align}
    \label{eqn:inside-partition-bound}
        \obj(\pi^p)_{L \times L} + \obj(\pi^p)_{R \times R} 
        & \le c_1 \cdot \obj(\pi^*)_{L \times L} + c_1 \cdot \obj(\pi^*)_{R \times R} \nonumber \\
        & \le c_1 \cdot \obj(\pi^*) .
    \end{align}

    By combining these two bounds, we upper bound the objective cost of $\pi^p$. From Equation~\ref{eqn:obj-decomposition},

    \begin{align*}
        \obj(\pi^p) & = \obj(\pi^p)_{L \times R} + \obj(\pi^p)_{L \times L} + \obj(\pi^p)_{R \times R} \\
        & \le c_2 \cdot \obj(\pi^*) + \obj(\pi^p)_{L \times L} + \obj(\pi^p)_{R \times R} && \text{(by Equation~\ref{eqn:crossing-cost-bound})} \\
        & \le c_2 \cdot \obj(\pi^*) + c_1 \cdot \obj(\pi^*) && \text{(by Equation~\ref{eqn:inside-partition-bound})} \\
        & = (c_1 + c_2) \cdot \opt(S) .
    \end{align*}
\end{proof}

\begin{theorem}\cite{mathieu2009rank}
\label{thm:schudy}
There is a randomized algorithm for the rank aggregation problem that, given any $\epsilon > 0$ and $n$ rankings on $d$ candidates, outputs a ranking with the cost at most $(1+\epsilon)\opt$ in time $O(\frac{1}{\epsilon}d^3\log d) + d2^{\tilde{O}(\epsilon^{-6})} + O(nd^2)$ with high probability.
\end{theorem}

Now, we are ready to prove our main result (Theorem~\ref{thm:main}).

\begin{proof}[Proof of Theorem~\ref{thm:main}]
     From Theorem~\ref{thm:bi-partition}, we have an algorithm for the colorful bi-partition problem with approximation factor $c_2 = 1$ and running time $O(d^2)$. From Theorem~\ref{thm:schudy} we have an algorithm for rank aggregation with approximation factor $c_1 = 1 + \epsilon$ for any $\epsilon > 0$ and running time $O(\frac{1}{\epsilon}d^3\log d) + d2^{\tilde{O}(\epsilon^{-6})} + O(nd^2)$. Theorem~\ref{thm:main} now follows directly from Theorem~\ref{thm:fair-approx-ratio}.  
     
    We remark that we can derandomize our algorithm by allowing an extra $d^{\tilde{O}(\epsilon^{-12})}$ additive factor in the running time, due to the current best (deterministic) PTAS for the rank aggregation problem~\cite{mathieu2009rank}.
\end{proof}

\section{Improved Fair Rank Aggregation using Closest Fair Ranking}
\label{sec:improvedfra}

In this section, we describe a generic algorithm for the fair rank aggregation under the Kendall-tau metric that works \emph{irrespective} of the definition of fairness under consideration. The only thing we need to have is an efficient procedure to solve the closest fair ranking problem. 

Given a ranking $\pi \in \S_d$, the \emph{closest fair ranking problem} asks to find a fair ranking $\sigma \in \S_d$ that minimizes the Kendall-tau distance $\kappa(\pi, \sigma)$. For a host of fairness notions for which we are already aware of efficient closest fair ranking algorithms, our generic algorithm immediately provides a $2.881$-approximation to the corresponding fair rank aggregation problem, breaking below the only known straightforward 3-approximation guarantee.

\begin{theorem}
\label{thm:below3approx}
    Suppose there is a $t(d)$-time algorithm $\mathcal{A}$ that solves the closest fair ranking problem. Then, there exists a $2.881$-approximation algorithm for the fair rank aggregation problem with running time $O(n^3 d^3 \log d + n^3 t(d) + n^4 d \log d)$.
\end{theorem}
The running time can be improved significantly, more specifically, the dependency on $n$ can be brought down to (near-)linear, using random sampling and \emph{coreset} construction (as detailed in Section~\ref{sec:faster-generic-algo}).

\paragraph{Implications to stricter fairness notions. }

Stronger fairness notions than that of Definition~\ref{def:fair-ranking} have been studied in the context of fair rank aggregation, such as $(\bar{\alpha}, \bar{\beta})$-block-$k$-fairness (see the full version of~\cite{chakraborty2022}).

\begin{definition}($(\bar{\alpha}, \bar{\beta})$-block-$k$-fair ranking)
\label{def:blockfairranking}
    Consider a partition of $d$ candidates into $g$ groups $G_1,\cdots, G_g$. For each group $G_i$ ($i \in [g]$), let us consider two parameters $\alpha_i,\beta_i \in [0,1] \cap \mathbb{Q}$. For $\bar{\alpha} = \left( \alpha_1,\cdots,\alpha_g \right)$, $\bar{\beta} = \left( \beta_1,\cdots,\beta_g \right)$, $k \in [d]$, and a given block size $b$ (such that $\forall_{i\in [g]},\;b\cdot \alpha_i, b \cdot \beta_i$ are integers), a ranking $\pi$ (on $d$ candidates) is said to be \emph{$(\bar{\alpha},\bar{\beta})$-block-$k$-fair} if for every $p$ such that $p \ge k$, and $p \mod b \equiv 0$, for every group $G_i$: The top-$p$ positions $\pi(1),\ldots, \pi(p)$ contain at least $\alpha_i \cdot p$ and at most $\beta_i \cdot p$ candidates from $G_i$.
\end{definition}

Further, they provide an $O(d^2)$-time algorithm that finds a closest $(\bar{\alpha}, \bar{\beta})$-block-$k$-fair ranking. Therefore, as an immediate corollary of Theorem~\ref{thm:below3approx}, we obtain a $2.881$-approximation algorithm for $(\bar{\alpha}, \bar{\beta})$-block-$k$-fair rank aggregation in time $O(n^3 d^3 \log d + n^4 d \log d)$ (the running time can be reduced using randomization), improving upon previously known 3-approximation under this stricter fairness notion.

\paragraph{Description of the algorithm. }Suppose we are given a set $S = \{\pi_1,\cdots,\pi_n\}$. Initialize $L = \emptyset$. For each $\pi_i$ in $S$, find a closest fair ranking $\sigma_i$ using $\mathcal{A}$, and add $\sigma_i$ to the set $L$. 

Iterate through all distinct 3-tuples $T := (\pi_i, \pi_j, \pi_k)$ from $S$. For each such tuple, construct an unweighted directed tournament graph $G_T$ over $[d]$ vertices as follows: For every pair of elements $(a, b)$, add the edge $(a, b)$ if at least two rankings in $T$ order $a$ before $b$; otherwise, add the edge $(b, a)$. Note,~\cite{mathieu2009rank} proposed an algorithm that finds a $(1 + \gamma)$-approximation (for any $\gamma > 0$) to the feedback arc set problem on tournaments. Run this algorithm with $\gamma = 0.00001$ on $G_T$ and delete the set of edges output by the algorithm to obtain a directed acyclic graph $\tilde{G}_T$. Let $\tilde{\pi}_T$ be the ranking obtained by taking the topological ordering of $\tilde{G}_T$. Next, use $\mathcal{A}$ to find a closest fair ranking $\tilde{\sigma}_T$ to $\tilde{\pi}_T$, and add it to $L$. Finally, output a ranking from $L$ that minimizes the objective value (sum of distances to the input rankings).

Below, we provide the pseudocode (Algorithm~\ref{alg:blackbox-fair}) of our generic fair rank aggregation algorithm.

\begin{algorithm}[h!]
\caption{\textsc{Generic Fair Rank Aggregation Algorithm}}
\label{alg:blackbox-fair}
\begin{algorithmic}[1]
\Procedure{GenericFairRA}{$S=\{\pi_1,\cdots,\pi_n\}$}
    \State Initialize an empty set $L$
    \For{$i \gets 1$ to $n$}
        \State $\sigma_i \gets $ A closest fair ranking to $\pi_i$.
        \State Set $L \gets L \cup \{\sigma_i\}$.
    \EndFor
    \For{$i \gets 1$ to $n$}
        \For{$j \gets i+1$ to $n$}
            \For{$k \gets j+1$ to $n$}
                \State Create an (unweighted) tournament graph $G_T \gets ([d], E)$ with $d$ vertices where
                \State \qquad $E = \{(a, b) \mid a \prec b \text{ in at least two rankings of $\{\pi_i, \pi_j, \pi_k\}$} \}$
                \State $E' \gets$ An 1.00001-approximate feedback arc set of $G_T$.
                \State $\tilde{G}_T \gets ([d], E \setminus E')$
                \State $\tilde{\pi}_T \gets $ Ranking obtained from topological ordering of $\tilde{G}_T$.
                \State $\tilde{\sigma}_T \gets $ A closest fair ranking to $\tilde{\pi}_T$.
                \State Set $L \gets L \cup \{\tilde{\sigma}_T\}$.
            \EndFor
        \EndFor
    \EndFor
    \\
    \Return $\arg \min_{\sigma \in L} \obj(S, \sigma)$.
\EndProcedure
\end{algorithmic}
\end{algorithm}

\paragraph{Running time}

There are $O(n^3)$ sets of three rankings of $S$. For each, it takes $O(d^2)$ time to construct the graph $G_T$ as there will be $O(d^2)$ edges. Applying an approximation algorithm to find a $1.00001$-approximate feedback arc set of $G_T$ runs in time $O(d^3\log d)$. Finding the closest fair ranking to $\tilde{\pi}_T$ takes time $t(d)$, and we need to do so for $O(n^3)$ many rankings. Computing the objective value $\obj(S,\sigma)$ of a ranking $\sigma$ can be done in time $O(n d \log d)$. Therefore, finding the ranking in $L$ that minimizes the objective value takes time $O(n^4 d \log d)$. So, overall, the runtime of the algorithm is $O(n^3 d^3 \log d + n^3 t(d) + n^4 d \log d)$.

\paragraph{Analysis of the algorithm} Our analysis at a very high level possesses certain similarities to that used in the \emph{Ulam median} problem in~\cite{chakraborty2023clustering}. However, there are two stark differences. First, there is no direct relation between the Ulam and Kendall-tau distance. In fact, there are rankings in $\S_d$ with both the Ulam and Kendall-tau distances being constant (or equal); on the other hand, there are rankings with Ulam distance one but Kendall-tau distance as large as $\Theta(d)$. Thus, any result or technique known for the Ulam distance does not immediately provide anything for the Kendall-tau distance, which is the metric we consider for fair rank aggregation in this paper. Second,~\cite{chakraborty2023clustering} only considers the Ulam median problem without any fairness constraint, and thus, incorporating fairness would be a major challenge. We need to circumvent both of these challenges in the analysis, and thus, the argument used below is significantly different than that in~\cite{chakraborty2023clustering}, albeit possessing similarities at the level of the proof framework.

Throughout this section, let $\sigma^*$ be an (arbitrary) optimal fair aggregate ranking. For any $\pi_i \in S$, let $\sigma_i$ be a closest fair ranking to $\pi_i$. Further, for any $\pi_i$, define the following set $I_i := \{(a, b) \mid a \prec_{\pi_i} b, b \prec_{\sigma^*} a \}$. Observe that $|I_i| = \kappa(\pi_i, \sigma^*)$. Let $\alpha, \beta \in [0, 1]$, $c \ge 1$ be parameters, the exact values of which we will set later. We emphasize that these parameters are used solely for the purpose of the analysis and have no role in the algorithm.

Recall that we write the median objective value of a ranking $\sigma$ as $\obj(S, \sigma) = \sum_{\pi_i \in S} \kappa(\sigma, \pi_i)$, and that $\opt = \obj(S, \sigma^*)$. Let us also define $\avg := \opt/n$.

For the analysis, we assume the following holds

\begin{equation}
\label{eqn:case1assumption}
    \forall \pi_i \in S, |I_i| > (1-\beta) \avg .
\end{equation}

Otherwise, we show that the objective value of the ranking $\sigma_i$ is strictly less than $3 \opt$. 

\begin{lemma}
\label{lemma:case1}
    Suppose there exists $\pi_i \in S$ such that $|I_i| \le (1 - \beta) \avg$. Then $\obj(\sigma_i) \le (3 - 2\beta) \opt$.
\end{lemma}

\begin{proof}
    Recall that $\sigma_i$ is a closest fair ranking to $\pi_i$. Consider the objective value of $\sigma_i$,

    \begin{align*}
        \obj(\sigma_i) & = \sum_{\pi_j \in S} \kappa(\sigma_i, \pi_j) \\
        & \le \sum_{\pi_j \in S} \left( \kappa(\sigma_i, \sigma^*) + \kappa(\sigma^*, \pi_j) \right) && \text{(By triangle inequality)} \\
        & = \sum_{\pi_j \in S} \kappa(\sigma_i, \sigma^*) + \sum_{\pi_j \in S} \kappa(\sigma^*, \pi_j) \\
        & = \sum_{\pi_j \in S} \kappa(\sigma_i, \sigma^*) + \opt &&\text{(By definition of $\opt$)} \\
        & \le \sum_{\pi_j \in S} \left( \kappa(\sigma_i, \pi_i) + \kappa(\pi_i, \sigma^*) \right) + \opt && \text{(By triangle inequality)} \\
        & \le \sum_{\pi_j \in S} \left( \kappa(\sigma^*, \pi_i) + \kappa(\pi_i, \sigma^*) \right) + \opt && \text{(As $\sigma_i$ is a closest fair rank to $\pi_i$)} \\
        & \le 2 \cdot n \cdot |I_i| + \opt && \text{(By definition)} \\
        & \le 2 (1-\beta) \opt + \opt && \text{(Since $|I_i| \le (1 - \beta) \avg$)} \\
        & \le (3 - 2\beta) \opt .
    \end{align*}
\end{proof}

Therefore, we assume (\ref{eqn:case1assumption}) for the remainder of this section.

Now, consider all the sets $T = \{\pi_i, \pi_j, \pi_k\}$ consisting of three distinct rankings from $S$. If there exists some $\{\pi_i, \pi_j, \pi_k\}$ such that the following holds

\begin{equation*}
    \forall r\ne  s \in \{i, j, k\}, |I_r \cap I_s| \le \alpha \avg ,
\end{equation*}

we show that the objective value of the ranking $\hat{\sigma}_T$ will be strictly less than $3\opt$. To build some intuition, see that if we attempt to order pairs corresponding to the majority vote of the three rankings, all pairs in $|I_r \cap I_s|$ will be inverted in order compared to $\sigma^*$. If, in total, there are not too many such pairs from the majority vote, then finding a ranking that follows the majority of the three will give a good approximation.

\begin{lemma}
\label{lemma:case2}
    Suppose there exists $T = \{\pi_i, \pi_j, \pi_k\}$ such that the following holds: $\forall r \ne  s \in \{i, j, k\}, |I_r \cap I_s| \le \alpha \avg$. Then there will be a fair ranking $\tilde{\sigma}_T \in L$ such that $\obj(\tilde{\sigma}_T) \le (1+12.0001\alpha) \opt$.
\end{lemma}

\begin{proof}
    Recall that our algorithm iterates over all possible sets of three rankings from $S$ and, for each, forms an unweighted directed tournament graph. Consider the graph $G_T$ constructed for this set $T = \{\pi_i, \pi_j, \pi_k\}$. Our algorithm finds a $1.00001$-approximation solution of the optimal feedback arc set, then removes the edges to form $\tilde{G}_T$, and finally takes a topological ordering of $\tilde{G}_T$ to obtain a ranking $\tilde{\pi}_T$.

    Let us now bound the size of an optimal feedback arc set of $G_T$. Let $B = (I_i \cap I_j) \cup (I_j \cap I_k) \cup (I_i \cap I_k)$. Suppose for every pair in $B$, we remove the edges between the pair in $G_T$. Consider the remaining edges in the graph. As they are not in $B$, they must then clearly follow the relative order for the pair of elements in $\sigma^*$. Therefore, the resulting graph is acyclic, and so $|B|$ gives an upper bound on the size of an optimal feedback arc set of $G_T$.

    Now, consider the ranking $\tilde{\pi}_T$, obtained from taking a topological ordering of $\tilde{G}_T$. Since the number of pairs in different pairwise order from $\sigma^*$ is in the worst case, all the pairs in $B$ and the $1.00001|B|$ pairs in the output approximate feedback arc set, we get

    \begin{align}
    \label{eqn:case2intermed}
        \kappa(\tilde{\pi}_T, \sigma^*) & \le (1+0.00001)|B| + |B| \nonumber \\
        & \le 2.00001(3 \alpha \avg) && \text{(By the assumption of the lemma)} \nonumber \\
        & \le 6.00003 \alpha \avg
    \end{align}

    Next, Algorithm~\ref{alg:blackbox-fair} finds a fair ranking $\tilde{\sigma}_T$, which is a closest fair ranking to $\tilde{\pi}_T$ and adds this ranking to $L$. We can bound the objective value of $\tilde{\sigma}_T$ as follows, 


    \begin{align*}
        \obj(\tilde{\sigma}_T) &= \sum_{\pi_j \in S} \kappa(\tilde{\sigma}_T, \pi_j) \\
        & \le \sum_{\pi_j \in S} \left( \kappa(\tilde{\sigma}_T, \tilde{\pi}_T) + \kappa(\tilde{\pi}_T, \pi_j) \right) && \text{(By triangle inequality)} \\
        & \le \sum_{\pi_j \in S} \left( \kappa(\sigma^*, \tilde{\pi}_T) + \kappa(\tilde{\pi}_T, \pi_j) \right) && \text{(As $\tilde{\sigma}_T$ is a closest fair rank to $\tilde{\pi}_T$)} \\
        & \le 6.00003 \alpha \opt + \sum_{\pi_j \in S} \kappa(\tilde{\pi}_T, \pi_j) && \text{(By Equation~\ref{eqn:case2intermed})} \\
        & \le 6.00003 \alpha \opt + \sum_{\pi_j \in S} \left( \kappa(\tilde{\pi}_T, \sigma^*) + \kappa(\sigma^*, \pi_j) \right) && \text{(By triangle inequality)} \\
        & \le 6.00003 \alpha \opt + 6.00003 \alpha \opt + \opt && \text{(By Equation~\ref{eqn:case2intermed} and definition of $\opt$)} \\
        & \le (1+12.0001 \alpha) \opt .
    \end{align*}
\end{proof}

For the rest of the section, we additionally assume that for any set $T = \{\pi_i, \pi_j, \pi_k\}$ of three rankings from $S$, the following holds

\begin{equation}
\label{eqn:case2assumption}
    \exists r \ne s \in \{i, j, k\}, |I_r \cap I_s| > \alpha \avg .
\end{equation}

Finally, we show that assuming (\ref{eqn:case1assumption}) and (\ref{eqn:case2assumption}), there will be a ranking in $L$ which is guaranteed to be strictly better than a 3-approximation of the optimal. Without loss of generality, assume that the rankings are indexed in order of non-decreasing size of $I_i$. That is, $|I_1| \le |I_2| \le  ... \le |I_n|$. By averaging, 
\begin{equation}
    \label{eq:size-I1}
    |I_1| \le \avg.
\end{equation}

In the following, let $\ell$ be the smallest integer such that $|I_1 \cap I_\ell| \le \alpha \avg$ (if such an $\ell$ does not exist, then let $\ell = n+1$). Let $S' = \{\pi_2, \pi_3, ..., \pi_{\ell-1}\}$. Thus
\begin{equation}
    \label{eq:intersect-large}
    \forall \pi_r \in S', |I_r \cap I_1| > \alpha \avg.
\end{equation}

For any $\pi_j, \pi_k$, it is straightforward to observe that
\begin{equation}
\label{eqn:case3distance}
    \kappa(\pi_j, \pi_k) = |I_j| + |I_k| - 2 |I_j \cap I_k| .
\end{equation} 
To see the above, observe that pairs in neither $I_j$ nor $I_k$ will not contribute to the distance, and pairs in $I_j \cap I_k$ contribute twice to the term $|I_j| + |I_k|$ but have the same pairwise order in $\pi_j$ and $\pi_k$ and so do not contribute to the distance.

Now, consider
\[
S_1 : = \{\pi_j \mid |I_j \cap I_1| > \alpha \avg\}, \text{ and}
\]
\[
S_\ell : = \{\pi_j \mid |I_j \cap I_\ell| > \alpha \avg\}.
\]

From now on in the analysis, we consider a parameter $c \ge 1$, the exact value of which we will set later.
\begin{lemma}
\label{lemma:case3a}
    Suppose that (\ref{eqn:case1assumption}) and (\ref{eqn:case2assumption}) hold. If $|S_1| \ge n/c$, then $\obj(\sigma_1) \le \left(3 - \frac{2 \alpha}{c}\right) \opt$.
\end{lemma}

\begin{proof}
    Recall that $\sigma_1$ is a closest fair ranking to $\pi_1$. Now, consider the objective value of $\sigma_1$, 

    \begin{align*}
        \obj(\sigma_1) &= \sum_{\pi_j \in S} \kappa(\pi_j, \sigma_1) \\
        & \le \sum_{\pi_j \in S} (\kappa(\pi_j, \pi_1) + \kappa(\pi_1, \sigma_1)) && \text{(By triangle inequality)} \\
        & \le \sum_{\pi_j \in S} (\kappa(\pi_j, \pi_1) + \kappa(\pi_1, \sigma^*)) && \text{(As $\sigma_1$ is a closest fair ranking to $\pi_1$)} \\
        & = \sum_{\pi_j \in S} (\kappa(\pi_j, \pi_1) + |I_1| ) && \text{(By definition of $|I_1|$)} \\
        & = \sum_{\pi_j \in S} \left( |I_j| + |I_1| - 2 |I_j \cap I_1|  \right) + \sum_{\pi_j \in S} |I_1| && \text{(By Equation~\ref{eqn:case3distance})} \\
        & \le \sum_{\pi_j \in S} (2 |I_1| + |I_j|) - \sum_{\pi_j \in S_1} 2 \alpha \avg && \text{(By Equation~\ref{eq:intersect-large})} \\
        & \le 2 \opt + \opt - \sum_{\pi_j \in S_1} 2 \alpha \avg && \text{(As $|I_1| \le \avg$)} \\
        & \le 3 \opt - 2 \alpha \opt \frac{|S_1|}{n} \\
        & \le \left( 3 - \frac{2 \alpha}{c} \right) \opt && \text{(As $|S_1| \ge n/c$)}.
    \end{align*}
    
\end{proof}

Now, when $c$ is too large, the above lemma only guarantees a poor approximation. However, in that case, we argue that $\sigma_\ell$ attains a much better approximation. Recall that $\ell$ is the smallest integer such that $|I_1 \cap I_\ell| \le \alpha \avg$. 

\begin{lemma}
\label{lemma:case3b}
    Suppose that (\ref{eqn:case1assumption}) and (\ref{eqn:case2assumption}) hold, and further $|S_1| < n/c$. Then, 
    \[
    \obj(\sigma_\ell) \le \left (3 + \frac{2 \beta}{c-1} - 2\left(1 - \frac{1}{c}\right) \alpha \right) \opt.
    \]
\end{lemma}

\begin{proof}

    Since $|S_1| \le n/c$ and we assume (\ref{eqn:case2assumption}), it must be that $|S_\ell| \geq 1 - n/c$. Recall that $\sigma_\ell$ is a closest fair ranking to $\pi_\ell$. Let us now bound the objective value of $\sigma_\ell$,

    \begin{align}\label{eq:ell-bound}
        \obj(\sigma_\ell) &= \sum_{\pi_j \in S} \kappa(\pi_j, \sigma_\ell) \nonumber \\
        & \le \sum_{\pi_j \in S} (\kappa(\pi_j, \pi_\ell) + \kappa(\pi_\ell, \sigma_\ell)) && \text{(By triangle inequality)} \nonumber \\
        & \le \sum_{\pi_j \in S} (\kappa(\pi_j, \pi_\ell) + \kappa(\pi_\ell, \sigma^*)) && \text{(As $\sigma_\ell$ is a closest fair ranking to $\pi_\ell$)} \nonumber \\
        & = \sum_{\pi_j \in S} (\kappa(\pi_j, \pi_\ell) + |I_\ell|) && \text{(By definition of $I_\ell$)} \nonumber \\
        & \le \sum_{\pi_j \in S} (2 |I_\ell| + |I_j|) - \sum_{\pi_j \in S_\ell} 2 \alpha \avg.
    \end{align}

    We now establish an upper bound on the size of $I_\ell$. In particular, we argue that
    \begin{equation}
        \label{eq:ell-size}
        \left |I_\ell \right | \le \left(1 + \frac{\beta}{c-1}\right) \avg.
    \end{equation}
    
    For the sake of contradiction, assume that $|I_\ell| > (1 + \lambda) \avg$, for $\lambda = \beta/(c-1)$. Recall, $S' = \{\pi_2, \pi_3, ..., \pi_{\ell-1}\}$. Then, we get
    \begin{align*}
        \opt &= \sum_{\pi_j \in S} \kappa(\sigma^*, \pi_j) \\
        & > (1 - \beta) \avg |S'| + (1 + \lambda) \avg |S \setminus S'| && \text{(As we assume (\ref{eqn:case1assumption}))} \\
        & = \opt + (\lambda |S \setminus S'| - \beta |S'|) \avg \\
        & = \opt + \lambda n - (\lambda + \beta) |S'| \avg
    \end{align*}
    Now, since $|S_1| < n/c$ and by the definition $S' \subseteq S_1$, $|S'| < n/c$. Then, for $\lambda = \beta/(c-1)$,
    \[
    \lambda n - (\lambda + \beta) |S'| > \frac{\beta}{c-1} \cdot n - \beta \cdot \frac{c}{c-1} \cdot \frac{n}{c} \ge 0,
    \]
    leading to a contradiction.

    Next, by our obtained bound on the objective value of $\sigma_\ell$ in Equation~\ref{eq:ell-bound},

    \begin{align*}
        \obj(\sigma_\ell) &\le \sum_{\pi_j \in S} (2 |I_\ell| + |I_j|) - \sum_{\pi_j \in S_\ell} 2 \alpha \avg \\
        & \le 2 \left(1 + \frac{\beta}{c-1}\right) \opt + \opt - \sum_{\pi_j \in S_\ell} 2 \alpha \avg && \text{(By Equation~\ref{eq:ell-size})}\\
        & \le 2\left(1 + \frac{\beta}{c-1}\right) \opt + \opt - 2\left(1 - \frac{1}{c} \right) \alpha \opt && \text{(As we assumed $|S_\ell| \ge \left(1 - \frac{1}{c}\right)n$)} \\
        & \le \left(3 + \frac{2 \beta}{c-1} - 2\left(1 - \frac{1}{c} \right) \alpha \right) \opt .
    \end{align*}
\end{proof}

Now, we conclude the proof of Theorem~\ref{thm:below3approx}.

\begin{proof}[Proof of Theorem~\ref{thm:below3approx}]
    By combining Lemmas~\ref{lemma:case1},~\ref{lemma:case2},~\ref{lemma:case3a},~\ref{lemma:case3b}, we deduce that $L$ must contain a fair ranking which has objective value at most the maximum of: $(3-2\beta) \opt$, $(1+12.0001\alpha) \opt$, $\left(3 - \frac{2 \alpha}{c}\right) \opt$, and $\left(3 + \frac{2 \beta}{c-1} - 2\left(1-\frac{1}{c}\right)\alpha\right) \opt$.

    By setting $\alpha = 0.1567$, $\beta = 0.0598$, and $c = 2.62$, we guarantee that in the worst case, at least one of the fair rankings in $L$ is guaranteed to have objective value at most $2.881 \opt$.
\end{proof}

\paragraph{Approximate fair aggregate ranking using an approximately close fair ranking algorithm. }

Suppose that we only have an $f$-approximate algorithm to find a closest fair ranking. Let the $f$-approximate closest fair ranking be applied with input ranking $\pi_i$ to obtain a $f$-approximate closest fair rank $\sigma_i$. Then the guarantee is that $\kappa(\pi_i, \sigma_i) \le f \kappa(\pi_i, \sigma^*)$. Therefore, the analysis of the above lemmas will need to be modified to account for this additional factor $f$ whenever the closest fair ranking algorithm is applied. This leads to the approximation factors as follows: Lemma~\ref{lemma:case1} gives $(f + 2 - (f+1)\beta) \opt$, Lemma~\ref{lemma:case2} gives $(1 + 6.00003f\alpha + 6.00003\alpha)\opt$, Lemma~\ref{lemma:case3a} gives $\left(f+2 - \frac{2 \alpha}{c} \right) \opt$, Lemma~\ref{lemma:case3b} gives $\left(f+2 +(f+1)\frac{\beta}{c-1}  - 2(1 - \frac{1}{c})\alpha\right) \opt$. For instance, using a 2-approximate closest fair ranking, with the parameters $\alpha = 0.1599$, $\beta = 0.04068$, and $c = 2.62$, our algorithm finds a $3.878$-approximate fair aggregate ranking, improving on the simple 4-approximation derived from~\cite{wei2022rank,chakraborty2022}.

\section{A faster generic fair rank aggregation algorithm}{\label{sec:faster-generic-algo}}

In this section, we discuss how to improve the running time (to bring down the dependency on the number of inputs from cubic to near-linear) of the algorithm discussed in Section~\ref{sec:improvedfra} while incurring only a small additional increase in the approximation factor. Let us first recall certain details of the algorithm. The algorithm iterates through all possible sets of three rankings of $S$ and, for each such set, finds a candidate ranking and adds it to $L$. Therefore, our algorithm has to find an approximate feedback arc set $O(n^3)$ times, and moreover, there are $O(n^3)$ rankings in $L$. Another bottleneck is the computation of the objective value for a ranking, which takes time $O(nd \log d)$ due to needing to compute the cost with respect to the whole input $S$.

\begin{theorem}
\label{thm:samplingfra}
    Suppose there is a $t(d)$-time algorithm $\mathcal{A}$ that solves the closest fair ranking problem. Then, there exists a randomized $2.888$-approximation algorithm for the fair rank aggregation problem with running time $O(d^3 \log^3 n \log d + n t(d) + n d \log n \log d)$. The algorithm succeeds with probability at least $1 - \frac{1}{n}$.
\end{theorem}

Before we describe the modifications, we first introduce an important tool needed for the algorithm, namely \emph{coreset}.

\begin{definition}[$\gamma$-coreset]
For a set $S$ of points in an arbitrary metric space $\mathcal{X}$ with the associated distance function $\dist(\cdot,\cdot)$, and an implicit set $X \subseteq \mathcal{X}$ (of potential medians), a weighted subset $P \subseteq S$ (with a weight function $w : P \rightarrow \mathbb{R}$) is a \emph{$\gamma$-coreset} (for any $\gamma >0$) of $S$ with respect to $X$ for the ($1$-)median problem if for any $x \in X$,
\[ (1 - \gamma)\sum_{y \in S} \dist(y,x) \le \sum_{p \in P} w(p) \cdot \dist(p,x) \le (1+\gamma)\sum_{y \in S} \dist(y,x).\]
\end{definition}
In simple words, the above-weighted subset $P$ is sufficient to well-approximate the objective value of any potential median point. There are many constructions for coresets in the literature. For our algorithm, we consider the following coreset construction by~\cite{feldman2011unified,bachem2018one,braverman2021coresets}.

\begin{theorem}[~\cite{feldman2011unified,bachem2018one,braverman2021coresets}]
\label{thm:coreset}
    There is an algorithm that, given a set $S$ of $n$ points of an arbitrary metric space $\mathcal{X}$ and an implicit set $X \subseteq \mathcal{X}$ (without loss of generality, assume $S \subseteq X$) and $\gamma > 0$, outputs a $\gamma$-coreset of $S$ with respect to $X$ for the median problem, of size $O(\gamma^{-2} \log |X|)$. The algorithm succeeds with probability at least $1 - \frac{1}{n^2}$ and runs in time $O(\gamma^{-2} n \log |X|)$.
\end{theorem}
It is worth noting that the original result is for a more general \emph{$k$-median} problem; however, since we are only interested in the (1-)median problem, we state the theorem accordingly. Here, we are interested in solving the median problem over the Kendall-tau metric.

\paragraph{Description of the algorithm}
The algorithm is similar to the deterministic one described in the previous section, but with only two major changes. So, we only highlight those two modifications below.

First, each ranking in $S$ is independently sampled to be included in a set $\bar{S}$ with probability $\frac{4s \log n}{n}$, for some constant $s > 1$, the value of which we will fix later. Then, we only iterate through all possible sets of three rankings of $\bar{S}$ (instead of $S$). Note that for each input ranking $\pi_i \in S$, we still find a closest fair ranking $\sigma_i$ and add it to $L$.

Second, we additionally construct a 1.0001-coreset $P$ of $S$ (with respect to the implicit potential median set $L$) using the algorithm of Theorem~\ref{thm:coreset}. Our modified algorithm finally outputs a ranking $\sigma$ that minimizes the weighted sum of distances to the coreset $P$, that is, $ \arg \min_{\sigma \in L} \sum_{\pi \in P} w(\pi) \cdot \kappa(\pi, \sigma)$. 

\paragraph{Analysis}

Before analyzing the above-modified version of our algorithm, let us recall some definitions. For any ranking $\pi_i \in S$, we define $I_i := \{(a, b) \mid a \prec_{\pi_i} b, b \prec_{\sigma^*} a \}$. Further, without loss of generality, we assume that the rankings of $S$ are indexed in order of non-decreasing $|I_i|$, that is, $|I_1| \le |I_2| \le ... \le |I_n|$. Finally, for any ranking $\pi_i \in S$, we define the following set
\[
S_i : = \{\pi_j \mid |I_i \cap I_j| > \alpha \avg\}.
\]

As before, let $\alpha, \beta \in [0, 1], c \ge 1$ be parameters, the exact values of which we will set later. Again, we emphasize that these values are used solely for the purpose of analyzing the algorithm.

As the algorithm still finds a closest fair ranking to each input ranking and adds it to $L$, Lemma~\ref{lemma:case1} still applies. We therefore assume in the remainder of the analysis that
\[\forall \pi_i \in S, |I_i| > (1-\beta) \avg .\]

We first show that if there is a $\pi_k$ among the first $\frac{n}{s}$ many rankings, which meets the condition that $|S_k| \ge \frac{n}{c}$, $\sigma_k$ will give a good approximation.
\begin{lemma}
\label{lemma:samplingone}
    Suppose that there exists a $\pi_k \in S$, for $k \le \frac{n}{s}$, such that $|S_k| \ge \frac{n}{c}$. Then 
    \[
    \obj(\sigma_k) \le \left(3 + \frac{2 \beta}{s - 1} - \frac{2 \alpha}{c}\right) \opt.
    \]
\end{lemma}
\begin{proof}
    From using the same argument as in the proof of Lemma~\ref{lemma:case3a}, we have the following

    \begin{equation*}
        \obj(\sigma_k) \le \sum_{\pi_j \in S} (2 |I_k| + |I_j|) - \sum_{\pi_j \in S_k} 2 \alpha \avg .
    \end{equation*}
    
    Now, from a similar argument as in the proof of Lemma~\ref{lemma:case3b}, we argue that $|I_k| \le \left(1 + \frac{\beta}{s - 1}\right) \avg$, as there are at most $\frac{n}{s}$ rankings before $\pi_k$.
    
    Therefore, 
    \begin{align*}
        \obj(\sigma_k) & \le \sum_{\pi_j \in S} (2 |I_k| + |I_j|) - \sum_{\pi_j \in S_k} 2 \alpha \avg \\
        & \le \left(3 + \frac{2 \beta}{s - 1} - \frac{2 \alpha}{c} \right) \opt .
    \end{align*}
\end{proof}

As such, from now on, we assume otherwise and that the following holds

\begin{equation}
\label{eqn:assumption1sampling}
    \forall k \le \frac{n}{s}, |S_k| < \frac{n}{c} .
\end{equation}

Now, for each $\pi_k$ ($k \le \frac{n}{s}$), let us consider the next $\frac{n}{c} + \frac{n}{s}$ many rankings. Define the following set 
\[ C_k := \{\pi_j \mid k < j \le k+\frac{n}{c} + \frac{n}{s} \text{ and } |I_k \cap I_j| \le \alpha \opt \} .\]

Since $|S_k| < \frac{n}{c}$, it must hold that $|C_k| \ge \frac{n}{s}$. 
For a ranking $\pi_\ell \in C_k$, additionally define the set 
\[F_{k, \ell} := \{\pi_j \mid |I_j \cap I_k| \le \alpha \opt \text{ and } |I_j \cap I_\ell| \le \alpha \opt \} . \] 

Now, fix any $k \le \frac{n}{s}$ and consider any $\pi_\ell \in C_k$. If $|F_{k, \ell}| \le \frac{n}{s}$, then we show that $\sigma_\ell$ will give a good approximation.

\begin{lemma}
\label{lemma:samplingtwo}
    Consider any $\pi_k$ where $k \le \frac{n}{s}$, and any $\pi_\ell \in C_k$. If $|F_{k, \ell}| < \frac{n}{s}$, then it holds that 
    \[
    \obj(\sigma_\ell) \le \left(3 + 2 \beta \frac{2c + s}{cs - s - 2c} - 2 \left(1 - \frac{1}{c} - \frac{1}{s} \right) \alpha \right) \opt.
    \]
\end{lemma}
\begin{proof}
    If $|F_{k, \ell}| \le \frac{n}{s}$, then $|S_k| + |S_\ell| \ge n - \frac{n}{s}$.
    
    As we have assumed that $|S_k| \le \frac{n}{c}$, we must have that $|S_\ell| \ge n - \frac{n}{s} - \frac{n}{c}$. Using a similar argument as in the proof of Lemma~\ref{lemma:case3b}, 

    \begin{equation*}
        \obj(\sigma_\ell) \le \sum_{\pi_j \in S} (2 |I_\ell| + |I_j|) - \sum_{\pi_j \in S_\ell} 2 \alpha \avg.
    \end{equation*}

    Again, we need to bound $|I_\ell|$, and as in the proof of Lemma~\ref{lemma:case3b}, assume that $|I_\ell| \ge (1 + \lambda) \avg$ where $\lambda = \beta \frac{2c + s}{c s - s - 2c}$. Observe that there are at most $\frac{n}{c} + \frac{2 n}{s}$ many rankings before $\pi_\ell$. Therefore,

    \begin{equation*}
        \opt > (1 - \beta) \avg \left(\frac{n}{c} + \frac{2 n}{s} \right) + (1 + \lambda) \avg \left(n - \frac{n}{c} - \frac{2 n}{s} \right)
    \end{equation*}

    and thus for $\lambda = \beta \frac{2c + s}{c s - s - 2c}$ we get a contradiction. Therefore, $|I_\ell| \le \left(1 + \beta \frac{2c + s}{c s - s - 2c}\right) \avg$.

    Hence, we have the following upper bound

    \begin{align*}
        \obj(\sigma_\ell) & \le \sum_{\pi_j \in S} (2 |I_\ell| + |I_j|) - \sum_{\pi_j \in S_\ell} 2 \alpha \avg. \\
        & \le \left(3 + 2 \beta \frac{2c + s}{c s - s - 2c} - 2\left(1 - \frac{1}{c} - \frac{1}{s} \right) \alpha \right) \opt .
    \end{align*}
\end{proof}

As such, from now on, we additionally assume the following
\begin{equation}
\label{eqn:assumption2sampling}
    \forall k \le \frac{n}{s}, \forall \pi_\ell \in C_k, |F_{k, \ell}| \ge \frac{n}{s} .
\end{equation}

Now, we show that given these assumptions, with high probability, there exist three rankings in $\bar{S}$, which give a good approximation.

\begin{lemma}
\label{lemma:samplingsbar}
    Suppose (\ref{eqn:assumption1sampling}) and (\ref{eqn:assumption2sampling}) hold. Then there exist three rankings, $T = \{\pi_k, \pi_\ell, \pi_m\}$ in $\bar{S}$ such that $\obj(\tilde{\sigma}_T) \le (1 + 12.0001 \alpha) \opt$, with probability at least $1 - \frac{1}{n^2}$.
\end{lemma}
\begin{proof}

    Fix a $\pi_k$ such that $k \le \frac{n}{s}$, and as we assume (\ref{eqn:assumption1sampling}), $|C_k| \ge \frac{n}{s}$ (as we have already argued). Fix a $\pi_\ell \in C_k$, and as we assume (\ref{eqn:assumption2sampling}), $|F_{k, \ell}| \ge \frac{n}{s}$. 

    Recall that we independently sample each ranking for inclusion to $\bar{S}$ with probability $\frac{4s \log n}{n}$. The probability that no ranking in $F_{k, \ell}$ (for this fixed choice of $k$ and $\ell$) is sampled can be bounded as follows

    \begin{align*}
        \Pr(\text{no rankings in } F_{k, \ell} \text{ are sampled}) & \le \left(1 - \frac{4s \log n}{n} \right)^{n/s} \\
        & \le e^{-4 \log n} \\
        & \le n^{-4} .
    \end{align*}

   Now, fix some ranking $\pi_k$ where $k \le \frac{n}{s}$. Let $H_k$ be the event where either no ranking in $C_k$ is sampled, or there is some $\pi_\ell \in C_k$ where no rankings in $F_{k, \ell}$ are sampled. Let us bound the probability of $H_k$ occurring,

   \begin{align*}
       \Pr(H_k) & \le \left(1 - \frac{4s \log n}{n} \right)^{n/s} + \frac{n}{s} \cdot n^{-4} && \text{(By applying union bound)} \\
       & \le n^{-4} + \frac{n^{-3}}{s}  \le n^{-3} && \text{(For some $s>1$)}.
   \end{align*}

    Observe that if the event $H_k$ does not occur for all $\pi_k$ where $k \le \frac{n}{s}$, then sampling any one such $\pi_k$ ensures that for the sampled $\pi_k$, there is a sampled $\pi_\ell \in C_k$. Further, for this sampled $\pi_\ell$, there must be a sampled $\pi_m \in F_{k, \ell}$. Therefore, let us now bound the probability of this bad event that either no such $\pi_k$ is sampled or that $H_k$ occurs for at least one such $k$.

    \begin{align*}
        \Pr(\text{no } \pi_k \text{ is sampled, or } H_k \text{ occurs for some } k) &\le \left(1 - \frac{4s \log n}{n} \right)^{n/s} + \frac{n}{s} \cdot n^{-3} && \text{(By applying union bound)} \\
        & \le n^{-4} + \frac{n^{-2}}{s}  \le n^{-2} && \text{(For some $s>1$)}.
    \end{align*}

    Therefore, with probability at least $1 - \frac{1}{n^2}$, we will sample some $\pi_k$ where $k \le \frac{n}{s}$, some $\pi_\ell \in C_k$, and some $\pi_m \in F_{k, \ell}$. By definition of $C_k$ and $F_{k, \ell}$, it holds that $\forall \{r, s\} \in \{k, \ell, m\}, |I_r \cap I_s| \le \alpha \opt$. By applying Lemma~\ref{lemma:case2} on these three rankings, we directly obtain that there exists some ranking $\hat{\sigma}_T$ added to $L$ such that $\obj(\hat{\sigma}_T) \le (1 + 12.0001 \alpha) \opt$.

\end{proof}

Finally, we also claim that with high probability, the size of the sampled set is sufficiently small.

\begin{lemma}
\label{lemma:sampledsetsizes}
    $|\bar{S}| \le 10s \log n$ with probability at least $1 - \frac{1}{n^2}$.
\end{lemma}
\begin{proof}
    Recall that each ranking is independently sampled for inclusion to $\bar{S}$ with probability $\frac{4s \log n}{n}$. Therefore, it is straightforward to see that $\mathbb{E}(|\bar{S}|) = 4s \log n$.

    By applying the standard Chernoff bound, we have

    \begin{equation*}
        \Pr(|\bar{S}| \ge 10s \log n) \le e^{\frac{-9s \log n}{3.5}} \le n^{-2} .
    \end{equation*}
\end{proof}

Now, we are ready to finish the proof of Theorem~\ref{thm:samplingfra}.

\begin{proof}[Proof of Theorem~\ref{thm:samplingfra}]
    By combining Lemmas~\ref{lemma:case1},~\ref{lemma:samplingone},~\ref{lemma:samplingtwo},~\ref{lemma:samplingsbar}, there is some ranking in $L$ which has objective value at most the maximum of: $(3 - 2 \beta) \opt$, $(1 + 12.0001 \alpha) \opt, (3 + \frac{2 \beta}{s - 1} - \frac{2 \alpha}{c}) \opt, \left(3 + 2 \beta \frac{2c + s}{c s - s - 2c} - 2 \left(1 - \frac{1}{c} - \frac{1}{s} \right) \alpha \right) \opt$.

    Fix $s = 50$. Then by setting $\alpha = 0.1572$, $\beta = 0.05652$, $c = 2.7$, we get that there is some ranking $\sigma \in L$ such that $\obj(S, \sigma) \le 2.887 \opt$. However, our algorithm now returns a ranking $\sigma_p$ that minimizes the objective value to the 1.0001-coreset $P$. Therefore, we need to bound the objective cost of $\sigma_p$ with respect to the original set $S$ as follows

    \begin{align*}
        \obj(S, \sigma_p) & \le 1.0001 \sum_{\pi \in P} w(\pi) \cdot \kappa(\sigma_p, \pi)  && \text{(As $P$ is a 1.0001-coreset of $S$)} \\
        & \le 1.0001 \sum_{\pi \in P} w(\pi) \cdot \kappa(\sigma, \pi)  && \text{(As $\sigma_p$ minimizes the sum of weighted distances to $P$)} \\
        & \le 1.0001 \cdot 1.0001 \obj(S, \sigma) && \text{(As $P$ is a 1.0001-coreset of $S$)}\\
        & \le 1.0001 \cdot 1.0001 \cdot 2.887 \opt \\
        & \le 2.888 \opt
    \end{align*}

    and so our algorithm outputs a ranking with an objective value at most $2.888 \opt$.

    By Lemma~\ref{lemma:samplingsbar}, $|\bar{S}| \le 500 \log n$. Therefore, our algorithm will only need to approximately solve the feedback arc set for $O(\log^3 n)$ sets of three rankings. Further, we only compute the weighted sum of distances to the coreset $P$. Now, it remains to provide an upper bound to the size of $P$, for which we use Theorem~\ref{thm:coreset}. For that purpose, let us now specify the implicit set of potential medians. Let us consider the set $L$ constructed in our deterministic algorithm (i.e., Algorithm~\ref{alg:blackbox-fair}), and to remove any ambiguity, let us denote that set as $L_{det}$. Recall, $\left |L_{det} \right | = O(n^3)$. Now, we apply Theorem~\ref{thm:coreset} by considering $L_{det}$ as the implicit set of potential medians (note, the set $L$ constructed by our modified randomized algorithm is always a subset of $L_{det}$), and as a consequence, we get $|P| =O\left(\log \left |L_{det} \right | \right) = O(\log n)$. So, the time taken to compute the weighted sum of distances to the coreset $P$ for a ranking in $L$ is $O(d \log n \log d)$. The time taken to construct the coreset is $O(n \log n)$. Therefore, the overall running time of our modified randomized algorithm is $O(d^3 \log^3 n \log d + n t(d) + n d \log n \log d)$ (if the time exceeds this bound, we declare the algorithm to fail).

    By taking a union bound over the failure probabilities of Theorem~\ref{thm:coreset}, Lemma~\ref{lemma:samplingsbar}, and Lemma~\ref{lemma:sampledsetsizes}, the algorithm succeeds with probability at least $1 - \frac{1}{n}$.
\end{proof}

\section{Experiments}

\begin{figure}[t!]
    \centering
    \includegraphics[width = 0.8\textwidth]{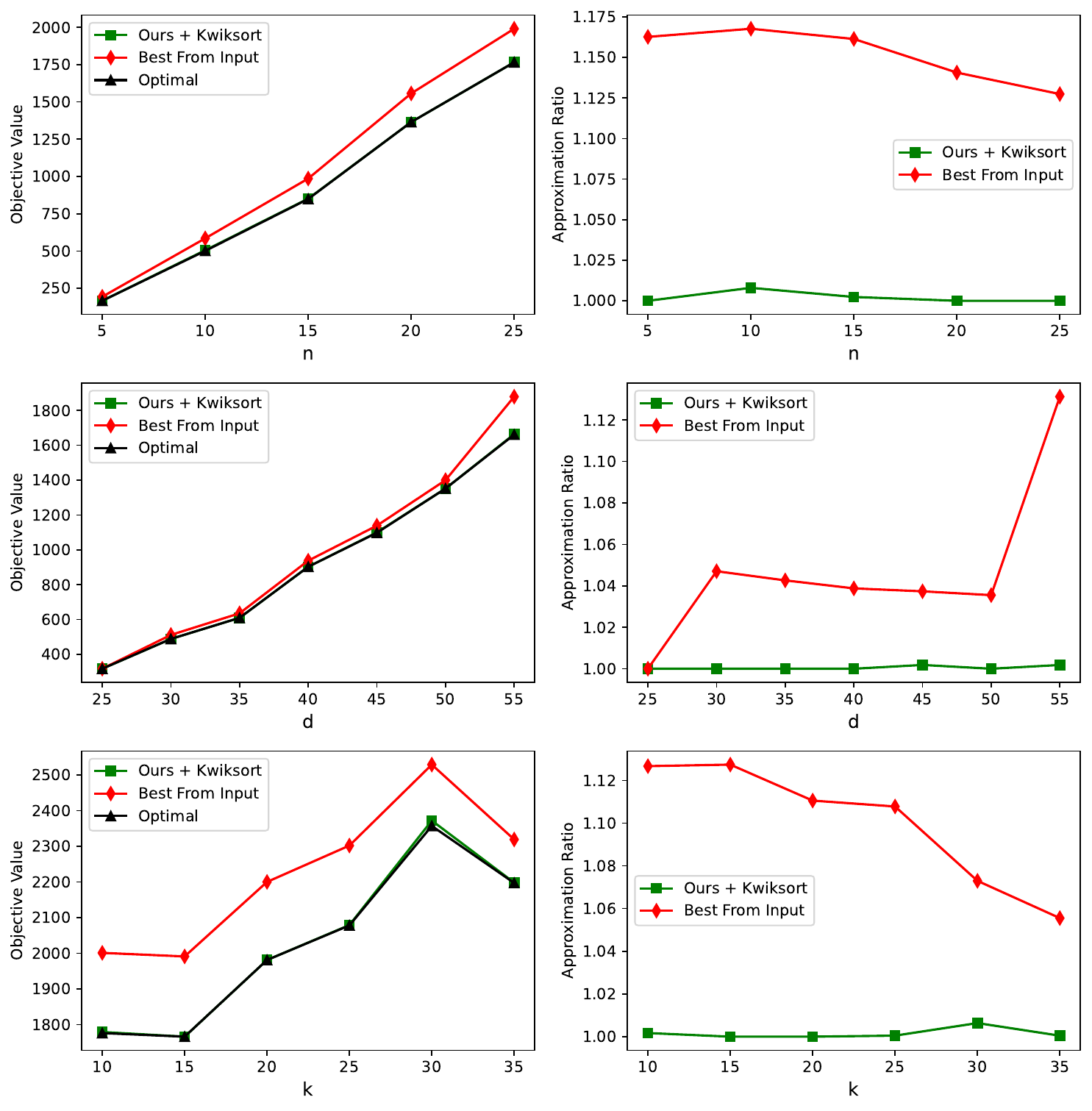}
    \caption{\textbf{Football dataset}. The $x$-axis indicates the value of the parameter ($n$, $d$, or $k$). The $y$-axis indicates the objective value of each output ranking on the left figure, with the corresponding approximation ratio on the right figures.}
    \label{fig:football}
\end{figure}

In this section, we provide an empirical evaluation of our algorithm on real-world datasets. We compare our algorithm against the performance of the best-known algorithms for the fair rank aggregation problem. The algorithms are implemented in Python 3.12. All experiments were performed on a laptop running Windows 11 using a Ryzen 6800HS processor and 16GB of RAM. We also use Integer Linear Programming (ILP) to find the optimal solution, where possible, for comparison. The Integer Linear Programs are implemented with CVXPY~\cite{cvxpy}, using SCIP~\cite{BolusaniEtal2024OO} as the solver. The code is available on Github\footnote{https://github.com/Aussiroth/Improved-Fair-Rank-Aggregation}.

\paragraph{Datasets} The first dataset we use was introduced in previous work studying fairness in rank aggregation~\cite{kuhlman2020rank}. The dataset is taken from a fantasy sports website for American football, where experts provide performance rankings over a set of (real) football players each week. The dataset contains rankings from experts across 16 weeks of the 2019 football season. In each week, the ranking of 25 experts on a set of players is given. We follow their work and divide the players into two groups based on the conference the player's team is in.

The second dataset was introduced in previous work studying proportionate fairness in rank aggregation~\cite{wei2022rank}. The dataset contains the rankings of 7 users over 268 movies. Each movie is placed into groups based on its genre, leading to 8 groups. We also preprocess the dataset to remove some movies to obtain a smaller dataset. This dataset contains movies from the 4 largest genres and has 58 movies. We perform experiments on both datasets.

\paragraph{Algorithms} We implement our algorithm as described in the previous sections. The approximation algorithm used to solve rank aggregation is Kwiksort~\cite{ailon2008}, and with Theorem~\ref{thm:fair-approx-ratio}, this implementation is, in theory, a $18/7$-approximation. The algorithm was chosen as it has good asymptotic runtime and is easy to implement, making it much more practical than the known PTAS.

We implement the best-known algorithm for this problem in previous work as the baseline~\cite{chakraborty2022} (also as in~\cite{wei2022rank}). Recall that the algorithm finds a closest fair ranking for each input ranking and then outputs the one that gives the minimum objective value. We refer to this algorithm as \textsc{BestFromInput}. This is a $3$-approximation algorithm.

For all experiments, the values of $\alpha_i, \beta_i$ are selected to be equal to the proportion of elements belonging to group $i$ in the dataset. This is a natural option, as it maintains the proportion of elements from each group in the top-$k$.

\paragraph{Results}

Figure~\ref{fig:football} evaluates the algorithms for one instance (week 4) of the football dataset. We perform experiments that independently vary the number of input rankings $n$, the number of players $d$, and the value of the parameter $k$. For the experiments that vary $n$ and $d$, we set $k = 15$. We see that our algorithm consistently performs better than \textsc{BestFromInput}, finding a fair aggregate ranking with a lower objective value in all of the instances, and is almost optimal. 

Figure~\ref{fig:movielens} evaluates the algorithms for the full Movielens dataset. We also perform experiments independently varying $n$, $d$, and $k$. For the experiments that vary $n$ and $d$, we fix the parameter $k = 30$. As the number of elements is too large for the ILP to scale to, we only compare the objective value of the ranking that is output by the two algorithms. Our algorithm performs much better than \textsc{BestFromInput} in all experiments.

\begin{figure}[h!]
    \centering
    \includegraphics[width = 0.8\textwidth]{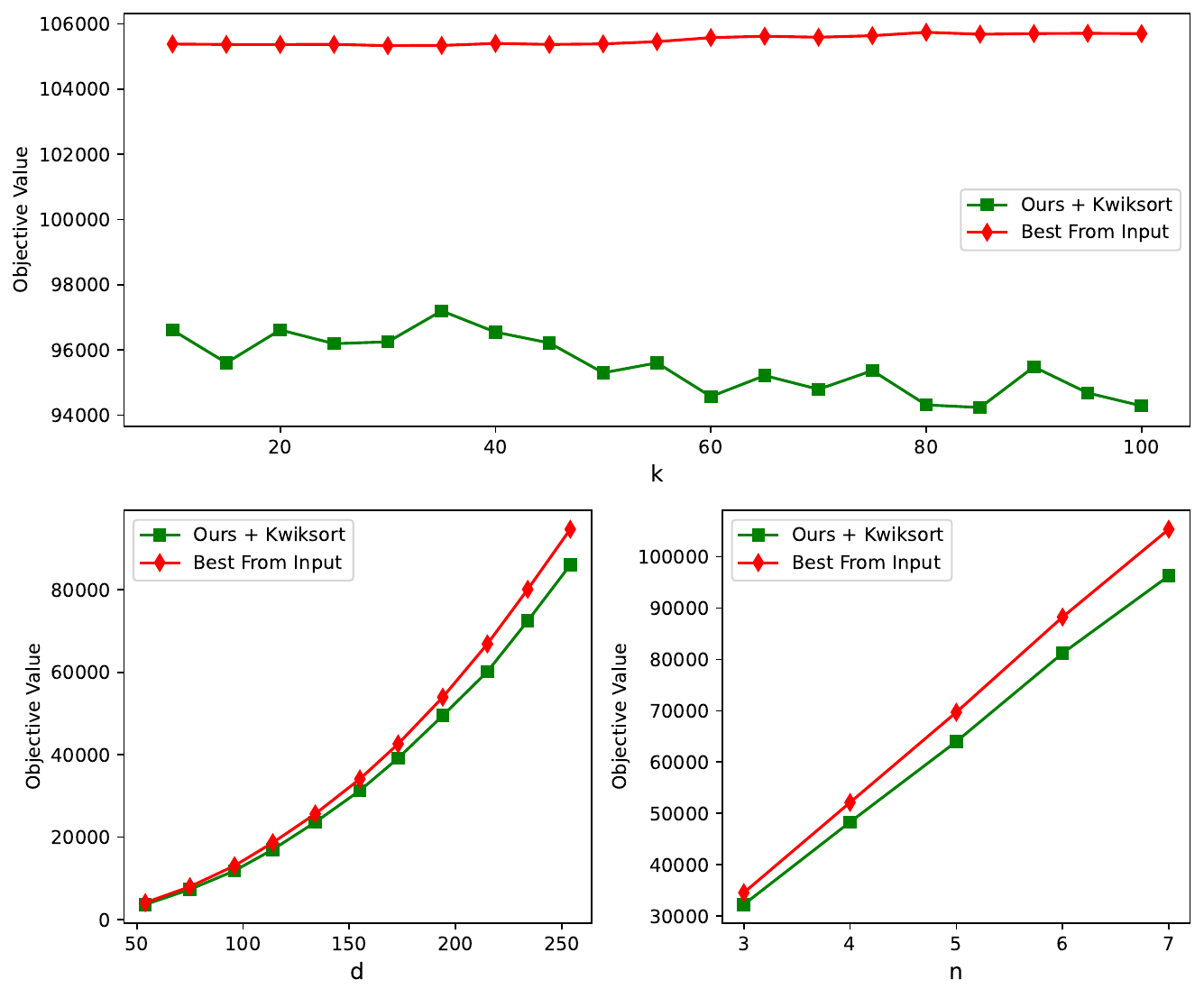}
    \caption{\textbf{Movielens dataset}. The x-axis indicates the value of the parameter ($n$, $d$ or $k$). The y-axis indicates the objective value of the output ranking for each algorithm.}
    \label{fig:movielens}
\end{figure}

Figure~\ref{fig:movielens-reduced} evaluates the algorithms for the reduced Movielens dataset. For this dataset, we focus on experiments for different values of $k$. We observe that in comparison to the football dataset, our algorithm with Kwiksort still performs much better than \textsc{BestFromInput}, but noticeably worse than optimal. We also plot an implementation of our algorithm that uses ILP to solve rank aggregation optimally. This performs better than using Kwiksort and is extremely close to optimal. This suggests that our algorithm is able to select a good set of top-$k$ elements.

\begin{figure}[h!]
    \centering
    \includegraphics[width = 0.8\textwidth]{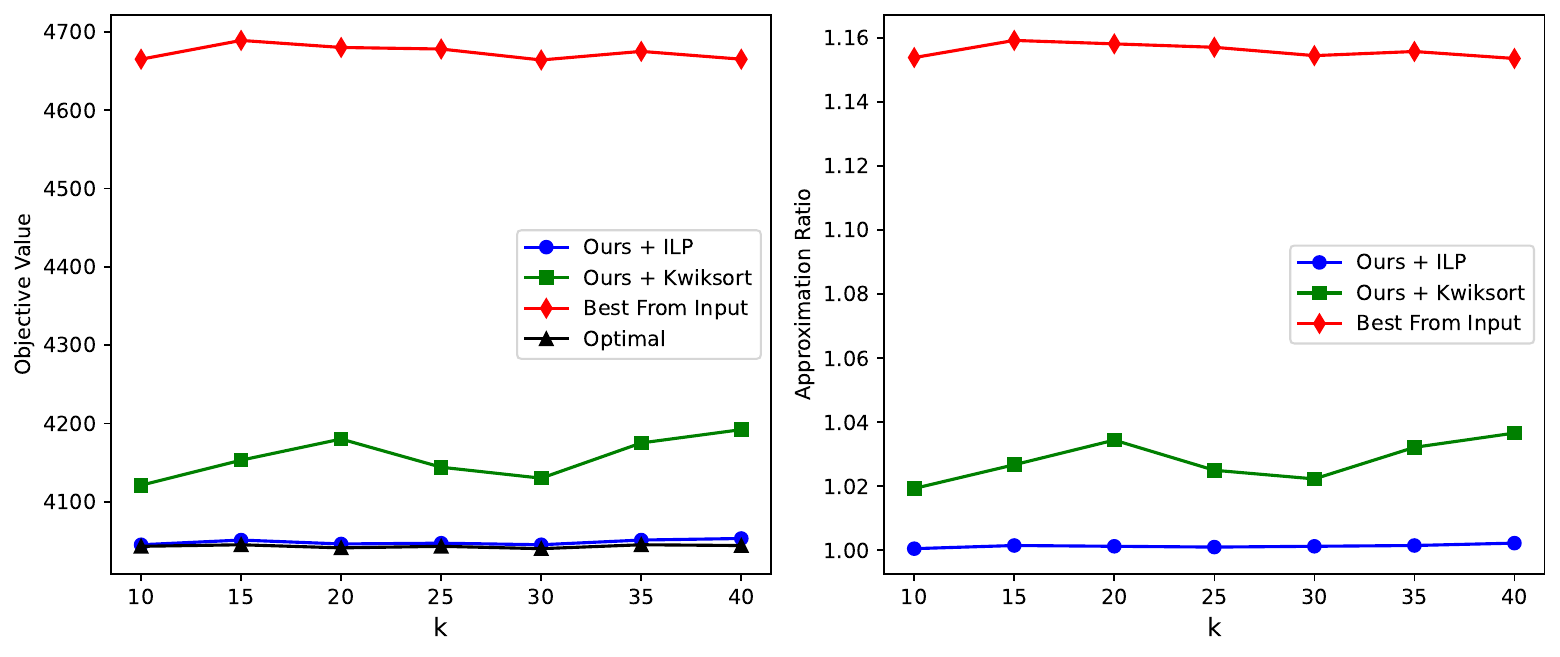}
    \caption{\textbf{Reduced Movielens dataset}. The x-axis indicates the value of the parameter $k$. The y-axis indicates the objective value of each output ranking on the left figure, with the corresponding approximation ratio on the right figure.}
    \label{fig:movielens-reduced}
\end{figure}

We also perform experiments where the values of $\alpha_i, \beta_i$ are varied instead of setting them as the group's proportion. For all experiments varying $n$ and $d$, we explore various values for the parameter $k$. We conduct experiments on 15 other instances (along with Week 4) of the football dataset as well. These experiments show similar results; see the appendix.

\section{Conclusion and Future Work}
In this work, we address the rank aggregation problem under the proportional fairness constraint as introduced in~\cite{wei2022rank,chakraborty2022}. We propose a novel algorithm to return a fair consensus ranking and establish (through theoretical analysis) that it achieves a $(2 + \epsilon)$-approximation for any $\epsilon > 0$, thereby improving upon the current best $3$-factor approximation bound. Our experimental results further show that our algorithm consistently produces nearly optimal fair consensus rankings in practice. We also present a generic $2.881$-approximation algorithm that works irrespective of the fairness definition as long as there is an efficient procedure to compute a closest fair ranking to any input.

An exciting open question is whether the approximation can be further improved, ideally achieving a PTAS that matches the current best approximation guarantee of the classical rank aggregation problem without fairness restrictions. It is noteworthy that a 2-factor is unavoidable for our algorithm (see the appendix), indicating that a fundamentally new approach may be required to enhance the approximation guarantee. Additionally, exploring other specific stricter fairness notions and demonstrating a comparable approximation guarantee (beating the bound obtained by our generic algorithm) for them presents another intriguing research direction.

\bibliographystyle{named} 
\bibliography{sample}

\newpage
\input{arxiv_appendix}
\end{document}

%% file: arxiv_appendix.tex
\begin{appendix}
\section{Tight Instance of our Fair Rank Aggregation Algorithm in Section~\ref{sec:fair-rank}}
In this paper, we provide a new algorithm to compute fair aggregate ranking and establish that our proposed algorithm achieves 2-approximation. In this section, we argue that the 2-factor approximation is not an artifact of our analysis. Instead, there are certain instances for which our algorithm cannot perform much better than the 2-factor.

Now, we describe a family of instances that show the analysis for our algorithm is tight; that is, we cannot achieve better than a 2-approximation. To aid in describing the instance, let us denote the elements in the following way. Consider two integers $s,t$, whose values we will set later. Let $A = \{a_1, \cdots, a_s\}$, $B = \{b_1, \cdots, b_s\}$, $C = \{c_1, \cdots, c_t\}$ and $D = \{d_1 , \cdots, d_t\}$. The universe of elements is $A \cup B \cup C \cup D$. Therefore $d = 2s + 2t$. Let there be two groups, $G_1, G_2$. $G_1 = C \cup D$ and $G_2 = A \cup B$. Let the parameters be $k = d/2$ and the values $\alpha_1 = t/(s+t)$, $\alpha_2 = s/(s+t)$ and $\beta_1=\beta_2=1$.

Let us now describe the input set $S$. $S$ contains $d-2s+1$ copies of the ranking
\begin{equation}
\label{eqn:maj-ranking}
    (c_1, c_2, \cdots, c_t, a_1, a_2, \cdots, a_s, b_1, b_2, \cdots, b_s, d_1, d_2, \cdots, d_t).
\end{equation}

$S$ contains (one copy of) the ranking \[ (b_1, b_2, \cdots, b_s, c_1, c_2, \cdots, c_t, d_1, d_2, \cdots, d_t, a_1, a_2, \cdots, a_s). \]

It is easiest to describe the remaining $2s-2$ rankings as an iterative process. Begin with the ranking of (\ref{eqn:maj-ranking})
and take the lowest-ranked element of $B$ and place it in the first rank. Take the highest-ranked element of $A$ and place it at the last rank. This creates a new ranking, of which there are two copies in $S$, \[ (b_s, c_1, c_2, \cdots, c_t, a_2, \cdots, a_s, b_1, b_2, \cdots, b_{s-1}, d_1, d_2, \cdots, d_t, a_1). \]

Now repeat this process, creating a new ranking of which there are two copies in $S$, \[ (b_{s-1}, b_s, c_1, c_2, \cdots, c_t, a_3, \cdots, a_s,\]\[ b_1, b_2, \cdots, b_{s-2}, d_1, d_2, \cdots, d_t, a_1, a_2) \]

And so on, until all the elements have been moved. So in total, $S$ contains $d = 2s + 2t$ rankings.

Next, consider the following fair ranking $\pi^*$: \[ (c_1, c_2, \cdots, c_t, a_1, a_2, \cdots, a_s, b_1, b_2, \cdots b_s, d_1, d_2, \cdots, d_t). \]

It can be verified that $\pi^*$ is an optimal fair aggregate ranking of the set $S$. However, we do not really need to use the fact that $\pi^*$ is an optimal fair aggregate ranking. What we argue is that if $\pi^p$ is the output of our algorithm, then
\[
\obj(\pi^p) \ge (2-\epsilon) \obj(\pi^*) \ge (2-\epsilon) \opt
\]
for any $\epsilon >0$.

First, we claim that an optimal colorful bi-partition from step 1 of our algorithm is the partition $B \cup C$. To see why, observe that $a_s$ has the largest in-degree of $A$, and $b_1$ has the smallest in-degree of $B$. In all but one ranking, $a_s$ is ranked just before $b_1$, so these $d-1$ rankings contribute $d-1$ more to the in-degree of $b_1$ compared to $a_s$. In the final ranking, $b_1$ is ordered first, and $a_s$ is ordered last, giving $a_s$ an in-degree of $d-1$ more than $b_1$. So, the two elements have the same in-degree. Suppose we solve the two partitions optimally in step 2 of the algorithm. Then our algorithm's output $\pi^p$ is \[ (c_1, c_2, \cdots, c_t, b_1, b_2, \cdots, b_s, a_1, a_2, \cdots, a_s, d_1, d_2, \cdots, d_t) \]

Now, consider the median objective cost of both the optimal ranking and the ranking produced by our approximation algorithm. We now carefully count the objective cost of both rankings. Now, observe that $A, B, C, D$ form a partition over the universe of elements; so we can sum the contributions of all possible pairs.

Recall from the paper we have defined that for any $X \subseteq [d], Y \subseteq [d]$, \[ \kappa_{X \times Y}(\pi, \sigma) := |\{(a, b) \in X \times Y \mid a \prec_{\pi} b \text{\hspace{2mm}but\hspace{2mm}} b \prec_{\sigma} a\}| \]

As well as for any $X \subseteq [d], Y \subseteq [d]$, \[ \obj(S, \pi)_{X \times Y} := \sum_{\pi_i \in S} \kappa_{X \times Y}(\pi, \pi_i). \]

It is straightforward to see that in the following cases, the objective cost is 0, as all rankings in $S$, $\pi^p$, and $\pi^*$ agree on the pairwise ordering for any such pair.

\begin{align*}
    \obj(\pi^*)_{C \times C} = \obj(\pi^p)_{C \times C} = 0 \\
    \obj(\pi^*)_{D \times D} = \obj(\pi^p)_{D \times D} = 0 \\
    \obj(\pi^*)_{C \times D} = \obj(\pi^p)_{C \times D} = 0 \\
    \obj(\pi^*)_{B \times D} = \obj(\pi^p)_{B \times D} = 0 \\
    \obj(\pi^*)_{A \times C} = \obj(\pi^p)_{A \times C} = 0 \\
\end{align*}

The next cost we consider is $\obj(\pi^*)_{A \times A}$. Observe that this is equal to $\obj(\pi^p)_{A \times A}$, as both rankings have identical pairwise ordering for all such pairs.
\begin{align*}
    \obj(\pi^*)_{A \times A} & = 2 (1(s-1) + 2(s-2) + 3(s-3) + \\
    &\cdots + (s-2)(2) + (s-1)1 ) \\
    & = 2 \sum_{i=1}^{s-1} i(s-i) \\
    & = \frac{s^3-s}{3} .
\end{align*}

To see this, in two rankings of $S$, we have $a_s$ ordered before $a_1$ to $a_{s-1}$, creating $2(s-1)$ inversions. There are two rankings of $S$ where $a_{s-1}$ and $a_s$ are ordered before $a_1$ to $a_{s-2}$, creating $2(2(s-2))$ inversions, and so on.  The same argument holds for $\obj(\pi^*)_{B \times B}$, so the following hold.
\begin{equation*}
    \obj(\pi^*)_{A \times A} = \obj(\pi^*)_{B \times B} = \obj(\pi^p)_{B \times B}
\end{equation*}

There are three remaining cases of pairs left. Let us consider $\obj(\pi^p)_{A \times D}$. Observe that in two rankings of $S$, $a_1$ is ordered after all $t$ elements of $D$; in two rankings of $S$, $a_1$ and $a_2$ is ordered after all $t$ elements of $D$; and so on. Therefore, we can write the following.

\begin{align*}
    \obj(\pi^p)_{A \times D} &= 2t(1) + 2t(2) + \cdots + 2t(s-1) + ts \\
    & = ts + 2t \sum_{i=1}^{s-1} i \\
    & = ts + 2t\frac{s(s-1)}{2} \\
    & = ts^2 \\
    & = (\frac{d}{2}-s)s^2 \text{\hspace{1cm}(As $d = 2s+2t$)} \\
    & = \frac{ds^2}{2} - s^3 .
\end{align*}

Now, by a similar argument, $\obj(\pi^p)_{B \times C} = \frac{ds^2}{2} - s^3$. Again see that $\pi^p$ and $\pi^*$ order all such pairs in both these cases identically;, therefore $\obj(\pi^*)_{A \times D} = \obj(\pi^p)_{B \times C} = \frac{ds^2}{2} - s^3$.

The final case is $\obj(\pi^*)_{A \times B}$ and $\obj(\pi^p)_{A \times B}$.

We study $\obj(\pi^p)_{A \times B}$ first. In $\pi^p$, $b_i \prec a_j$ for all $i \in [s], j \in [s]$. There are $d-2s+1$ rankings which order $a_j \prec b_i$ for all $i \in [s], j \in [s]$. Consider the remaining $2s-1$ rankings. There are two rankings in $S$ where $a_j \prec b_i$ for all $i \in \{1, 2, \cdots, s-1\}, j \in \{2, 3, \cdots, s\}$, followed by two rankings in $S$ where $a_j \prec b_i$ for all $i \in \{1, 2, \cdots, s-2\}, j \in \{3, 4, \cdots, s\}$, and so on.
Summing these two contributions leads to a median objective cost of

\begin{align*}
    \obj(\pi^p)_{A \times B} &=  s^2(d-2s+1) + 2(s-1)^2 + 2(s-2)^2 +\\
    &\cdots + 2(1)^2 \\
    & = ds^2 - 2s^3 + s^2 + 2 \sum_{i=1}^{s-1} i^2 \\
    & = ds^2 - 2s^3 + s^2 + 2(\frac{s^3}{3} + \frac{s}{6} - \frac{s^2}{2}) \\
    & = ds^2 - \frac{4s^3}{3} + \frac{s}{3}
\end{align*}

We now study $\obj(\pi^*)_{A \times B}$. $\pi^*$ orders $a_i \prec b_j$ for any $i \in [s], j \in [s]$. Thus,
\begin{align*}
    \obj(\pi^*)_{A \times B} &= (2s-1)(2s-1) + (2s-3)(2s-3) +\\
    &\cdots + 3(3) + 1(1) \\
    & = \sum_{i=0}^{s-1} (2i+1)^2 \\
    & = \frac{4s^3-s}{3}
\end{align*}

To observe how this summation comes about, consider $b_s$ and $a_1$. We can see that $b_s$ is ranked above all $a_i$ for $i \in [s]$ in $2s-1$ rankings of $S$, and $a_1$ is ranked after all $b_j$ for all $j \in [s]$ in the same rankings. However, we need to avoid over-counting the pair $b_s$ and $a_1$. Similar observations for $b_{s-1}$ and $a_2$ give rise to the $(2s-3)(2s-3)$ term, and so on.

Now, we can sum everything up to obtain $\obj(\pi^p)$.
\begin{align*}
    \obj(\pi^p) &= \obj(\pi^p))_{A \times A} + \obj(\pi^p))_{B \times B} \\
    & + \obj(\pi^p)_{C \times C} + \obj(\pi^p)_{D \times D} \\ 
    & + \obj(\pi^p)_{A \times B} + \obj(\pi^p)_{A \times C} + \obj(\pi^p)_{A \times D} \\
    & + \obj(\pi^p)_{B \times C} + \obj(\pi^p)_{B \times D} + \obj(\pi^p)_{C \times D} \\
    & = 2\frac{s^3-s}{3} + ds^2 - \frac{4s^3}{3} + \frac{s}{3} + 2(\frac{ds^2}{2} - s^3) \\
    &= 2ds^2 - \frac{8s^3}{3} - \frac{s}{3} .
\end{align*}

By doing the same summation for $\pi^*$, we have
\begin{align*}
    \obj(\pi^*) &= 2\frac{s^3-s}{3} + \frac{4s^3-s}{3} + 2(\frac{ds^2}{2} - s^3)\\
    &= ds^2 - s .
\end{align*}

Then for any $\epsilon > 0$, we can pick suitable values of $s, d$ (where $s$ is much smaller than $d$, or in other words, $s << t$) to create an instance such that $\obj(\pi^p) \ge (2-\epsilon)\obj(\pi^*)$, showing that our analysis is tight.

\section{Further Details on Experiments}

Here, we describe the Integer Linear Program (ILP) used in computing the optimal result. Recall that the rank aggregation problem can be seen as a special case of feedback arc set on a weighted tournament. The tournament graph $G = (V, A)$ has vertex set $[d]$ and a pair of directed edges between each vertex. Let $n_{ab} = |\{ \pi \in S \mid a \prec_{\pi} b\}|$. Set the weight of the edge $(a, b)$ to be $w(a,b) = n_{ab} / n$ (where $n = |S|$). 

\begin{tabular}{l l}
        minimize & $\sum_{e = (a, b) \in E} w_e x_e$ \\
        subject to & \\
        $\forall$ distinct $a, b \in [d]$ & $x_{ab} + x_{ba} = 1$ \\
        $\forall$ distinct $a, b, c \in [d]$ & $x_{ab} + x_{bc} + x_{ca} \ge 1$ \\
        $\forall$  $a \in [d]$ & $- \sum_{b \in [d]} x_{ab} \le -k + 2d \cdot y_a$ \\
        $\forall$  $a \in [d]$ & $- \sum_{b \in [d]} x_{ab} \ge -(k - 1) - 2d(1 - y_a)$ \\

        $\forall$ $i \in [g]$ & $\lfloor \alpha_i k \rfloor \le \sum_{a \in G_i} y_a \le \lceil \beta_i k \rceil$ \\
        
\end{tabular}

The variable $x_{a,b}$ is an indicator of whether the edge $(a,b)$ is removed or not. That is, $x_{a, b} = 1$ iff the edge $(a, b)$ is removed, thus ordering $b$ before $a$ in the aggregate ranking. 

Similarly, the variable $y_a$ is an indicator of whether the element $a$ is placed in the top-$k$ positions or not. That is, $y_a = 1$ iff the element $a$ is placed in the top-$k$. The third and fourth inequalities achieve this. For some $a \in [d]$, the third inequality forces $y_a = 1$ if there are strictly less than $k$ elements ordered before it (that is, it must be in the top-$k$). If there are at least $k$ elements ordered before it, the fourth inequality forces $y_a = 0$. 
The last inequality is to enforce the fairness constraints for each group.

\paragraph{Implementation details}
For running the experiments, please note that the use of the following libraries. Numpy version 2.0.0 is used. The implementation uses Jupyter Notebook with Jupyterlab version 4.1.5. The integer linear programs are implemented using CVXPY version 1.4.2, and PySCIPOpt 5.1.1 was used as the solver. Other mixed integer program solvers should also work. To install CVXPY with SCIP, consider using pip with the command \texttt{pip install cvxpy[SCIP]}.

Note that the Kwiksort algorithm depends on randomness. To reproduce our results exactly, the seed should be set in the following manner to the Numpy random generator used. The generator should be seeded each time before the algorithm is run, with the following input array $(\texttt{dataset}, k, n, d, \Pi_{i \in [g]} 100 \cdot \alpha_i)$. \texttt{Dataset} is an integer from 1 to 16, depending on which instance of the football dataset is used. When using the Movielens dataset, this value should be excluded. Examples are provided in the code appendix.

Details about the folder structure of the code appendix are included as a readme in the root folder.

\paragraph{Further experimental results}

We plot experimental results on the other instances of the football dataset, as well as deeper studies on week 4 of the football dataset.

The plots of Figure~\ref{fig:football-all} plot the results of the paper for the other 15 instances in the football dataset. The values of $\bar{\alpha}$ are similarly chosen to be equal to the proportion of elements belonging to each group, and we set $k = 15$. The heading for each set of plots indicates the input instance. In all instances, our algorithm performs significantly better than \textsc{BestFromInput}.

\begin{figure}[h]
    \centering
    \includegraphics[width = 0.5\textwidth]{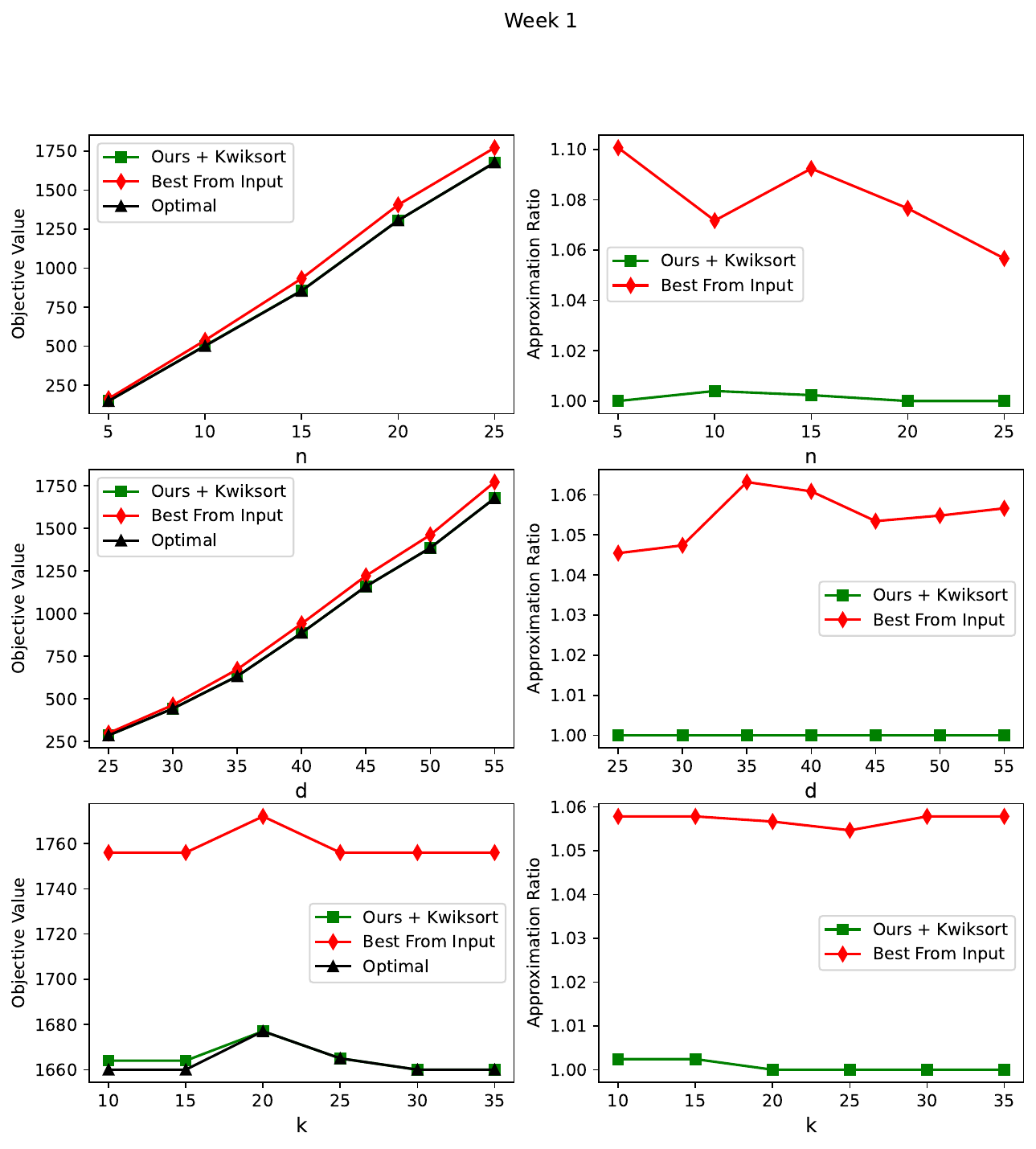}\\
    \includegraphics[width = 0.5\textwidth]{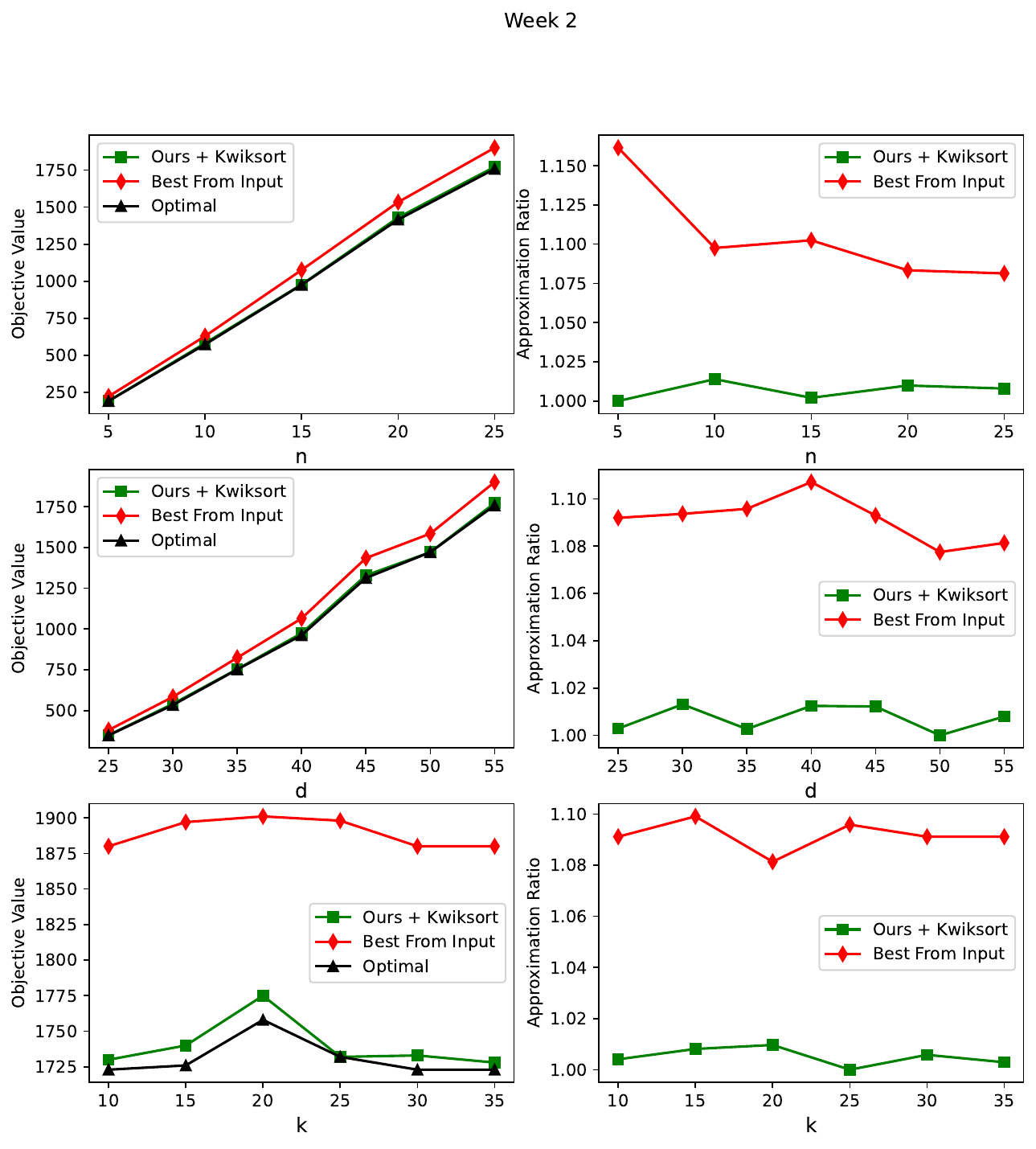}
\end{figure}

\begin{figure}[h]
    \centering
    \includegraphics[width = 0.5\textwidth]{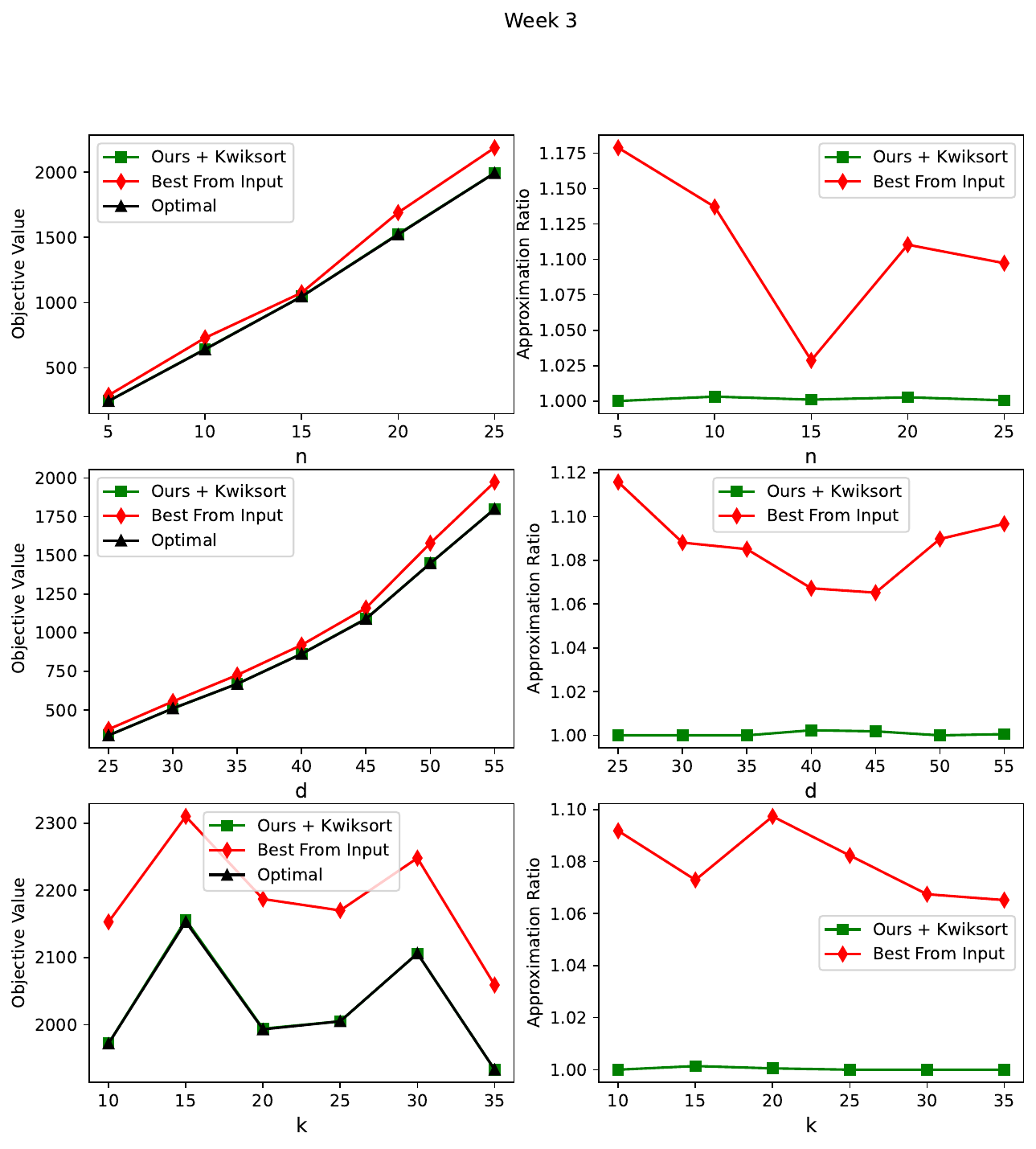}\\
    \includegraphics[width = 0.5\textwidth]{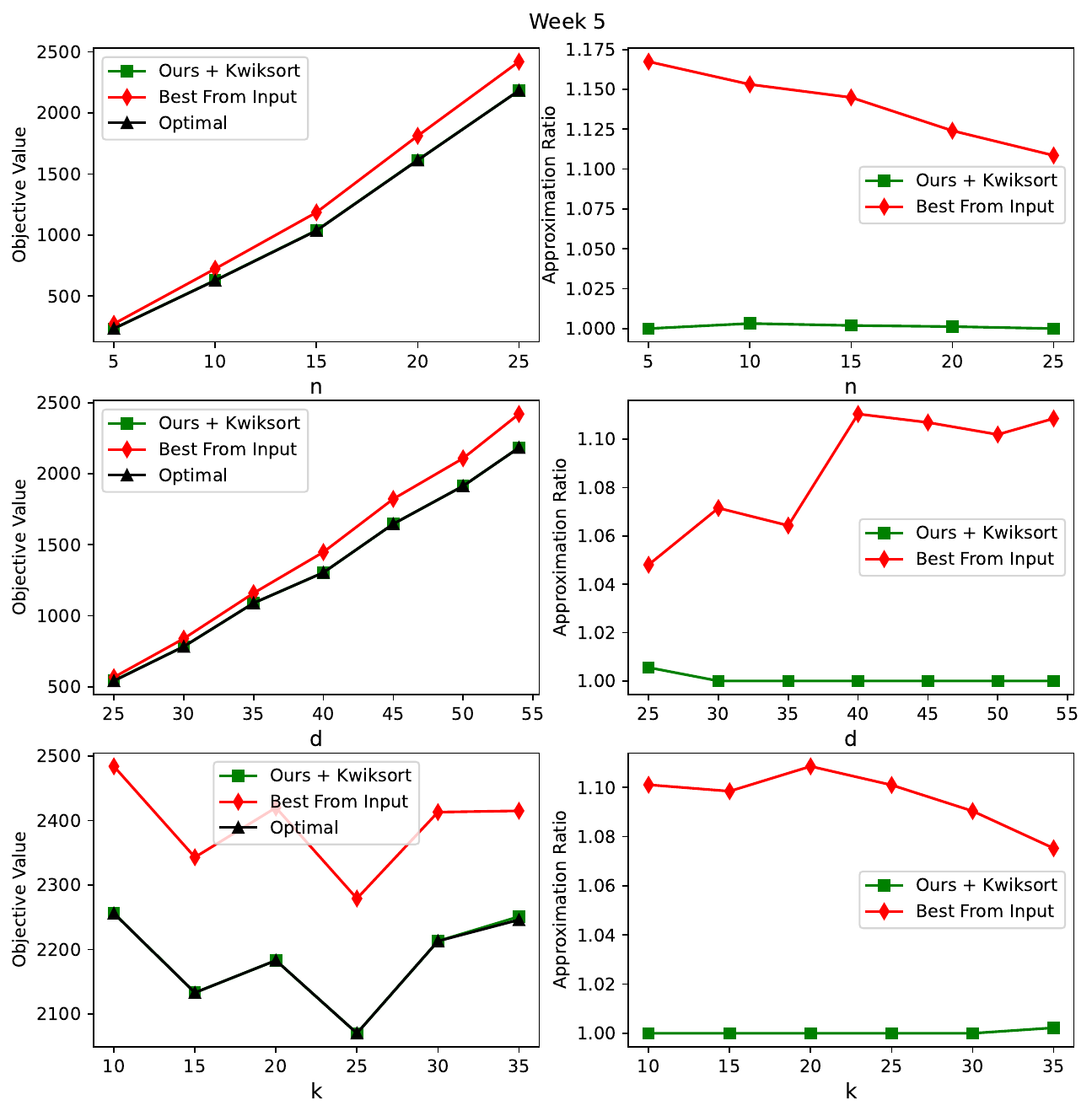}
\end{figure}

\begin{figure}[h]
    \centering
    \includegraphics[width = 0.5\textwidth]{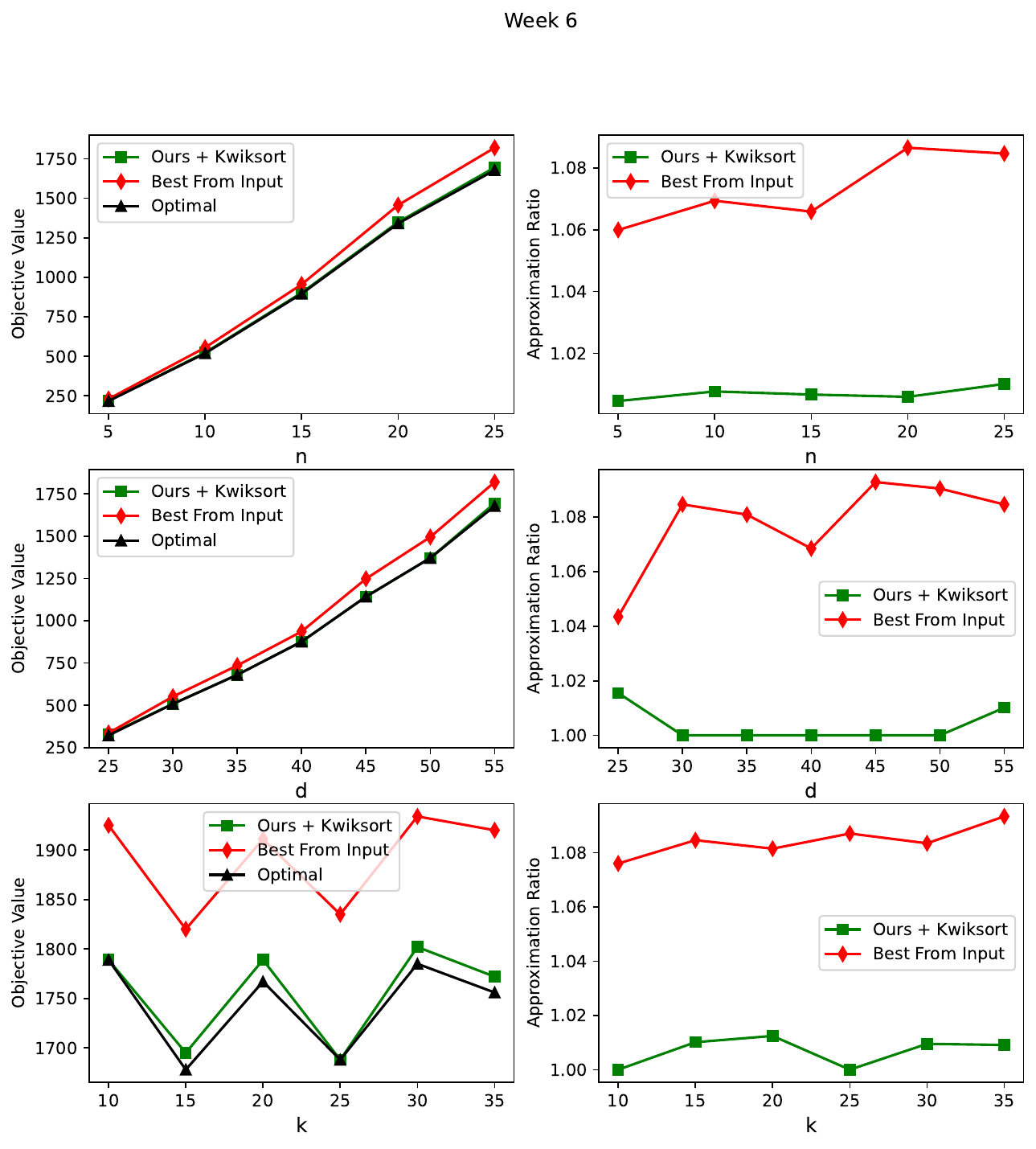}\\
    \includegraphics[width = 0.5\textwidth]{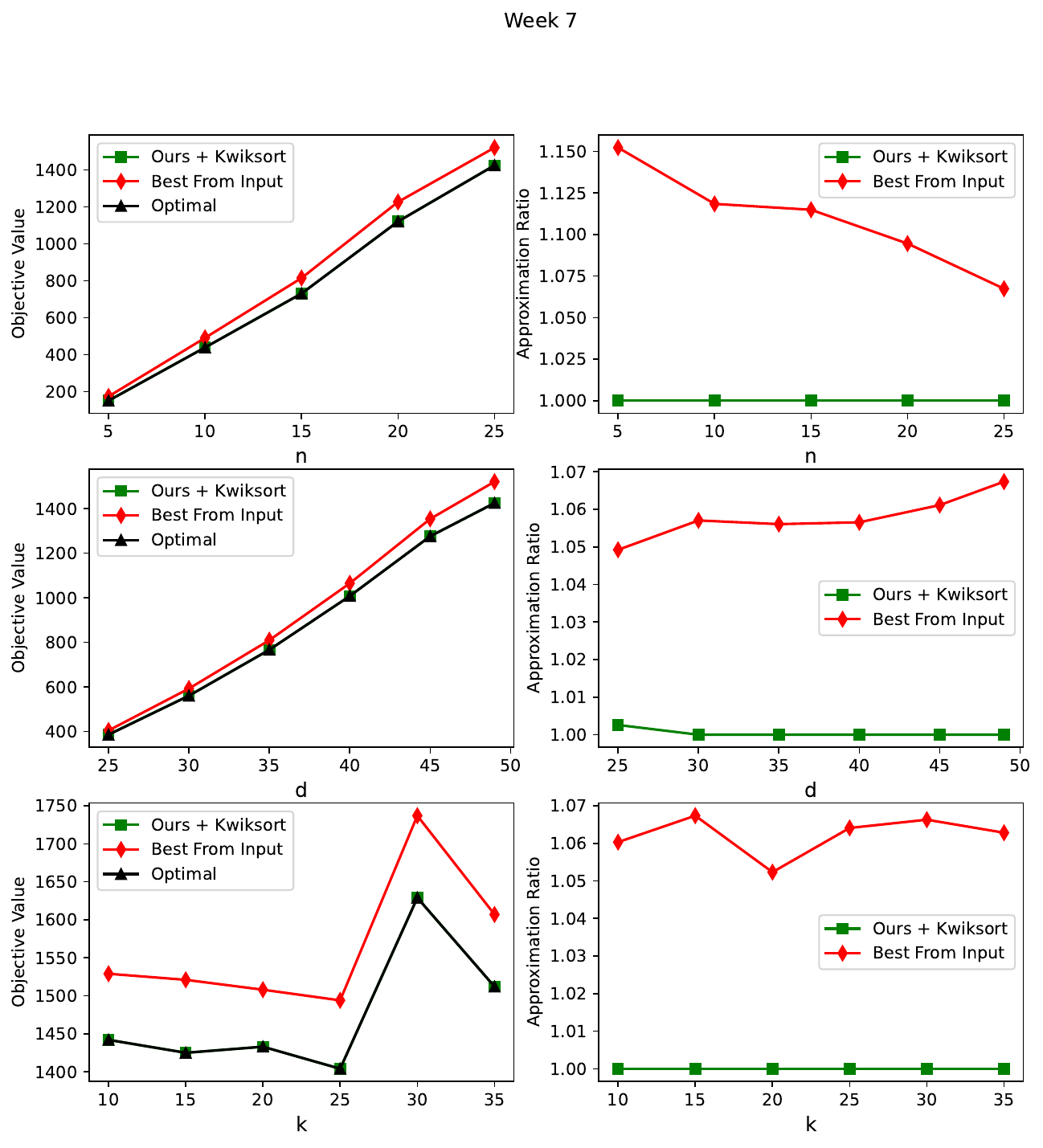}
\end{figure}

\begin{figure}[h]
    \centering
    \includegraphics[width = 0.5\textwidth]{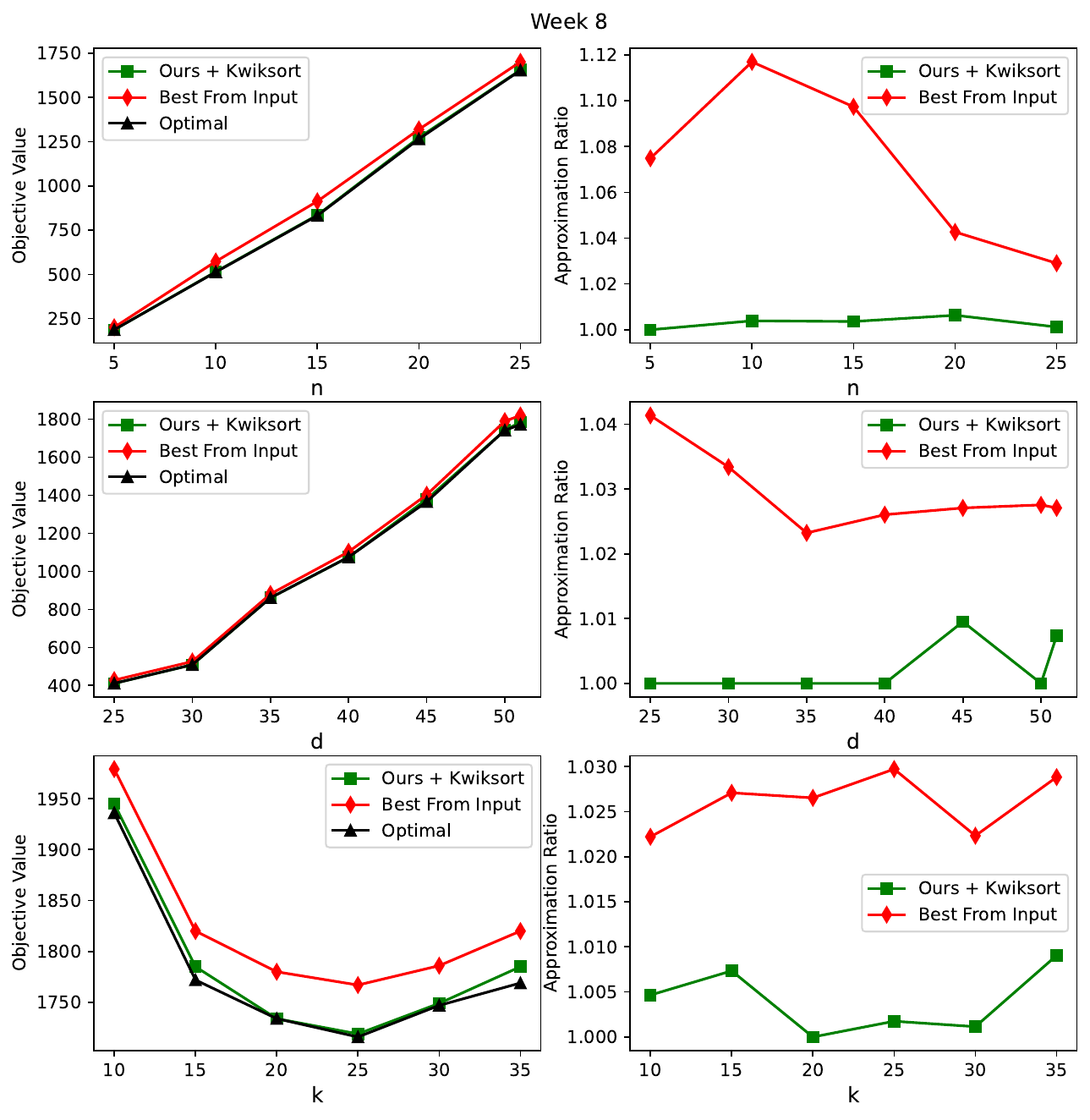}\\
    \includegraphics[width = 0.5\textwidth]{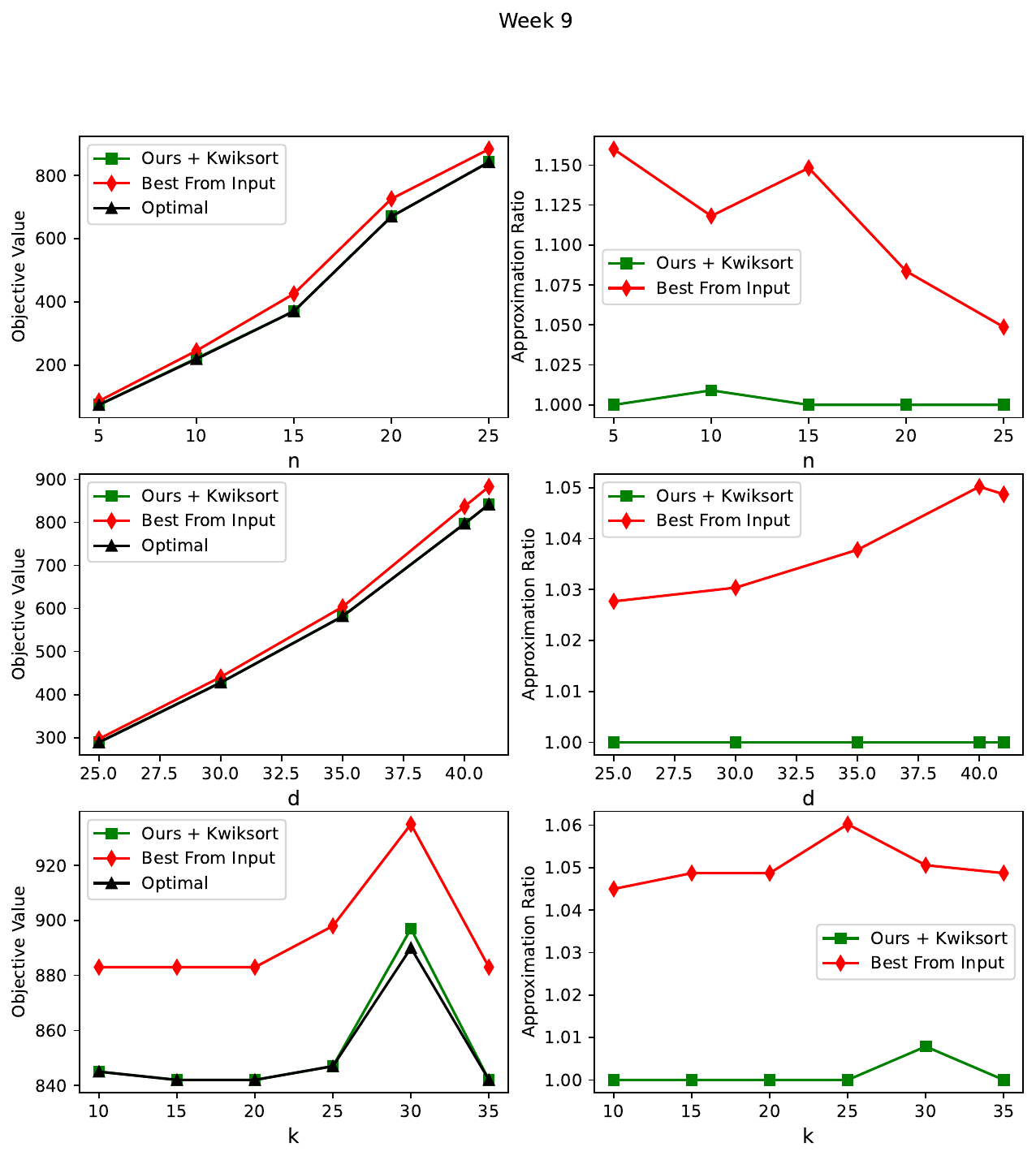}
\end{figure}

\begin{figure}[h]
    \centering
    \includegraphics[width = 0.5\textwidth]{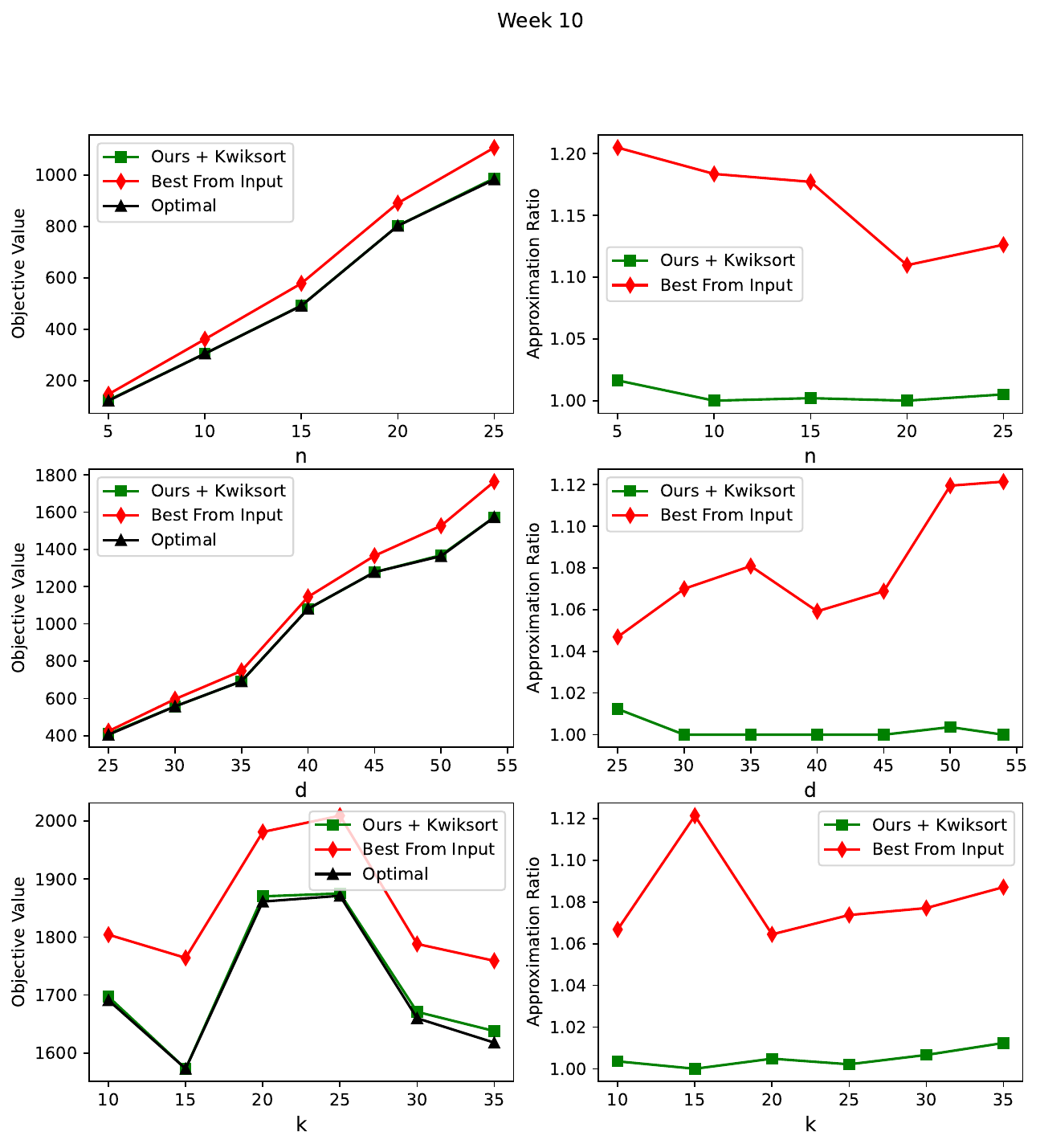}\\
    \includegraphics[width = 0.5\textwidth]{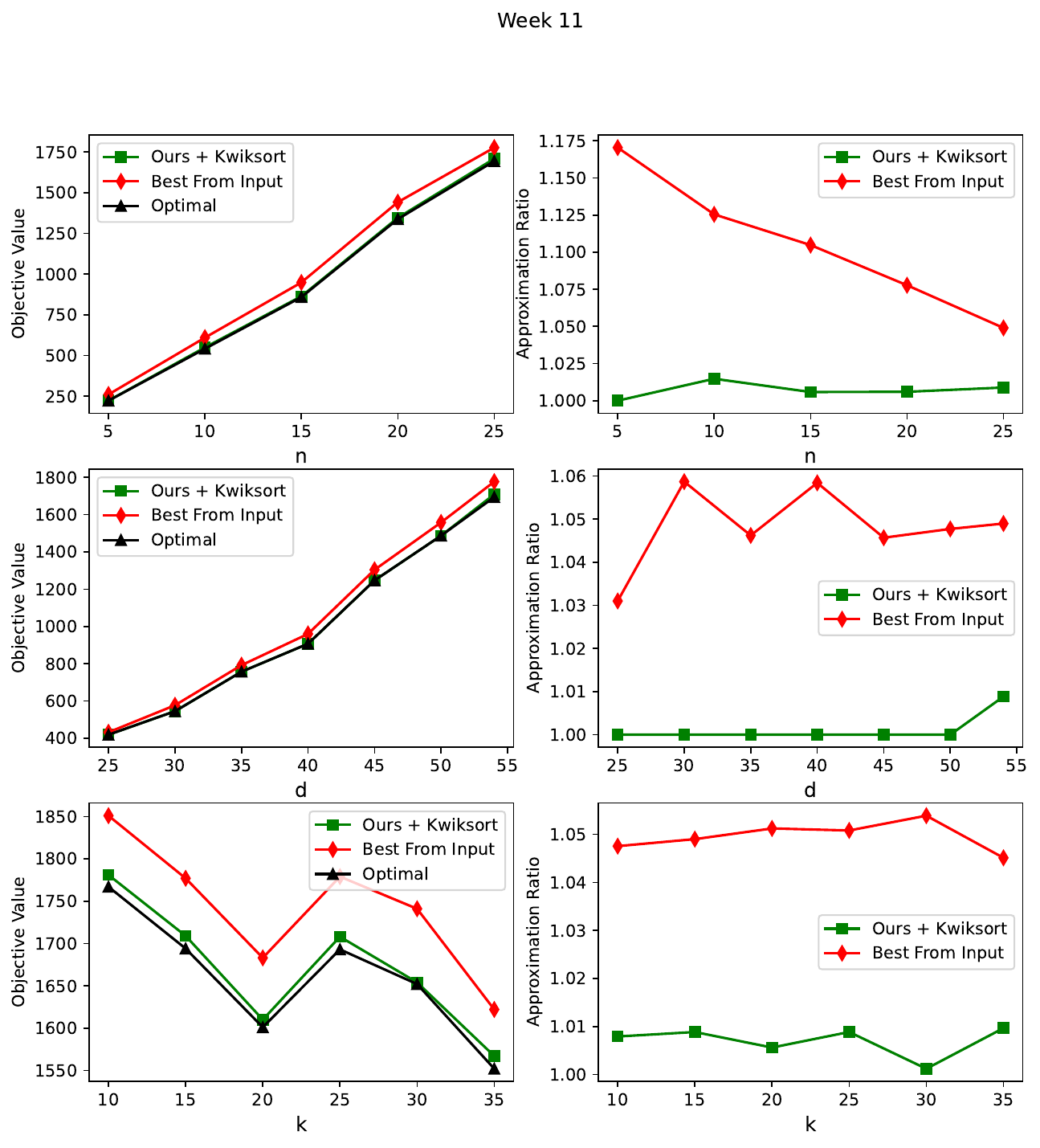}
\end{figure}
\begin{figure}[h]
    \centering
    \includegraphics[width = 0.5\textwidth]{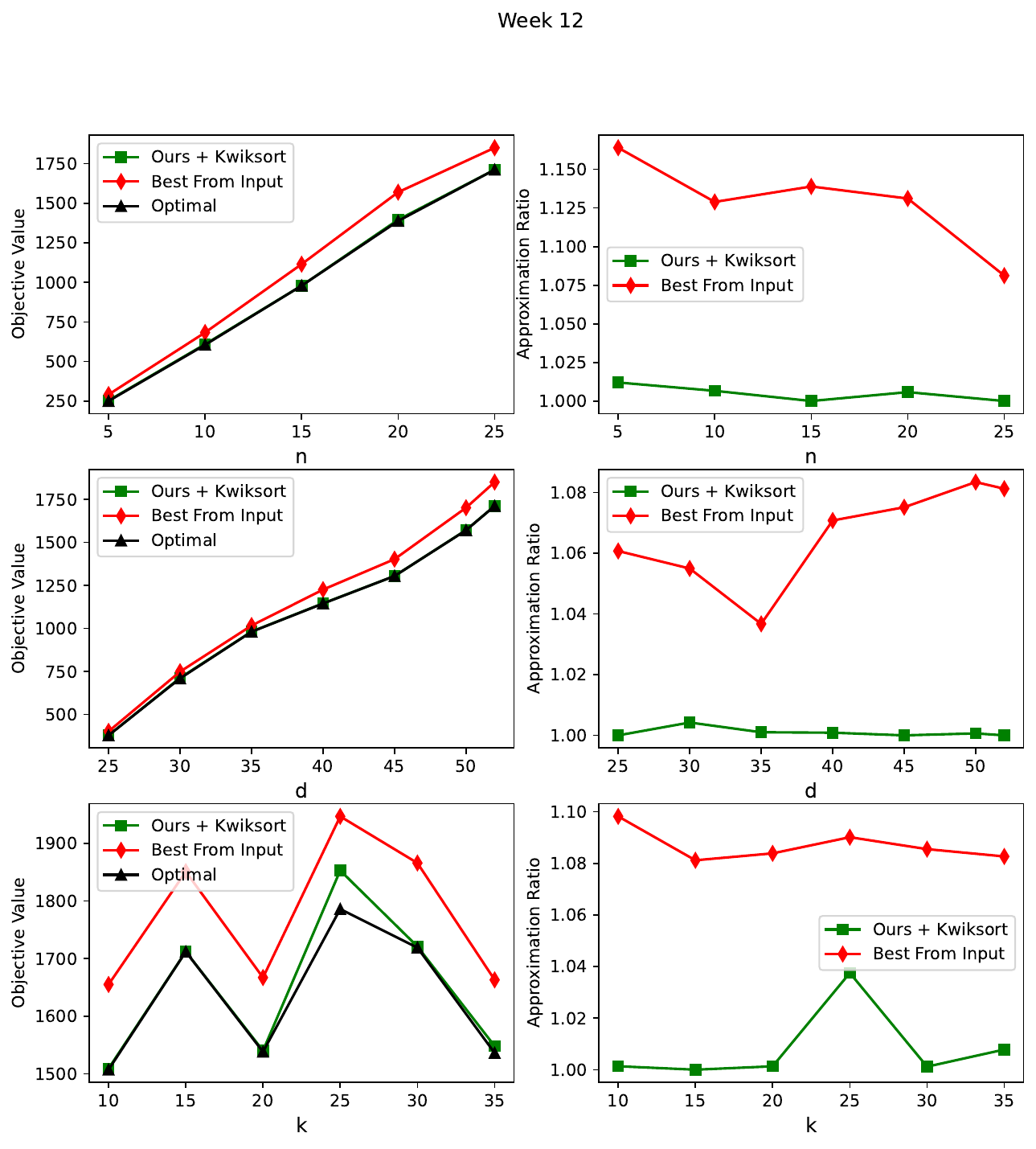}\\
    \includegraphics[width = 0.5\textwidth]{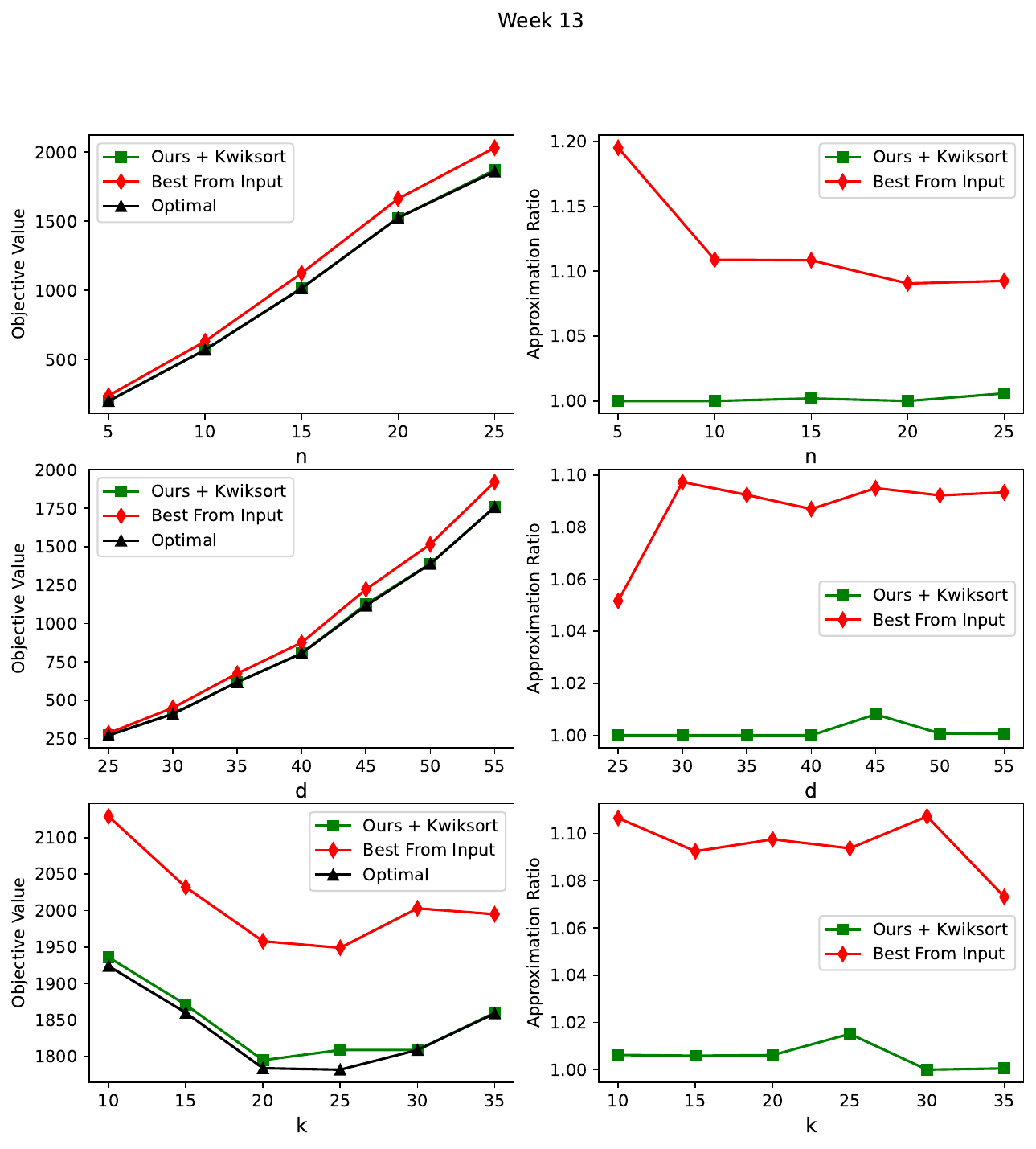}
\end{figure}
\begin{figure}[h]
    \centering
    \includegraphics[width = 0.5\textwidth]{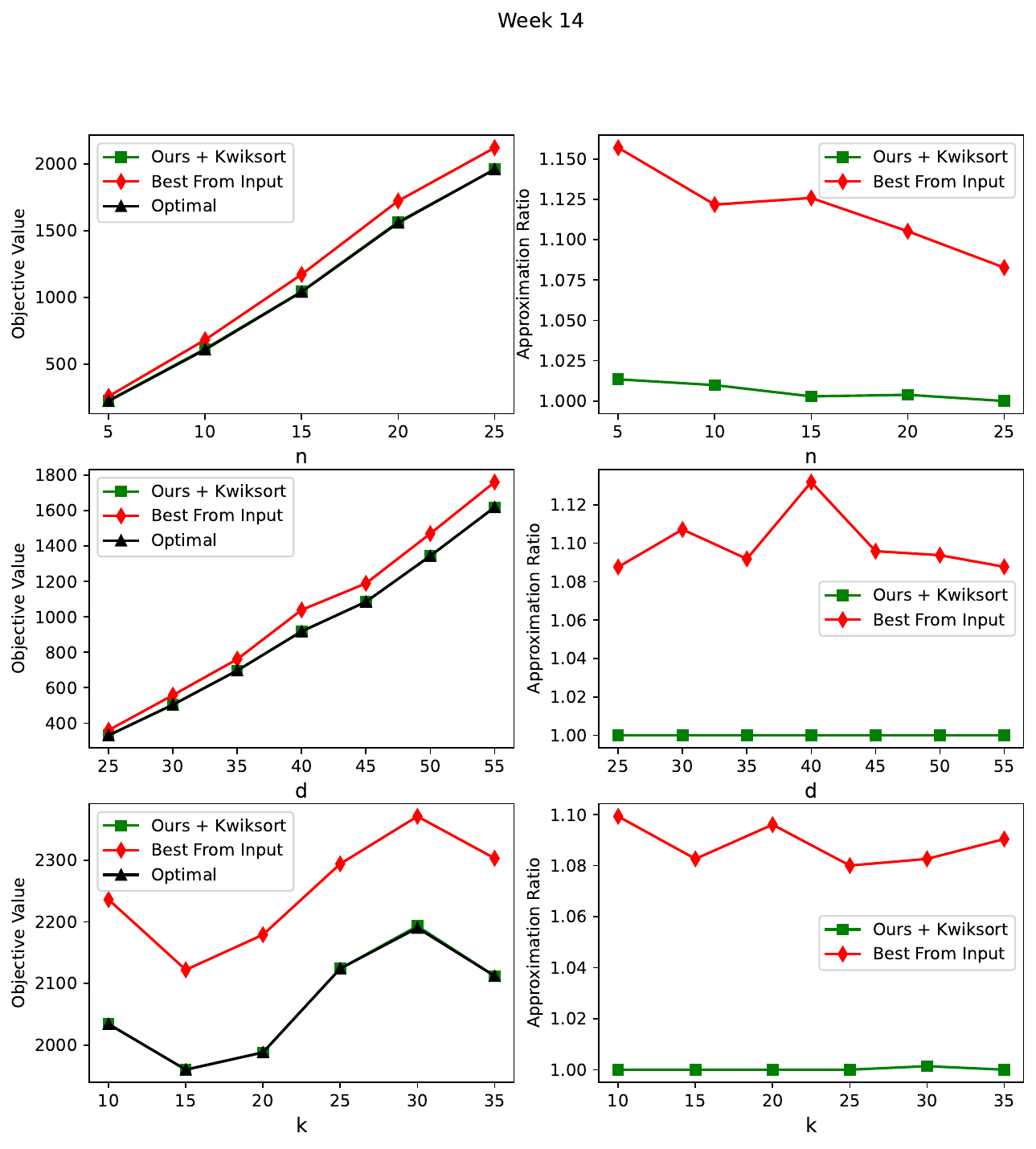}\\
    \includegraphics[width = 0.5\textwidth]{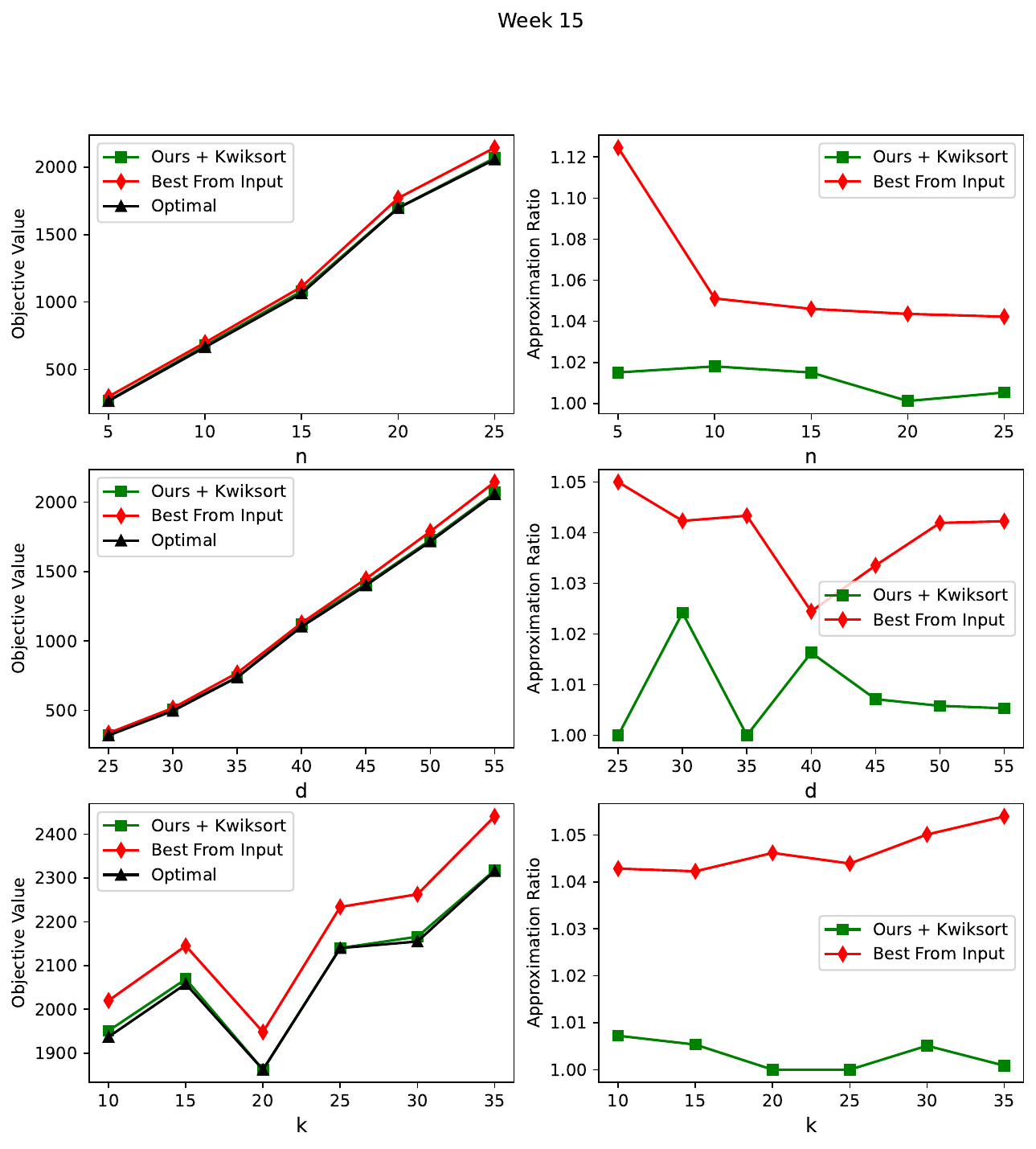}
\end{figure}
    
\begin{figure}[h]
    \centering
    \includegraphics[width = 0.5\textwidth]{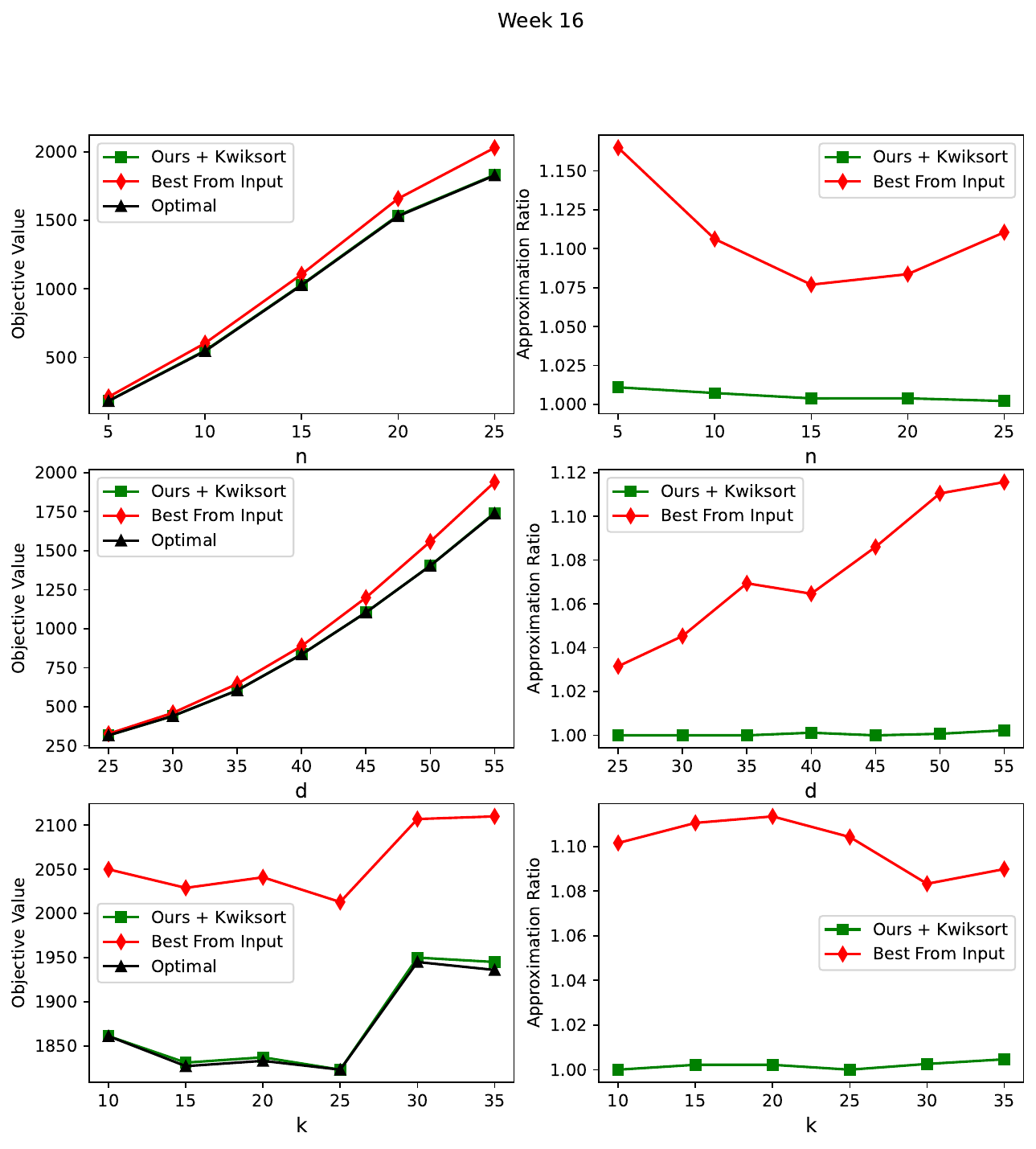}
    \caption{\textbf{Football dataset}. The input instance is the title of the set of plots. The x-axis indicates the value of the parameter ($n$, $d$ or $k$). The y-axis indicates the objective value of each output ranking on the left figure, with the corresponding approximation ratio on the right figure.}
    \label{fig:football-all}
\end{figure}

Further, we plot the experiments where we vary $n$ and $d$ for different values of the parameter $k$. In our paper, we plot only for $k = 15$ for experiments on the football dataset. Here, we plot for $k = 10$ and $k = 20$ for the same dataset when varying $n$ in Figure~\ref{fig:football-4-n} and when varying $d$ in Figure~\ref{fig:football-4-d}. Overall, we see that the results for different values of $k$ are similar. Note that for some instances, the number of elements is less than 55. In such cases, the final value of $d$ used in the experiment is equal to the total number of players. We also do the same for the Movielens dataset. We perform experiments for $k = 40$ and $k = 50$ when varying the values of $n$ and $d$ for the Movielens dataset as well, with the results placed in Figure~\ref{fig:movielens-othernd}. For all such experiments, we see that the performance of the algorithms remain consistent for different values of $k$ for these datasets.

\begin{figure}[h]
    \centering
    \includegraphics[width = 0.5\textwidth]{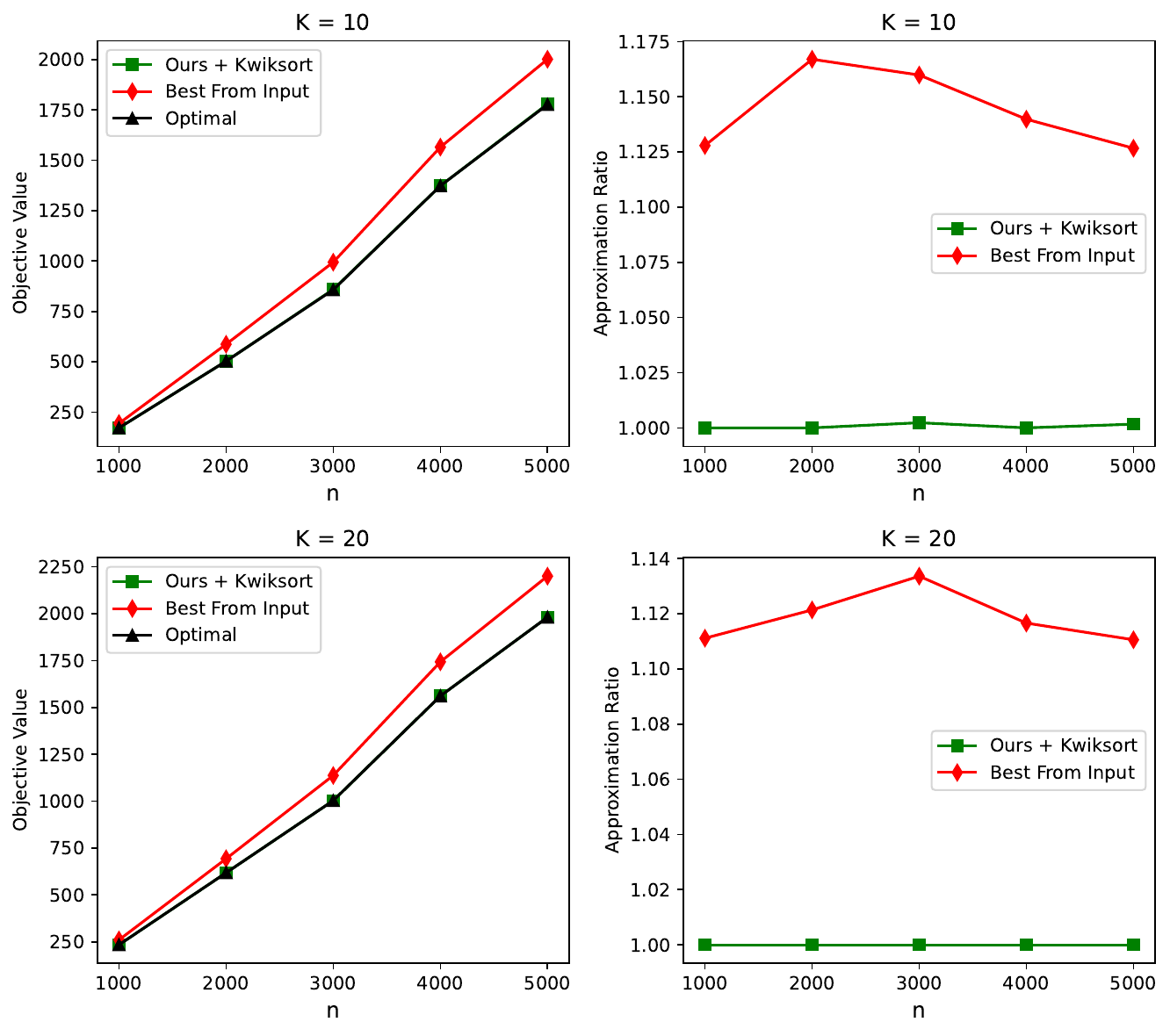}
    \caption{\textbf{Football dataset week 4}. Additional experiments varying $n$. The x-axis indicates the value of the parameter $n$. The y-axis indicates the objective value of each output ranking on the left figure, with the corresponding approximation ratio on the right figure.}
    \label{fig:football-4-n}
\end{figure}

\begin{figure}[h]
    \centering
    \includegraphics[width = 0.5\textwidth]{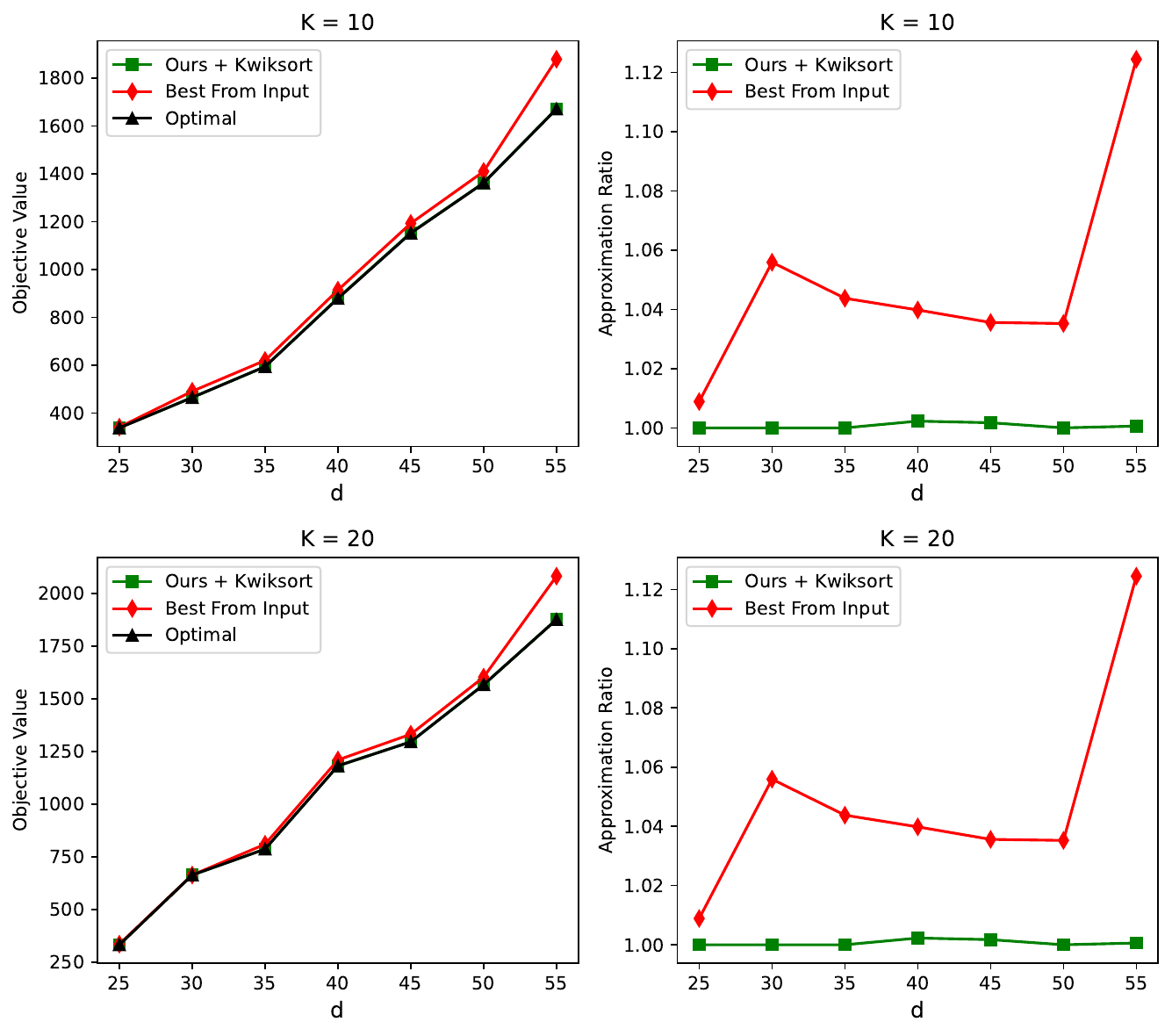}
    \caption{\textbf{Football dataset week 4}. Additional experiments varying $d$. The x-axis indicates the value of the parameter $d$. The y-axis indicates the objective value of each output ranking on the left figure, with the corresponding approximation ratio on the right figure.}
    \label{fig:football-4-d}
\end{figure}

\begin{figure}[h]
    \centering
    \includegraphics[width = 0.5\textwidth]{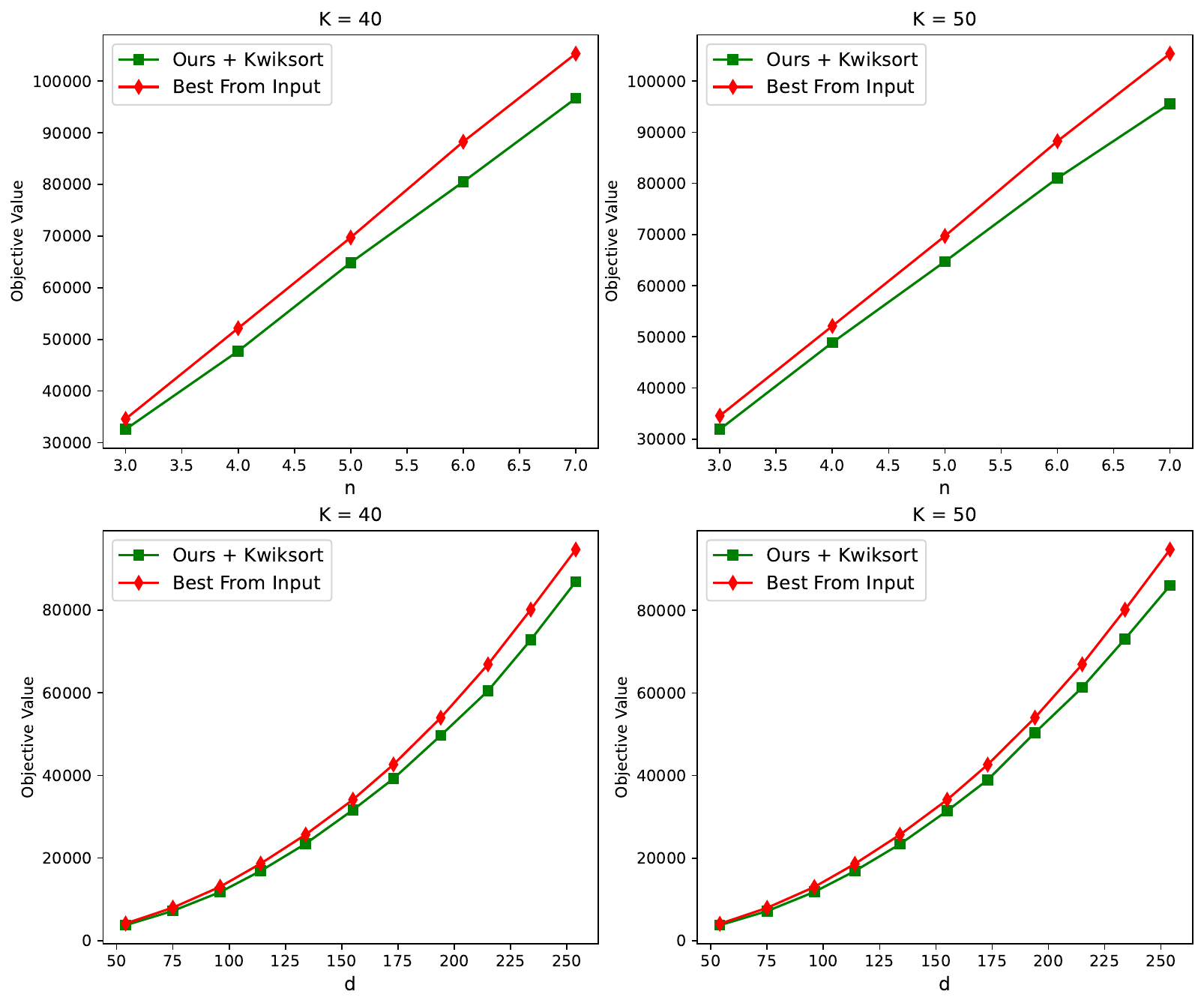}
    \caption{\textbf{Movielens dataset}. Additional experiments. The x-axis indicates the value of the parameter ($n$ or $d$). The y-axis indicates the objective value of each output ranking. Note that the two plots on the left are for $k=40$, and the two on the right for $k=50$.}
    \label{fig:movielens-othernd}
\end{figure}

We also perform experiments where we vary the values of $\bar{\alpha}$ to study the performance of our algorithm when the fairness parameter is varied. We do this for week 4 of the football dataset as well as the reduced Movielens dataset. The results for the football dataset are shown in figure~\ref{fig:football-alpha}. Here, we keep $\alpha_2 = 0.4$ and vary the value of $\alpha_1$, as this showed more interesting results with the objective value of the optimal fair ranking increasing as the fairness lower bound increased. We tested with both $k = 20$ and $k = 30$. The results for the experiments on the reduced Movielens dataset are in Figure~\ref{fig:movielens-alpha-10} and Figure~\ref{fig:movielens-alpha-20}. For these experiments, we first set $\alpha_i = 0.1$ for all $i \in [4]$, then pick one group $j$ to progressively increase the value of $\alpha_j$. As the value of $\alpha_j$ is increased, for each experiment, the objective value of the optimal fair ranking increases as expected. We tested with both $k = 10$ and $k = 20$. Even so, the performance of all the algorithms is quite consistent. Our algorithm still performs significantly better in all cases, showcasing that, in practice, it is able to choose top-$k$ elements that can lead to close to optimal results. 

\begin{figure}[h]
    \centering
    \includegraphics[width = 0.5\textwidth]{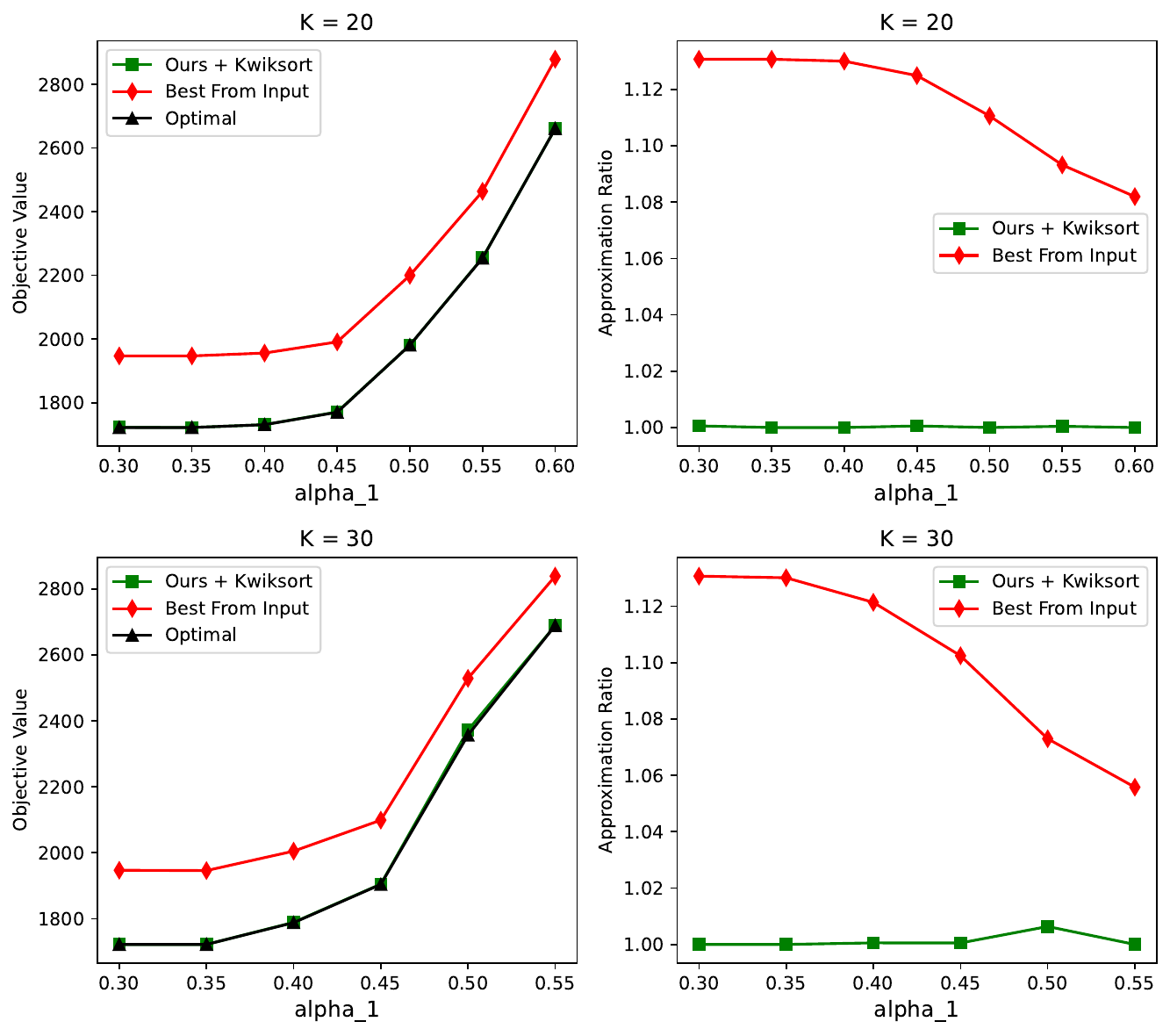}
    \caption{\textbf{Football dataset week 4}. Additional experiments varying $\alpha_1$. The x-axis indicates the value of the parameter $\alpha_1$. The y-axis indicates the objective value of each output ranking on the left figure, with the corresponding approximation ratio on the right figure.}
    \label{fig:football-alpha}
\end{figure}

\begin{figure}[h]
    \centering
    \includegraphics[width = 0.5\textwidth]{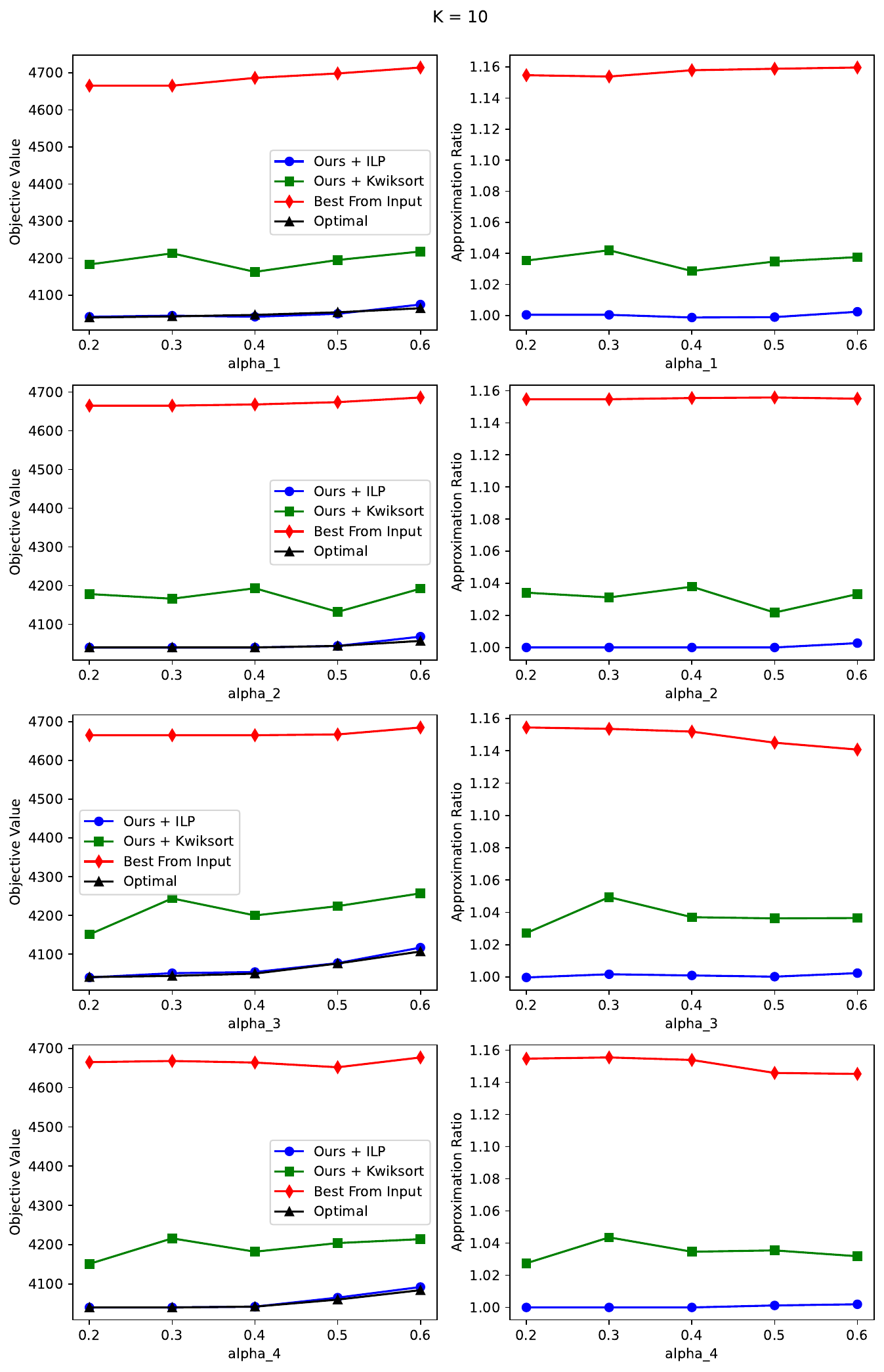}
    \caption{\textbf{Reduced Movielens dataset}. Additional experiments varying values of $\alpha_i$, $k = 10$. The x-axis indicates the value of the parameter $\alpha_i$. The y-axis indicates the objective value of each output ranking on the left figure, with the corresponding approximation ratio on the right figure.}
    \label{fig:movielens-alpha-10}
\end{figure}

\begin{figure}[h]
    \centering
    \includegraphics[width = 0.5\textwidth]{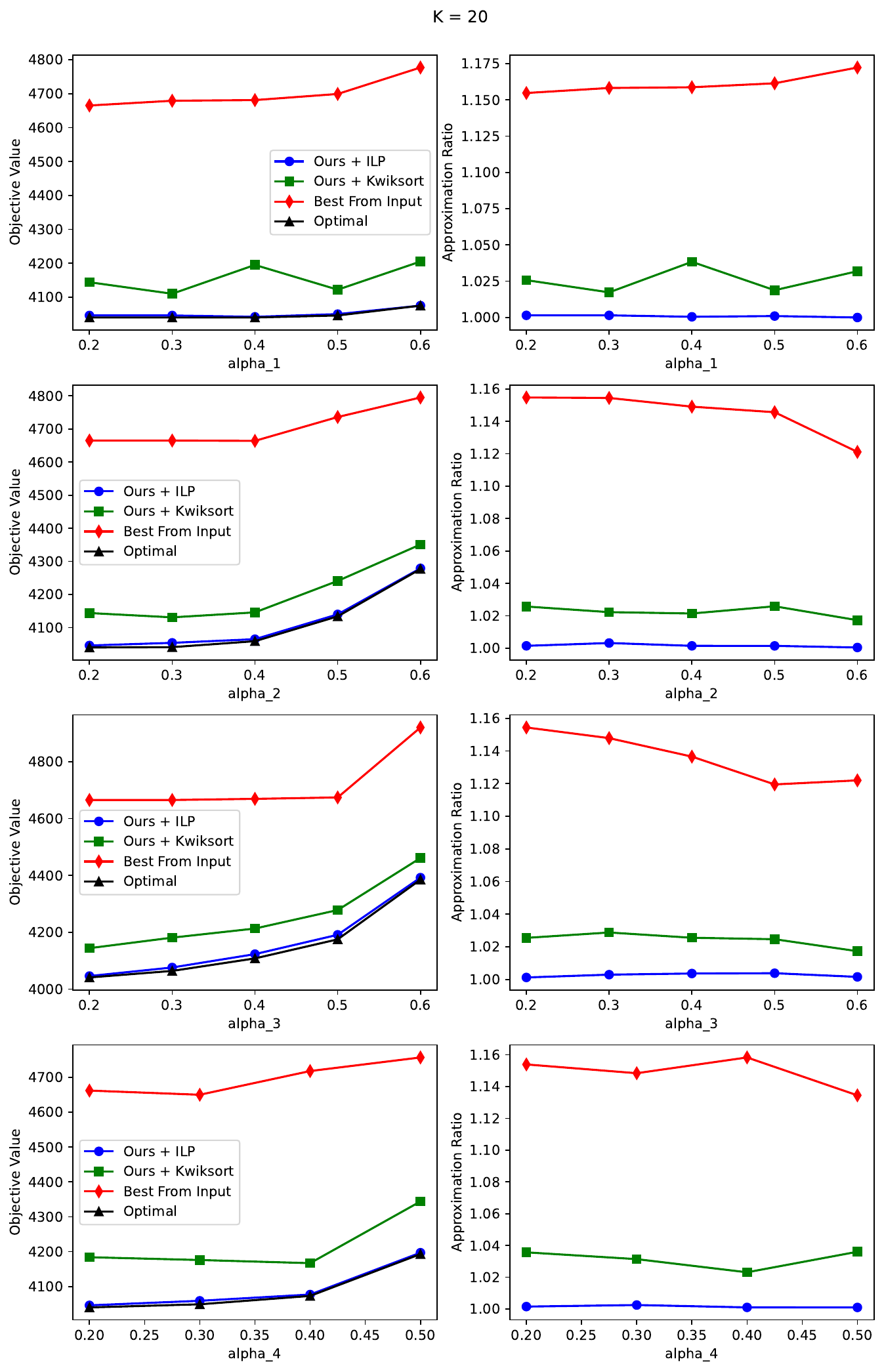}
    \caption{\textbf{Reduced Movielens dataset}. Additional experiments varying values of $\alpha_i$, $k = 20$. The x-axis indicates the value of the parameter $\alpha_i$. The y-axis indicates the objective value of each output ranking on the left figure, with the corresponding approximation ratio on the right figure.}
    \label{fig:movielens-alpha-20}
\end{figure}

We also take note of the running time of two implementations of our algorithm to solve each instance for various values of $k$ and note them in the following table. The values of $\bar{\alpha}, \bar{\beta}$ are set to be equal to the proportion of its size with respect to the group. As expected, the implementation, which uses Integer Linear Programming to solve rank aggregation optimally on each partition, is relatively slow and scales poorly as the input size increases. Implementing the Kwiksort algorithm to solve it approximately solves the instances much faster while still finding a solution to the optimal. This makes it more suitable in practice for fair rank aggregation. The time taken is the average of 5 runs due to variance in the frequency behavior of the processor. The list of running times is listed in Table~\ref{table-runtime}.

\begin{table*}[t]
\centering
\begin{tabular}{c|c|c}

    Input Instance and parameters & Ours + ILP & Ours + Kwiksort \\
    \hline
    Football Week 1, $k = 10$ & 90.3 & 0.0549 \\
    Football Week 1, $k = 15$ & 65.2 & 0.0534 \\
    Football Week 1, $k = 20$ & 41.1 & 0.0511 \\
    Football Week 2, $k = 10$ & 85.1 & 0.0554 \\
    Football Week 2, $k = 15$ & 53.1 & 0.0524 \\
    Football Week 2, $k = 20$ & 36.4 & 0.0516 \\
    Football Week 3, $k = 10$ & 113 & 0.0677 \\
    Football Week 3, $k = 15$ & 92.5 & 0.0594 \\
    Football Week 3, $k = 20$ & 69.4 & 0.0551 \\
    Football Week 4, $k = 10$ & 157 & 0.0585 \\
    Football Week 4, $k = 15$ & 88.7 & 0.0558 \\
    Football Week 4, $k = 20$ & 51.4 & 0.0527 \\
    Football Week 5, $k = 10$ & 99.5 & 0.0527 \\
    Football Week 5, $k = 15$ & 59.2 & 0.0515 \\
    Football Week 5, $k = 20$ & 33.0 & 0.0493 \\
    Football Week 6, $k = 10$ & 161 & 0.0549 \\
    Football Week 6, $k = 15$ & 82.5 & 0.0520 \\
    Football Week 6, $k = 20$ & 49.8 & 0.0508 \\
    Football Week 7, $k = 10$ & 53.6 & 0.0462 \\
    Football Week 7, $k = 15$ & 30.4 & 0.0436 \\
    Football Week 7, $k = 20$ & 20.0 & 0.0452 \\
    Football Week 8, $k = 10$ & 74.0 & 0.0506 \\
    Football Week 8, $k = 15$ & 43.3 & 0.0473 \\
    Football Week 8, $k = 20$ & 26.6 & 0.0460 \\
    Football Week 9, $k = 10$ & 33.1 & 0.0418 \\
    Football Week 9, $k = 15$ & 12.7 & 0.0338 \\
    Football Week 9, $k = 20$ & 11.1 & 0.0329 \\
    Football Week 10, $k = 10$ & 148 & 0.0534 \\
    Football Week 10, $k = 15$ & 82.0 & 0.0500 \\
    Football Week 10, $k = 20$ & 41.3 & 0.0495 \\
    Football Week 11, $k = 10$ & 151 & 0.0534 \\
    Football Week 11, $k = 15$ & 91.2 & 0.0501 \\
    Football Week 11, $k = 20$ & 56.3 & 0.0488 \\
    Football Week 12, $k = 10$ & 129 & 0.0533 \\
    Football Week 12, $k = 15$ & 71.2 & 0.0472 \\
    Football Week 12, $k = 20$ & 33.4 & 0.0471 \\
    Football Week 13, $k = 10$ & 168 & 0.0592 \\
    Football Week 13, $k = 15$ & 102 & 0.0564 \\
    Football Week 13, $k = 20$ & 59.2 & 0.0557 \\
    Football Week 14, $k = 10$ & 267 & 0.0673 \\
    Football Week 14, $k = 15$ & 166 & 0.0629 \\
    Football Week 14, $k = 20$ & 87.1 & 0.0606 \\
    Football Week 15, $k = 10$ & 179 & 0.0628 \\
    Football Week 15, $k = 15$ & 95.6 & 0.0538 \\
    Football Week 15, $k = 20$ & 57.2 & 0.0540 \\
    Football Week 16, $k = 10$ & 166 & 0.0617 \\
    Football Week 16, $k = 15$ & 103 & 0.0550 \\
    Football Week 16, $k = 20$ & 62.6 & 0.0549 \\
    Movielens, $k = 30$ & - & 0.242 \\
    Movielens, $k = 40$ & - & 0.233 \\
    Movielens, $k = 50$ & - & 0.236 \\
    Reduced Movielens, $k = 10$ & 148 & .0125 \\
    Reduced Movielens, $k = 15$ & 95.1 & .0128 \\
    Reduced Movielens, $k = 20$ & 55.7 & .0121
\end{tabular}
\caption{Runtime of implemented algorithms across datasets. All runtimes are in seconds.}
\label{table-runtime}
\end{table*}

\end{appendix}

%% file: sample.bib
@article{kemeny1959,
  title={Mathematics without numbers},
  author={Kemeny, John G},
  journal={Daedalus},
  volume={88},
  number={4},
  pages={577--591},
  year={1959},
  publisher={JSTOR}
}

@article{young1978,
  title={A consistent extension of Condorcet’s election principle},
  author={Young, H Peyton and Levenglick, Arthur},
  journal={SIAM Journal on applied Mathematics},
  volume={35},
  number={2},
  pages={285--300},
  year={1978},
  publisher={SIAM}
}

@article{young1988,
  title={Condorcet's theory of voting},
  author={Young, H Peyton},
  journal={American Political science review},
  volume={82},
  number={4},
  pages={1231--1244},
  year={1988},
  publisher={Cambridge University Press}
}

@misc{harman1992,
  title={Ranking Algorithms.},
  author={Harman, Donna},
  year={1992}
}

@inproceedings{dwork2001,
  title={Rank aggregation methods for the web},
  author={Dwork, Cynthia and Kumar, Ravi and Naor, Moni and Sivakumar, Dandapani},
  booktitle={Proceedings of the 10th international conference on World Wide Web},
  pages={613--622},
  year={2001}
}

@inproceedings{fagin2003,
  title={Efficient similarity search and classification via rank aggregation},
  author={Fagin, Ronald and Kumar, Ravi and Sivakumar, Dandapani},
  booktitle={Proceedings of the 2003 ACM SIGMOD international conference on Management of data},
  pages={301--312},
  year={2003}
}

@article{popov2007,
  title={Multiple genome rearrangement by swaps and by element duplications},
  author={Popov, V Yu},
  journal={Theoretical computer science},
  volume={385},
  number={1-3},
  pages={115--126},
  year={2007},
  publisher={Elsevier}
}

@article{ailon2008,
  title={Aggregating inconsistent information: ranking and clustering},
  author={Ailon, Nir and Charikar, Moses and Newman, Alantha},
  journal={Journal of the ACM (JACM)},
  volume={55},
  number={5},
  pages={1--27},
  year={2008},
  publisher={ACM New York, NY, USA}
}

@article{biedl2009,
  title={On the complexity of crossings in permutations},
  author={Biedl, Therese and Brandenburg, Franz J and Deng, Xiaotie},
  journal={Discrete Mathematics},
  volume={309},
  number={7},
  pages={1813--1823},
  year={2009},
  publisher={Elsevier}
}

@inproceedings{gleich2011,
  title={Rank aggregation via nuclear norm minimization},
  author={Gleich, David F and Lim, Lek-heng},
  booktitle={Proceedings of the 17th ACM SIGKDD international conference on Knowledge discovery and data mining},
  pages={60--68},
  year={2011}
}

@article{azari2013,
  title={Generalized method-of-moments for rank aggregation},
  author={Azari Soufiani, Hossein and Chen, William and Parkes, David C and Xia, Lirong},
  journal={Advances in Neural Information Processing Systems},
  volume={26},
  year={2013}
}

@article{bachmaier2015,
  title={On the hardness of maximum rank aggregation problems},
  author={Bachmaier, Christian and Brandenburg, Franz J and Glei{\ss}ner, Andreas and Hofmeier, Andreas},
  journal={Journal of Discrete Algorithms},
  volume={31},
  pages={2--13},
  year={2015},
  publisher={Elsevier}
}

@book{brandt2016,
  title={Handbook of computational social choice},
  author={Brandt, Felix and Conitzer, Vincent and Endriss, Ulle and Lang, J{\'e}r{\^o}me and Procaccia, Ariel D},
  year={2016},
  publisher={Cambridge University Press}
}

@article{barbero2018,
  title={Exploring the complexity of layout parameters in tournaments and semicomplete digraphs},
  author={Barbero, Florian and Paul, Christophe and Pilipczuk, Micha{\l}},
  journal={ACM Transactions on Algorithms (TALG)},
  volume={14},
  number={3},
  pages={1--31},
  year={2018},
  publisher={ACM New York, NY, USA}
}

@article{chakraborty2022,
  title={Fair rank aggregation},
  author={Chakraborty, Diptarka and Das, Syamantak and Khan, Arindam and Subramanian, Aditya},
  journal={Advances in Neural Information Processing Systems},
  volume={35},
  pages={23965--23978},
  year={2022},
  note = {Full version: arXiv preprint arXiv:2308.10499}
}

@inproceedings{wei2022rank,
  title={Rank aggregation with proportionate fairness},
  author={Wei, Dong and Islam, Md Mouinul and Schieber, Baruch and Basu Roy, Senjuti},
  booktitle={Proceedings of the 2022 international conference on management of data},
  pages={262--275},
  year={2022}
}

@article{bartholdi1989voting,
  author    = {J.J. Bartholdi and C.A. Tovey and M.A. Trick},
  title     = {Voting schemes for which it can be difficult to tell who won the election},
  journal   = {Social Choice and Welfare},
  volume    = {6},
  number    = {2},
  pages     = {157--165},
  year      = {1989}
}

@misc{Dwork2002RankAR,
  title={Rank Aggregation Revisited},
  author={Cynthia Dwork and Ravi Kumar and Moni Naor and D. Sivakumar},
  year={2002}
}

@misc{mathieu2009rank,
  title={How to rank with fewer errors},
  author={Mathieu, C and Schudy, W},
  year={2009},
  howpublished="\url{https://cs.brown.edu/people/wschudy/papers/fast_journal.pdf}",
  note={Unpublished}
}

@inproceedings{celis2018ranking,
  author    = {L. Elisa Celis and Damian Straszak and Nisheeth K. Vishnoi},
  title     = {Ranking with fairness constraints},
  booktitle = {International Colloquium on Automata, Languages, and Programming (ICALP)},
  volume    = {107},
  pages     = {28:1--28:15},
  year      = {2018},
}

@book{condorcet1785essai,
  author    = {M. J. Condorcet},
  title     = {Essai sur l'application de l'analyse \`a la probabilit\'e des d\'ecisions rendues \`a la pluralit\'e des voix},
  year      = {1785},
  note      = {Reprinted by AMS Bookstore in 1972},
}

@article{borda1781memoire,
  author    = {J. Borda},
  title     = {M\'emoire sur les \'elections au scrutin},
  journal   = {Histoire de l'Acad\'emie Royale des Sciences},
  year      = {1781},
}

@article{huang2019coresets,
  title={Coresets for clustering with fairness constraints},
  author={Huang, Lingxiao and Jiang, Shaofeng and Vishnoi, Nisheeth},
  journal={Advances in neural information processing systems},
  volume={32},
  year={2019}
}

@inproceedings{chen2019proportionally,
  title={Proportionally fair clustering},
  author={Chen, Xingyu and Fain, Brandon and Lyu, Liang and Munagala, Kamesh},
  booktitle={International conference on machine learning},
  pages={1032--1041},
  year={2019},
  organization={PMLR}
}

@article{bera2019fair,
  title={Fair algorithms for clustering},
  author={Bera, Suman and Chakrabarty, Deeparnab and Flores, Nicolas and Negahbani, Maryam},
  journal={Advances in Neural Information Processing Systems},
  volume={32},
  year={2019}
}

@inproceedings{backurs2019scalable,
  title={Scalable fair clustering},
  author={Backurs, Arturs and Indyk, Piotr and Onak, Krzysztof and Schieber, Baruch and Vakilian, Ali and Wagner, Tal},
  booktitle={International Conference on Machine Learning},
  pages={405--413},
  year={2019},
  organization={PMLR}
}

@article{cvxpy,
  title={CVXPY: A Python-embedded modeling language for convex optimization},
  author={Diamond, Steven and Boyd, Stephen},
  journal={Journal of Machine Learning Research},
  volume={17},
  number={83},
  pages={1--5},
  year={2016}
}

@techreport{BolusaniEtal2024OO,
  author = {Suresh Bolusani and Mathieu Besan{\c{c}}on and Ksenia Bestuzheva and Antonia Chmiela and Jo{\~{a}}o Dion{\'{i}}sio and Tim Donkiewicz and Jasper van Doornmalen and Leon Eifler and Mohammed Ghannam and Ambros Gleixner and Christoph Graczyk and Katrin Halbig and Ivo Hedtke and Alexander Hoen and Christopher Hojny and Rolf van der Hulst and Dominik Kamp and Thorsten Koch and Kevin Kofler and Jurgen Lentz and Julian Manns and Gioni Mexi and Erik~M\"{u}hmer and Marc E. Pfetsch and Franziska Schl{\"o}sser and Felipe Serrano and Yuji Shinano and Mark Turner and Stefan Vigerske and Dieter Weninger and Lixing Xu},
  title = {{The SCIP Optimization Suite 9.0}},
  type = {Technical Report},
  institution = {Optimization Online},
  month = {February},
  year = {2024},
  url = {https://optimization-online.org/2024/02/the-scip-optimization-suite-9-0/}
}

@article{kuhlman2020rank,
  title={Rank aggregation algorithms for fair consensus},
  author={Kuhlman, Caitlin and Rundensteiner, Elke},
  journal={Proceedings of the VLDB Endowment},
  volume={13},
  number={12},
  year={2020}
}

@article{borooah2010social,
  title={Social exclusion and jobs reservation in India},
  author={Borooah, Vani K},
  journal={Economic and Political Weekly},
  pages={31--35},
  year={2010},
  publisher={JSTOR}
}

@article{deshpande2005affirmative,
  title={Affirmative action in India and the United States},
  author={Deshpande, Ashwini},
  year={2005},
  publisher={Washington, DC: World Bank},
  journal={-}
}

@article{costello2016views,
  title={Who views online extremism? Individual attributes leading to exposure},
  author={Costello, Matthew and Hawdon, James and Ratliff, Thomas and Grantham, Tyler},
  journal={Computers in Human Behavior},
  volume={63},
  pages={311--320},
  year={2016},
  publisher={Elsevier}
}

@inproceedings{kay2015unequal,
  title={Unequal representation and gender stereotypes in image search results for occupations},
  author={Kay, Matthew and Matuszek, Cynthia and Munson, Sean A},
  booktitle={{ACM} conference on human factors in computing systems},
  pages={3819--3828},
  year={2015}
}

@article{bolukbasi2016man,
  title={Man is to computer programmer as woman is to homemaker? debiasing word embeddings},
  author={Bolukbasi, Tolga and Chang, Kai-Wei and Zou, James Y and Saligrama, Venkatesh and Kalai, Adam T},
  journal={Advances in Neural Information Processing Systems (NeurIPS)},
  volume={29},
  year={2016}
}

@article{baruah1996proportionate,
  title={Proportionate progress: A notion of fairness in resource allocation},
  author={Baruah, Sanjoy K and Cohen, Neil K and Plaxton, C Greg and Varvel, Donald A},
  journal={Algorithmica},
  volume={15},
  number={6},
  pages={600--625},
  year={1996},
  publisher={Springer}
}

@inproceedings{kliachkin2024fairness,
  title={Fairness in Ranking: Robustness through Randomization without the Protected Attribute},
  author={Kliachkin, Andrii and Psaroudaki, Eleni and Mare{\v{c}}ek, Jakub and Fotakis, Dimitris},
  booktitle={2024 IEEE 40th International Conference on Data Engineering Workshops (ICDEW)},
  pages={201--208},
  year={2024},
  organization={IEEE}
}

@inproceedings{dinitz2022fair,
  title={Fair disaster containment via graph-cut problems},
  author={Dinitz, Michael and Srinivasan, Aravind and Tsepenekas, Leonidas and Vullikanti, Anil},
  booktitle={International Conference on Artificial Intelligence and Statistics},
  pages={6321--6333},
  year={2022},
  organization={PMLR}
}

@inproceedings{feldman2011unified,
  title={A unified framework for approximating and clustering data},
  author={Feldman, Dan and Langberg, Michael},
  booktitle={Proceedings of the forty-third annual ACM symposium on Theory of computing},
  pages={569--578},
  year={2011}
}

@inproceedings{bachem2018one,
  title={One-shot coresets: The case of k-clustering},
  author={Bachem, Olivier and Lucic, Mario and Lattanzi, Silvio},
  booktitle={International conference on artificial intelligence and statistics},
  pages={784--792},
  year={2018},
  organization={PMLR}
}

@inproceedings{braverman2021coresets,
  title={Coresets for clustering in excluded-minor graphs and beyond},
  author={Braverman, Vladimir and Jiang, Shaofeng H-C and Krauthgamer, Robert and Wu, Xuan},
  booktitle={Proceedings of the 2021 ACM-SIAM Symposium on Discrete Algorithms (SODA)},
  pages={2679--2696},
  year={2021},
  organization={SIAM}
}

@inproceedings{chakraborty2021approximating,
  title={Approximating the median under the ulam metric},
  author={Chakraborty, Diptarka and Das, Debarati and Krauthgamer, Robert},
  booktitle={Proceedings of the 2021 ACM-SIAM Symposium on Discrete Algorithms (SODA)},
  pages={761--775},
  year={2021},
  organization={SIAM}
}

@inproceedings{chakraborty2023clustering,
  title={Clustering Permutations: New Techniques with Streaming Applications},
  author={Chakraborty, Diptarka and Das, Debarati and Krauthgamer, Robert},
  booktitle={14th Innovations in Theoretical Computer Science Conference (ITCS 2023)},
  year={2023},
  organization={Schloss-Dagstuhl-Leibniz Zentrum f{\"u}r Informatik}
}
